\pdfoutput=1
\documentclass[a4paper,UKenglish,autoref,numberwithinsect]{mylipics2019}

\usepackage{algorithm}
\usepackage{algorithmic}
\usepackage{mathtools}
\usepackage{latexsym}
\usepackage{stmaryrd}
\SetSymbolFont{stmry}{bold}{U}{stmry}{m}{n}
\usepackage{anyfontsize}
\usepackage{ifthen}
\usepackage{calc}
\usepackage{tikz}
\usepackage{xspace}
\usepackage[normalem]{ulem}
\usetikzlibrary{arrows}
\usetikzlibrary{shapes.geometric}
\usetikzlibrary{calc}
\usetikzlibrary{decorations.markings}
\usetikzlibrary{shapes.geometric}
\usetikzlibrary{calc}
\tikzset{current point is local=true}

\newcounter{claimcounterglobal}
\newcounter{claimcounter}
\renewenvironment{claim}[1][]{
  \renewcommand{\proof}{\smallskip\par\noindent\textit{Proof. }}
  \medbreak\par\noindent%
  \ifthenelse{\equal{#1}{}}{%
      \setcounter{claimcounter}{0}%
      \refstepcounter{claimcounterglobal}%
      \refstepcounter{claimcounter}%
      \textit{Claim~\arabic{claimcounter}.}
  }{%
    \ifthenelse{\equal{#1}{resume}}{%
      \refstepcounter{claimcounterglobal}%
      \refstepcounter{claimcounter}%
      \textit{Claim~\arabic{claimcounter}.}
    }{%
      \textit{Claim~#1.}
    }
  }
}{
  \par\medbreak
}

\newcommand{\case}[1]{\par\medskip\noindent\textit{Case #1: }}
\newenvironment{cs}{
  \begin{description}
    \renewcommand{\case}[1]{\item[\itshape\mdseries Case ##1:]}
  }{
  \end{description}
}
\newcommand{\uend}{\hfill$\lrcorner$}
 
\newcommand{\uenda}{\tag*{$\lrcorner$}} 

\newenvironment{eroman}{\enumerate[(i)]}{\endenumerate}

\renewcommand{\mathbf}[1]{\boldsymbol{#1}}
\renewcommand{\max}{\operatorname{max}}
\renewcommand{\min}{\operatorname{min}}

\renewcommand{\phi}{\varphi}
\newcommand{\bigmid}{\;\big|\;}
\newcommand{\Bigmid}{\;\Big|\;}
\newcommand{\ceil}[1]{\left\lceil#1\right\rceil}
\newcommand{\floor}[1]{\left\lfloor#1\right\rfloor}

\newcommand{\formel}[1]{\textsf{\upshape #1}}
\newcommand{\logic}[1]{\textsf{\upshape #1}}

\newcommand{\LL}{\logic L}
\newcommand{\LC}[1][]{\logic C\ifx\uchyph#1\uchyph\else^{#1}\fi}

\newcommand{\LW}[1][]{\logic C\ifx\uchyph#1\uchyph\else_{\textsf{\upshape w}}^{#1}\fi}

\newcommand{\FOL}{\logic{FO}}

\newcommand{\free}{\mathrm{free}}

\newcommand{\height}{\operatorname{ht}}

\newcommand{\ZZ}{\mathbb Z}

\newcommand{\dagle}{\trianglelefteq}

\renewcommand{\int}{\operatorname{int}}

\newcommand{\iso}[1]{\formel{iso}_{#1}}

\newcommand{\CB}{\mathcal B}
\newcommand{\CC}{\mathcal C}

\newcommand{\CE}{\mathcal E}

\newcommand{\CQ}{\mathcal Q}

\newcommand{\FB}{{\mathbf B}}
\newcommand{\FC}{{\mathbf C}}
\newcommand{\FD}{{\mathbf D}}

\newcommand{\FG}{{\mathbf G}}
\newcommand{\FH}{{\mathbf H}}

\newcommand{\FN}{{\mathbf N}}

\newcommand{\FQ}{{\mathbf Q}}
\newcommand{\FR}{{\mathbf R}}
\newcommand{\FS}{{\mathbf S}}

\newcommand{\FX}{{\mathbf X}}
\newcommand{\FY}{{\mathbf Y}}

\newcommand{\Fb}{{\mathbf b}}

\newcommand{\Fe}{{\mathbf e}}
\newcommand{\Ff}{{\mathbf f}}
\newcommand{\Fg}{{\mathbf g}}

\newcommand{\Fbd}{\mathbf{bd}}
\newcommand{\Fint}{\mathbf{int}}
\newcommand{\Fcl}{\mathbf{cl}}

\newcommand{\FCQ}{\mathbf{G(\CQ)}}

\newcommand{\hier}[1][]{\par\bigskip\hrule\mbox{}\\{\red HIER WEITER%
\ifthenelse{\equal{#1}{}}{}{\par\noindent#1}}\\\hrule\mbox{}\par\bigskip}

\newcommand{\dist}{\operatorname{dist}}

\newcommand{\eg}{\operatorname{eg}}
\newcommand{\og}{\operatorname{og}}
\newcommand{\nog}{\operatorname{ng}}

\newcommand{\Cut}{\operatorname{Cut}}
\newcommand{\at}{\operatorname{at}}
\newcommand{\art}{\operatorname{art}}

\newcommand{\Aut}{\operatorname{Aut}}

\renewcommand{\tilde}{\widetilde}

\renewcommand{\hat}{\widehat}

\newcommand{\llcurly}{\{\hspace{-3.5pt}\{}
\newcommand{\rrcurly}{\}\hspace{-3.5pt}\}}

\newcommand{\atp}{\operatorname{atp}}

\nolinenumbers

\title{A Linear Upper Bound on the Weisfeiler-Leman Dimension of Graphs of Bounded Genus}

\titlerunning{A Linear Upper Bound on the WL Dimension of Graphs of Bounded Genus}

\author{Martin Grohe}{RWTH Aachen University, Aachen, Germany}{grohe@informatik.rwth-aachen.de}{}{}
\author{Sandra Kiefer}{RWTH Aachen University, Aachen, Germany}{kiefer@informatik.rwth-aachen.de}{}{}

\authorrunning{M.\ Grohe and S.\ Kiefer}

\begin{document}

\maketitle

\begin{abstract}
  The Weisfeiler-Leman (WL) dimension of a graph is a measure for the
  inherent descriptive complexity of the graph. While originally
  derived from a combinatorial graph isomorphism test called the
  Weisfeiler-Leman algorithm, the WL dimension can also be
  characterised in terms of the number of variables that is required
  to describe the graph up to isomorphism in first-order logic with
  counting quantifiers.

  It is known that the WL dimension is upper-bounded for all graphs
  that exclude some fixed graph as a minor \cite{gro17}. However, the
  bounds that can be derived from this general result are
  astronomic. Only recently, it was proved that the WL dimension of
  planar graphs is at most $3$ \cite{kieponschwe17}.

  In this paper, we prove that the WL dimension of graphs embeddable
  in a surface of Euler genus $g$ is at most $4g+3$. For the WL dimension of graphs embeddable in an orientable surface of Euler genus $g$, our approach yields an upper bound of $2g+3$.
\end{abstract}

\section{Introduction}
The Weisfeiler-Leman (WL) algorithm is a simple combinatorial graph
isomorphism test. The 1-dimensional version of the algorithm, also
known as \emph{colour refinement} and \emph{naive vertex classification}, is known
since at least the mid 1960s, and it is widely used as a subroutine in
almost all practical graph isomorphism tools (see, for instance,
\cite{darliffsakmar04,junkas07,mck81,mckaypip14}), but also in machine
learning~(see, for instance, \cite{ahmkermlanat13,NIPS2017_6703,ICLR_6703,morritfey+19,sheschlee+11}). The 2-dimensional version can be traced
back to an article by Weisfeiler and Leman that appeared 50 years
ago~\cite{weilem68}. It is closely related to the algebraic theory of
coherent configurations. The generalisation to higher dimensions is
due to Babai (see~\cite{bab16,caifurimm92}), and again it plays an important
role as a subroutine in graph isomorphism algorithms, albeit more on
the theoretical side. Notably, the $(\log n)$-dimensional version is
used as a subroutine in Babai's quasipolynomial graph isomorphism
test~\cite{bab16}. 

The connection between the WL algorithm and logic was
made by Immerman and Lander~\cite{immlan90} and Cai, F\"urer, and
Immerman~\cite{caifurimm92}. They showed that two graphs are
distinguished by the $k$-dimensional WL algorithm if and only if they
can be distinguished in the logic $\LC[k+1]$, the $(k+1)$-variable
fragment of first-order logic which uses counting quantifiers of the form
$\exists^{\geq p}x$. 
The connection between the WL algorithm and logical
definability is at the core of some of the most interesting
developments in descriptive complexity theory (see, for example,
\cite{gro17,imm99,ott97}). Only recently, it was noted that the
WL algorithm (and thus the finite variable counting logic) has further
surprising characterisations. In a breakthrough paper, Atserias and
Maneva~\cite{atsman13} showed that the dimension $k$ of the
WL algorithm required to distinguish two graphs corresponds
to the level of the Sherali-Adams relaxation of the natural integer
linear program for graph isomorphism testing (also see
\cite{groott15,mal14}). This spawned a lot of work relating the WL
algorithm to semidefinite programming \cite{atsoch18,odowriwu+14} and
algebraic (Gr\"obner basis) approaches \cite{bergro15,gragropagpak18a}
to graph isomorphism testing. These results can also be phrased in
terms of propositional proof complexity. The latest facet of the theory is
a characterisation in terms of homomorphism counts from graphs of tree
width $k$ \cite{delgrorat18}. Various aspects of the WL
algorithm and its relation to logic have been studied in detail in
recent years~(see, for instance, \cite{arvkobratver15,arvkobratverb17,fur17,kieschwe16,kieschwesel15,krever15}).

Cai, Fürer, and Immerman~\cite{caifurimm92} proved that for every $k$
there are non-isomorphic 3-regular graphs $G_k,H_k$ of size $O(k)$
that cannot be distinguished by the $k$-dimensional WL algorithm. Thus,
as an isomorphism test, the $k$-dimensional WL algorithm is
incomplete. But, in view of the wide variety of seemingly unrelated
combinatorial, logical, and algebraic characterisations of the
algorithm, \emph{we are convinced that the structural information the
  algorithm is able to detect is of fundamental importance.} 

  The basic
parameter of the algorithm is the dimension, corresponding to the
number of variables in logical and the degree of polynomials in
algebraic characterisations. It yields a structural invariant called
the \emph{WL dimension} of a graph $G$ \cite{gro17}, defined to be the
least $k$ such that the $k$-dimensional WL algorithm distinguishes $G$
from every graph $H$ that is not isomorphic to $G$ (we say that $k$-WL
\emph{identifies} $G$), or equivalently, the least $k$ such that $G$
can be characterised up to isomorphism (or \emph{identified}) in the
logic $\LC[k+1]$. It is also convenient to define the \emph{WL
  dimension} of a class $\CC$ of graphs to be the maximum of the
WL dimensions of all graphs in $\CC$ if this maximum exists, or
$\infty$ otherwise. We see the WL dimension as a measure for the inherent
combinatorial or descriptive complexity of a graph or class of
graphs. We are mostly interested in the relation between the WL dimension
and other graph invariants.

Work in descriptive complexity shows that the WL dimension is bounded
for many natural graph classes, among them trees~\cite{immlan90},
graphs of bounded tree width~\cite{gromar99}, planar
graphs~\cite{gro98a}, graphs of bounded genus~\cite{gro00,gro12}, all
graph classes that exclude some fixed graph as a minor \cite{gro17},
interval graphs \cite{kobkuhlau+11,lau10}, and graphs of bounded rank
width \cite{groneu19}. However, most of these results do not give
explicit bounds on the WL dimension, and the bounds that can be
derived from the proofs are usually bad.
Only recently, the second author of this paper, jointly with
Ponomarenko and Schweitzer, established an almost tight bound for
planar graphs \cite{weilem68}: the WL dimension of planar graphs is at
most $3$, and there are planar graphs of WL dimension $2$. 

In this paper we establish bounds for graphs that can be embedded into
an arbitrary surface, for example, a torus or a projective plane. By the
classification theorem for surfaces (see
\cite[Theorem~3.1.3]{mohtho01}), up to homeomorphism (that is,
topological equivalence) all surfaces fall into only two countably
infinite families, the family $(\FS_k)_{k\ge0}$ of orientable surfaces
and the family $(\FN_\ell)_{\ell\ge 1}$ of non-orientable surfaces. For
example, the sphere $\FS_0$, the torus $\FS_1$, and the double torus
$\FS_2$ are the first three orientable surfaces, and the projective
plane $\FN_1$ and the Klein bottle $\FN_2$ are the first two
non-orientable surfaces. The \emph{Euler genus} $\eg(\FS)$ of a surface
$\FS$ is $2k$ if $\FS$ is homeomorphic to the orientable surface
$\FS_k$, and $\ell$ if $\FS$ is homeomorphic to the non-orientable surface
$\FN_\ell$. We define the \emph{Euler genus} of a graph $G$ to be the least
$g$ such that $G$ is embeddable (that is, can be drawn without edge
crossings) in a surface of Euler genus $g$ (see
Figure~\ref{fig:k5torus} for an example).

\begin{figure}
  \centering
\begin{tikzpicture}
  [
  line width=0.3mm,
  ]

  \path (0,0) node {\includegraphics[height=4cm]{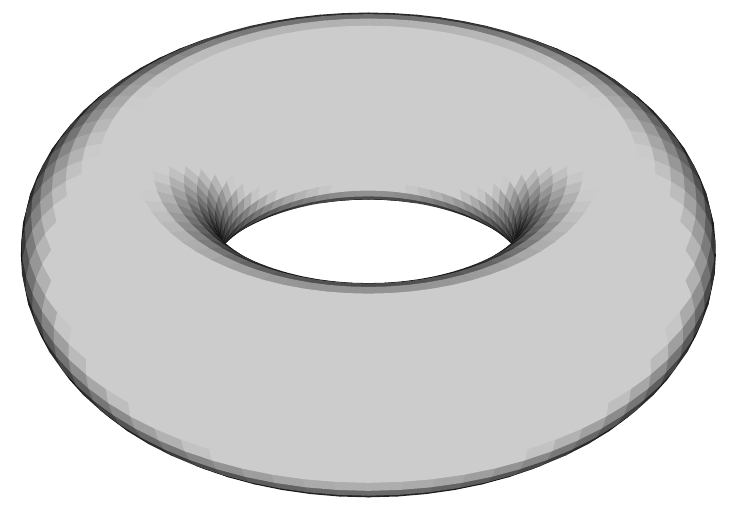}};

  \fill[red] (0.7,-0.5) circle (0.8mm) (0.8,-1.3) circle (0.8mm)
  (-0.8,-1.3) circle (0.8mm) (-0.7,-0.5) circle (0.8mm) (1.5,-0.6) circle (0.8mm); 

  \draw[red] (-2,0.45) arc (180:0:2cm and 1cm);
  \draw[red] (-2,0.45) arc (180:248:2cm and 1.01cm);

  \draw[red] (0.7,-0.5)  .. controls (1,-0.5) and (1.3,-0.55) ..  (1.5,-0.6);
  \draw[red] (0.8,-1.3)  .. controls (1,-1.3) and (1.4,-0.9) ..  (1.5,-0.6);
  \draw[red] (2,0.45) arc (360:316:2cm and 1.56cm);
  \draw[red] (-0.8,-1.3)  .. controls (0,-1.8) and (1.4,-1.8) ..  (1.5,-0.6);

  \draw[red] (0.7,-0.5) .. controls (0.3,-0.6) and (-0.3,-0.6) .. (-0.7,-0.5);
  \draw[red] (0.8,-1.3) .. controls (0.3,-1.25) and (-0.3,-0.85) .. (-0.7,-0.5);
  \draw[red] (0.8,-1.3) .. controls (0.3,-1.4) and (-0.3,-1.4) .. (-0.8,-1.3);
  \draw[red] (0.7,-0.5) .. controls (0.8,-0.7) and (0.85,-1.1) .. (0.8,-1.3);
  \draw[red] (-0.7,-0.5) .. controls (-0.8,-0.7) and (-0.85,-1.1)
  .. (-0.8,-1.3);
  \draw[red] (0.7,-0.5) .. controls (0.6,-0.3) and (0.4,-0.2)  .. (0.25,-0.2);
  \draw[red][dashed] (0.25,-0.2) .. controls (0.00,-0.3)  .. (-0.05,-1)
  .. controls (-0.1,-1.6) and (-0.2,-1.7) .. (-0.5,-1.8);
  \draw[red] (-0.5,-1.8) .. controls (-0.75,-1.65) and  (-0.77,-1.45)  .. (-0.8,-1.3);
\end{tikzpicture}
  \caption{Embedding of $K_5$ into the torus}
  \label{fig:k5torus}
\end{figure}
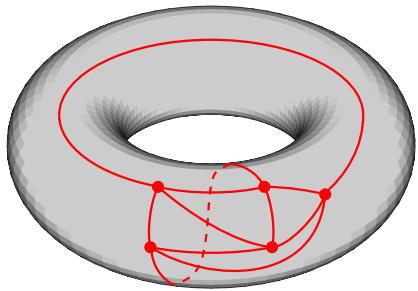

\begin{theorem}\label{thm:main}
  The WL dimension of a graph of Euler genus $g$ is at most $4g+3$.
\end{theorem}

For graphs embeddable in orientable surfaces, we can improve the bound further.

\begin{corollary}\label{cor:orient}
  The WL dimension of a graph embeddable in an orientable surface of
  Euler genus $g$ is at most $2g+3$.
\end{corollary}

As mentioned above, it was first proved in \cite{gro00} that the
WL dimension of graphs of bounded genus is bounded. A more detailed
proof of the same result can be found in the journal paper
\cite{gro12}. Neither of the two papers gives an explicit bound on the
WL dimension. The proof of \cite{gro12} only yields a quadratic bound
(in terms of the genus). It seems that the proof of \cite{gro00} gives
a linear bound, albeit with a large constant factor of at least $80$
(not all details are worked out there, so it is difficult to
determine the exact bound). The proof in both of these papers is based
on the fact that sufficiently large graphs of minimum degree at least
$3$ embedded in a surface will have a facial cycle of length at most $6$. The proof we give
here is completely different. It is based on the straightforward idea
of removing a non-contractible cycle to reduce the genus and then applying
induction. The problem with this idea is that we cannot define
non-contractible cycles, only families of such cycles that may
intersect in complicated patterns. Understanding these leads to
significant technical complications, but in the end enables us to
obtain a much better bound than the simpler proofs of
\cite{gro00,gro12}. Our proof is based on a simplified version of a
construction from \cite[Chapter~15]{gro17}, applied there to graphs ``almost embeddable'' in a surface.

\subsection*{Outline of the Paper} In Section $\ref{sec:prelim}$, we
introduce the conventions as well as some topological notions and
facts that we use throughout the paper. In
Section~\ref{subsec:wl-logic}, we introduce the WL dimension and relate it
to logic. In Section~\ref{sec:sps-patch-necklace}, we introduce the graph-theoretic machinery that we need in the proof of our main theorem. The
proof is outlined in Section \ref{sec:mainthm}. The detailed proof is
long and complicated, and we defer it to a technical appendix.

\section{Preliminaries}\label{sec:prelim}

We introduce the definitions and conventions regarding notation in this paper, which mostly follow \cite[Chapters~9 and~15]{gro17}. 

\subsection{Graphs} \label{subsec:graphs}
All graphs in this paper are finite, simple, and undirected. For a
graph $G$, we denote by $V(G)$ and $E(G)$ its set of vertices and
edges, respectively. We denote an edge between vertices $v$ and $w$ by
$uv$. Depending on the context, we sometimes view the edge set $E(G)$
as a subset of $\binom{V(G)}{2}$ and sometimes as an irreflexive
symmetric binary relation on $V(G)$; this should cause no confusion. The
\emph{order} of a graph $G$ is $|G|\coloneqq|V(G)|$, and we let
$\|G\|\coloneqq|E(G)|$.

For a set $V \subseteq V(G)$, we set $N^G(V) \coloneqq \{w
\mid w \in V(G)\setminus V, \exists v \in V \colon vw \in
E(G)\}$. Here, and in similar notations, we omit the superscript ${}^G$
if $G$ is clear from the context. 

For two graphs $G$ and $H$, we denote by $G \cup H$ the graph with vertex set $V(G) \cup V(H)$ and edge set $E(G) \cup E(H)$.
A graph $H$ is a \emph{subgraph} of $G$ (we write $H \subseteq G$) if $V(H) \subseteq V(G)$ and $E(H) \subseteq E(G)$. In this case we let $N(H)\coloneqq N\big(V(H)\big)$. We denote by $G[V] \coloneqq (V, E(G)
\cap \{uv\mid u,v \in V\})$ the \emph{subgraph of $G$
induced by $V$}. For a set $X$ (not necessarily a subset of $V(G)$) we
let $G \setminus X\coloneqq G[V(G) \setminus
X]$, and for a graph $H$, we let $G \setminus H \coloneqq G\setminus
V(H)$. 

For $k\ge 1$, the graph $G$ is \emph{$k$-connected} if $|G| > k$ and for
every $V \subseteq V(G)$ with $|V| < k$, the graph $G \setminus V$ is
connected. A \emph{$k$-separator} of $G$ is a set $S \subseteq V(G)$
of size $|S|=k$ such that there are vertices $u,v\in V(G)\setminus S$
which belong to the same connected component of $G$, but to different
connected components of $G\setminus S$.

Let $H\subseteq G$. For a connected component $A$ of $G\setminus H$, the
vertices in $N(A)\subseteq V(H)$ are \emph{vertices of
  attachment} of
$A$. An \emph{$H$-bridge} is a subgraph $B\subseteq
\big(V(G),E(G)\setminus E(H)\big)$ such that either $B=(\{u,v\},\{uv\})$ for
some edge $uv \in E(G)\setminus E(H)$ or $B$ is the union of a connected
component $A$ of $G\setminus H$ together with all its vertices of
attachment and all edges with at least one endvertex in $V(A)$. 
The \emph{vertices of attachment} of an $H$-bridge $B$ are the
vertices in $V(B)\cap V(H)$. We denote the set of vertices of
attachment of $B$ by $\at(B)$.

An \emph{arc-coloured graph}~$(G, \chi)$ is a graph~$G$ with a
function~$\chi\colon \big\{(u,u)\mid u\in V(G)\} \cup \{ (u,v) \mid
\{u,v\} \in E(G)\big\}\rightarrow \mathcal{C}$, where $\CC$ is some set of
colours. In an arc-coloured graph
we interpret~$\chi(u,u)$ as the vertex colour of~$u$ and
for~$uv \in E(G)$ we interpret~$\chi(u,v)$ as the colour of the arc
from~$u$ to~$v$. In particular it may be the case
that~$\chi(u,v) \neq \chi(v,u)$, that is, the two orientations of an
(undirected) edge $uv$ may receive different colours. A
\emph{vertex-coloured graph} is the special case of an arc-coloured
graph where all arcs receive the same colour, say, $1$, that is,
$\chi(u,v)=1$ for all $u\neq v$.
Whenever we refer to coloured graphs in this paper, we mean arc-coloured graphs. To simplify the notation, we usually do not mention the colouring
explicitly and just denote an arc-coloured graph by $G$, implicitly
assuming that the colouring is $\chi$. 

For a (possibly coloured) graph $G$ and a sequence of vertices $v_1, \dots, v_\ell$, we write $G_{v_1, \dots, v_\ell}$ to denote the graph resulting from individualising every vertex $v_i$, i.e., by assigning every $v_i$ for $i \in [\ell]$ a unique colour. When comparing two graphs with individualised vertices $G_{v_1, \dots, v_\ell}$ and $H_{v'_1, \dots, v'_\ell}$ we assume that for $i \in [\ell]$, the two vertices $v_i$ and $v'_i$ have the same colours.

We write $G \cong H$ to indicate that the graphs $G$ and $H$ are isomorphic via a colour-preserving isomorphism.
An \emph{automorphism} of $G$ is an isomorphism from $G$ onto itself. The set of automorphisms of $G$ equipped with concatenation forms a group, also denoted by $\Aut(G)$. For a vertex $v \in V(G)$, the \emph{orbit} of $v$ is the set $\{\pi(v) \mid \pi \in \Aut(G)\}$. A set $V \subseteq V(G)$ is called a \emph{block} of $\Aut(G)$ if for every $\pi \in \Aut(G)$ it holds that $V \cap \pi(V) \in \{\emptyset, V\}$, i.e., if every automorphism of $G$ maps $V$ onto itself or onto a set that is disjoint to $V$.

For a set $W \subseteq V(G)$, let $G/W$ be the graph obtained from $G$
by identifying all vertices in $W$ and eliminating loops and parallel
edges. We usually denote the vertex of $G/W$ representing the set $W$
by $w$. Formally, $G/W$ is the graph with vertex set $V(G/W) \coloneqq
V(G \setminus W) \cup \{w\}$ and edge set $E(G/W) \coloneqq E(G \setminus W) \cup \{vw \mid v \in V(G \setminus W), \exists w' \in W : vw' \in E(G)\}$.
Moreover, if $G$ has the colouring $\chi$ with range
$\mathcal{C}$, then $G/W$ has the colouring $\chi'$ where
$\chi'(w,w) = \emptyset$ and $\chi'(u,v) = \chi(u,v)$ and
 $\chi'(u,w) \coloneqq \llcurly \chi(u,w') \mid w' \in W
\rrcurly$ and $\chi'(w,v) \coloneqq \llcurly \chi'(w',v) \mid w' \in W
\rrcurly$ for all $u,v \in V(G \setminus W)$ and all $w \in W$. (We use
$\llcurly\ldots\rrcurly$ as notation for multisets.)
For a subgraph $A\subseteq G$, we let $G/A\coloneqq G/V(A)$, with the
convention of denoting the vertex of $G/A$ representing $V(A)$ by $a$.

\subsection{Topology}\label{subsec:topo}
In this section we review basic notions of surface topology and graph
embeddings. In our presentation and notation, we follow
\cite[Chapter~9]{gro17}. Many more details can be found there, in
\cite{mohtho01}, and in \cite[Appendix~B]{die10}.

We denote topological spaces like surfaces, curves, and embedded
graphs by bold-face letters. A \emph{simple curve} in a topological space is
a homeomorphic image of the real interval $[0,1]$, equipped with the usual
topology. Similarly, a \emph{simple closed curve} is a homeomorphic image of
the 1-sphere. A \emph{closed disk} is a
homeomorphic image of $\{x\in\mathbb R^2\mid \|x\|\le 1\}$ equipped
with the usual topology, and an \emph{open disk} is a
subspace of $\mathbb R^2$ that is homeomorphic to $\mathbb R^2$ (viewed as a
topological space).
A topological space $\FX$ is \emph{arcwise
  connected} if for any two points $x,y\in\FX$ there is a simple curve
with endpoints $x$ and $y$. For a subset $\FY\subseteq\FX$, we define the
\emph{boundary} of $\FY$ in $\FX$ to be the set $\Fbd_{\FX}(\FY)$ of
all points $x\in\FX$ such that every neighbourhood of $x$ has a
nonempty intersection with both $\FY$ and $\FX\setminus\FY$. The
\emph{interior} of $\FY$ is $\Fint_{\FX}(\FY)\coloneqq\FY\setminus
\Fbd_{\FX}(\FY)$, and the \emph{closure} of $\FY$ is $\Fcl_{\FX}(\FY)\coloneqq\FY\cup
\Fbd_{\FX}(\FY)$. We omit the subscript $\FX$ if the space, usually a
surface, is clear from the context.

A \emph{surface} is an arcwise
connected 2-manifold (intuitively, a space that looks like a disk in a
small neighbourhood of every point).\footnote{In this paper, we only consider surfaces without
boundary.} Recall from the introduction that up to
homeomorphism there are only two families $(\FS_g)_{g\ge0}$ and
$(\FN_g)_{g\ge 1}$ of surfaces. $\FS_0$ is the $2$-sphere, and for
$g\ge 1$, $\FS_g$ is the surface obtained from the $2$-sphere by
adding $g$ handles, and $\FN_g$ is the surface obtained
from the $2$-sphere by adding $g$ crosscaps. Intuitively, adding a
\emph{handle} to a surface means punching two holes into the surface and
gluing a cylinder to these holes. Adding a \emph{crosscap} means punching a hole
into the surface and gluing a Möbius strip to this hole. The \emph{Euler genus} $\eg(\FS)$ of a surface
$\FS$ is $2g$ if $\FS$ is homeomorphic to $\FS_g$ and $g$ if $\FS$ is homeomorphic to 
$\FN_g$.

Let $\Fg$ be a simple closed curve in a surface $\FS$. Then $\Fg$ is
\emph{contractible} if it is the boundary of a closed disk in $\FS$,
otherwise $g$ is \emph{non-contractible}. If $\Fg$ is non-contractible,
we can obtain one or two surfaces of strictly smaller Euler genus by
the following construction: we cut the surface along $\Fg$; what
remains is a surface with one or two holes in it. Then we glue a
disk onto these hole(s) and obtain one or two simpler surfaces. For a
more detailed description of this construction, see \cite[Appendix~B]{die10}.

Formally, an embedded graph in a surface $\FS$ is a pair
$G=\big(V(G),E(G)\big)$ where $V(G)\subseteq \FS$ is a finite set and $E(G)$
is a set of simple curves in $\FS$ such
that for all $\Fe\in E(G)$, both endpoints and no internal point of
$\Fe$ are in $V(G)$ and any two distinct $\Fe,\Fe'\in E(G)$ have at
most one endpoint and no internal points in common. $\FG$ denotes the
point set $V(G)\cup\bigcup_{\Fe\in E(G)}\Fe\subseteq\FS$. Sometimes,
we also regard $\FG$ as a topological (sub)space (of $\FS$). The
\emph{underlying graph} of an embedded graph $G$ is the graph with
vertex set $V(G)$ and edge set $\{\Fe\cap V(G)\mid \Fe\in E(G)\}$. We
usually blur the distinction between an embedded graph $G$ and its
underlying ``abstract'' graph.  The \emph{faces} of $G$ are the
arcwise connected components of the space $\FS\setminus\FG$. It is
easy to see that for every face $\Ff$ of $G$ there is a subgraph
$B\subseteq G$ such that the (topological) boundary $\Fbd(\Ff)$ of
$\Ff$ in $\FS$ is
precisely $\FB$. We call $B$ a \emph{facial subgraph} of $G$.

We say that an (abstract) graph $G$ is \emph{embeddable} into a
surface $\FS$ if it is isomorphic to (the underlying graph of) a graph
embedded in $\FS$.  The \emph{Euler genus} $\eg(G)$ of a graph $G$ is
the least $g$ such that $G$ is embeddable into a surface of Euler
genus $g$. It is useful to also
define the \emph{orientable genus} $\og(G)$ of a graph $G$ to be the smallest $g$
such that $G$ is embeddable into $\FS_g$ and the \emph{non-orientable
  genus} $\nog(G)$ of $G$ to be the smallest $g$
such that $G$ is embeddable into $\FN_g$. Then
$\eg(G)=\min\{2\og(G),\nog(G)\}$.

The graphs of Euler genus $0$ are precisely the \emph{planar graphs}
because a graph can be embedded into the 2-sphere $\FS_0$ if and
only if it can be embedded into the plane $\FR^2$. The class of all graphs of
Euler genus at most $g$ is denoted by $\CE_g$.

A
\emph{non-contractible cycle} in a graph $G$ embedded in $\FS$ is a
cycle $C\subseteq G$ such that $\FC$ is a non-contractible simple
closed curve in $\FS$.

\begin{fact}\label{fact:disks}
  Let $\FS$ be a surface, and let $\FD,\FD'\subseteq\FS$ be closed disks such
  that $\Fbd(\FD')\cap\FD$ is a simple curve. Then $\FD\cup\FD'$ is a
  closed disk.
\end{fact}

\begin{fact}[see Fact~9.1.14, \cite{gro17}]\label{fact:disk-or-nc}
  Let $G$ be a graph embedded in a surface $\FS$. Then either $G$ contains a non-contractible cycle or there
  is a closed disk $ \FD \subseteq \FS$ such that $\FG \subseteq \FD$.

  In the latter case, if $G$ is 2-connected, the disk $\FD$ can be chosen in such a way
  that there is a cycle $C\subseteq G$ such that $\Fbd(\FD)=\FC$.
\end{fact}

Let $\FS$ be a surface and let $G$ be a graph embedded in $\FS$. A set
$\FX \subseteq \FS$ is \emph{$G$-normal} if
$\FX \cap \FG \subseteq V(G)$. The \emph{representativity} $\rho(G)$
of $G$ is the maximum $r \in \mathbb{N}$ such that every $G$-normal
non-contractible simple closed curve $\Fg$ in $\FS$ intersects $\FG$ in
at least $r$ vertices. $G$ is \emph{polyhedrally embedded} in $\FS$ if
$G$ is 3-connected and $\rho(G)\ge 3$. Note that, particularly, every
3-connected plane graph is polyhedrally embedded in $\FS_0$.
Polyhedrally embedded graphs have several useful properties (see
\cite[Fact~9.1.17]{gro17}). In particular, all facial subgraphs of a
polyhedrally embedded graph are chordless and non-separating cycles
\cite{robvit90}. Conversely, for every graph embedded in a surface,
all contractible, chordless, and non-separating cycles are facial
subgraphs (see \cite[Lemma~9.1.15]{gro17}). (Here a cycle
$C\subseteq G$ is \emph{chordless} if it is an induced subgraph of
$G$, and it is \emph{non-separating} if $G\setminus V(C)$ is
connected.) This is a generalisation of the well-known
theorem that the facial subgraphs of a 3-connected plane graph are
precisely the chordless and non-separating cycles. It implies
Whitney's Theorem \cite{whi32} that
all plane embeddings of 3-connected planar graphs have the same facial
cycles and that, up to homeomorphism, a 3-connected planar graph has
a unique embedding into the sphere $\FS_0$.

\section{Finite Variable Logic with Counting}\label{sec:logic}

Here we give a detailed introduction into the logic
$\LC$, the extension of $\FOL$ by \emph{counting
  quantifiers} and its finite variable fragments, and we prove several
technical lemmas.

We interpret the logic $\LC$ over graphs, possibly 
coloured. In a logical context, we view a graph $G$ as a
relational structure whose vocabulary consists of a single binary
relation $E$. We view a coloured graph $(G,\chi)$ as a relational
structure whose vocabulary contains, in addition to the binary
relation symbol $E$, a binary relation symbol $R_c$ for every colour
$c$ in the range of $\chi$.
This relation symbol is interpreted by the
set of all pairs $(u,v)$ such that $\chi(u,v)=c$.

An occurrence of a variable $x$ is \emph{free} in a formula $\varphi$
if it is outside of all subformulae $\exists^{\geq p} x \psi$.
We often write $\varphi(x_1,\dots,x_\ell)$ to indicate that the free
variables of $\varphi$ are among $x_1,\dots,x_\ell$. (Not all of these
variables are required to appear in $\phi$.)
Then we also denote by $\varphi(y_1,\dots,y_\ell)$ the result of
substituting variables $y_1,\dots,y_\ell$ for the free occurrences of
$x_1,\dots,x_\ell$.

For a graph $G$ and
vertices $u_1,\ldots,u_\ell\in V(G)$, we write
$G\models\phi(u_1,\ldots,u_\ell)$ to denote that $G$ satisfies $\phi$ if
for all $i$ the variable $x_i$ is interpreted by $u_i$. Moreover, we
write $\phi[G,u_1,\ldots, u_i,x_{i+1},\ldots,x_\ell]$ to denote the
set of all $(\ell-i)$-tuples $(u_{i+1},\ldots,u_\ell)$ such that
$G\models\phi(u_1,\ldots,u_\ell)$.

For a logic $\LL$ and two graphs $G$ and $H$, we say \emph{$\LL$ distinguishes $G$ and $H$} if there is a formula $\varphi \in \LL$ such that $G \models \varphi$ and $H \not\models \varphi$. Similarly \emph{$\LL$ identifies $G$} if for every graph $H \not\cong G$, it holds that $\LL$ distinguishes $G$ and $H$.

\emph{Atomic formulae} in the language of (arc-coloured) graphs are of
the form $x_1 = x_2$, $E(x_1,x_2)$, or $R_c(x_1,x_2)$, where
$x_1,x_2$ are variables.
$\LC$-formulae are constructed from the atomic formulae using negation
$\neg \varphi$, disjunction $(\varphi \vee \psi)$, and counting
quantifiers $\exists^{\geq p} x\phi$ where $p \in \mathbb{N}_{\geq 1}$
and $x$ is a variable, and $\varphi$, $\psi$ are formulae. As
abbreviations, we also use conjunctions $(\phi\wedge\psi)$,
implications $(\phi\to\psi)$, and standard existential and
universal quantifiers $\exists x\phi$, $\forall x\phi$
($\forall x \varphi$ abbreviates
$\neg \exists^{\geq 1} x \neg \varphi$) as well as variants of the
counting quantifiers such as $\exists^{<p}x\phi$ and
$\exists^{=p}x\phi$. We also use $\textsc{true}$ for $\forall x (x =
x)$ and $\textsc{false}$ for $\neg \textsc{true}$ and $\phi \leftrightarrow \psi$ for $(\phi \rightarrow \psi) \wedge (\psi \rightarrow \phi)$. 

As a notational
convention throughout the paper, we shall use $x$, $y$, $z$ for
variables in first-order logic, whereas $u$, $v$, $w$ denote graph
vertices. 
The semantics of
the logic $\LC$ is defined in the usual way by inductively defining a
satisfaction relation $\models$ between pairs $(G,\nu)$ consisting of a
graph $G$ and an assignment $\nu$ of values in $V(G)$ to the variables
and formulae $\phi$. The only step going beyond standard first-order
logic is that of counting quantifiers: $(G,\nu) \models \exists^{\geq p} x \varphi$ if and
only if there are distinct vertices $v_1,\dots,v_p \in V(G)$ such that
$\big(G,\nu(v_i/x)\big) \models \varphi$ for each $v_i$, where
$\nu(v_i/x)$ is the assignment identical to $\nu$ except that
$\nu(v_i/x)(x) = v_i$.

Observe that $\LC$ is
only a syntactical extension of $\FOL$ with not more expressive power, because $\exists^{\ge
  p}x\phi(x)$ is equivalent to $\exists x_1\ldots\exists
x_p\Big(\bigwedge_{i\neq j}x_i\neq
x_j\wedge\bigwedge_i\phi(x_i)\Big)$. However, we are mainly interested
in the fragments $\LC[k]$ of $\LC$ consisting of all formulae with at
most $k$ variables. If $p>k$, then $\exists^{\ge p}x$ cannot be
expressed in the $k$-variable fragment of $\FOL$, thus $\LC[k]$ is strictly more expressive than the
$k$-variable fragment of $\FOL$. The logics $\LC[k]$ have played an important role in finite
model theory since the 1980s.  

We say a formula $\phi\in\LC$ has \emph{width} $k$ if every subformula of $\phi$ has at most $k$ free variables. 
We denote the $\LC$-formulae of width $k$ by $\LW[k]$. 

\begin{example}
  The following formula in $\LC[7]$ has width $3$:
  \[
    \exists x_1(E(x,x_1)\wedge\exists x_2(E(x_1,x_2)\wedge \exists
    x_3(E(x_2,x_3)\wedge\exists x_4(E(x_3,x_4)\wedge\exists x_5
    E(x_5,y))))).
  \]
  It is equivalent to the $\LC[3]$-formula
  \[
    \exists z(E(x,z)\wedge\exists x(E(z,x)\wedge \exists
    z(E(x,z)\wedge\exists x(E(z,x)\wedge\exists z
    E(z,y))))).
  \]
\end{example}
 
We will use the following well-known characterisation of $\LC[k]$.

\begin{lemma}\label{lem:width}
  Every $\LC$-formula of width $k$ is equivalent to a $\LC[k]$-formula.
\end{lemma}

We omit the straightforward proof. We note that to translate a $\LC$-formula
of width $k$ into a $\LC[k]$-formula, we only have to rename bound variables.
Also note that every $\LC[k]$-formula has width $k$.

\begin{example}\label{exa:dist}
  For every $k\ge0$ we define a $\LW[3]$-formula $\formel{dist}_{\le k}$
  such that for every graph $G$ and all vertices  $u,u'\in V(G)$ it
  holds that $G \models \formel{dist}_{\le k}(u,u')$ if and only if
  $u$ and $u'$ have distance at most $k$ in $G$. We let 
  \[\formel{dist}_{\le k}(x,x') \coloneqq
  \begin{cases}
    x = x' & \text{if } k = 0\\ 
     \exists y_k \big( E(x,y_k) \wedge \formel{dist}_{\le k-1}(y_k,x')\big) & \text{otherwise.} 
  \end{cases}
  \]
  Note that for $k\ge 1$, the $\LW[3]$-formula
  $\formel{dist}_{=k}(x,x')\coloneqq\formel{dist}_{\le k}(x,x')\wedge\neg
  \formel{dist}_{\le k-1}(x,x')$ states that $x$ and $x'$ have distance
  exactly $k$. Moreover, in every graph of order at most $n$ the  $\LW[3]$-sentence
  $\formel{conn}_n\coloneqq\forall x\forall x'\formel{dist}_{\le n-1}(x,x')$
  states that the graph is connected.
  \uend
\end{example}

The following lemma bounds the number of variables needed for avoiding
definable subsets.

\begin{lemma}\label{lem:avoiding-path}
    Let $\phi(x_1,\ldots,x_k,y)\in\LW[\ell]$. Then there is a formula
    $\formel{comp}_\phi(x_1,\ldots,x_k,y,y') \in \LW[\max\{k+3,\ell\}]$ such
    that for all graphs $G$ of order $|G|\le n$ and all
    $u_1,\ldots,u_k,v,v'\in V(G)$,
    \[
     G \models
      \formel{comp}_\phi(u_1,\ldots,u_k,v,v')\iff\parbox[t]{7cm}{$v$
        and $v'$ belong to the same connected component of
        $G\setminus\phi[G,u_1,\ldots,u_k,y]$.}
    \]
  \end{lemma}

  \begin{proof}
    Without loss of generality, we assume that in all formulae of the form $\exists^{\ge p}z\chi$ that we consider, the variable $z$ occurs free in $\chi$. We let $\psi(x_1,\ldots,x_k,y,y')$ be the formula obtained from
    the formula $\formel{dist}_{\le n-1}(y,y') $ of
    Example~\ref{exa:dist} by replacing each subformula
    $\exists^{\ge p}z\chi$ by
    $\exists^{\ge p}z(\neg\phi(x_1,\ldots,x_k,z)\wedge\chi)$. Then, letting $U \coloneqq \phi[G,u_1,\ldots,u_k,y]$, for
    all $v,v'\in V(G)\setminus U$ we have
    $G\models\psi[u_1,\ldots,u_k,v,v']$ if and only if $v$ and $v'$
    belong to the same connected component of $G\setminus U$. Note
    that $\psi(x_1,\ldots,x_k,y,y')\in \LW[\max\{k+3,\ell\}]$, because
    the formula $\chi\in\LW[3]$ has at most two free variables besides $z$.

    Now
    $\formel{comp}_\phi(x_1,\ldots,x_k,y,y')\coloneqq
    \neg\phi(x_1,\ldots,x_k,y)\wedge \neg\phi(x_1,\ldots,x_k,y')\wedge
    \psi(x_1,\ldots,x_k,y,y')$.
  \end{proof}

\begin{lemma}\label{lem:rel2comp}
  Let $n\ge 1$, $\ell \ge 3$, and $1 \le k \le \ell$. Then for every $\LW[\ell]$-formula
  $\psi(x_1,\ldots,x_k)$ there is a $\LW[\ell]$-formula
  $\hat\psi(x_1,\ldots,x_k)$ such that for every graph $G$ of
  order $|G|\le n$, every
  connected component $A$ of $G$, and all $u_1,\ldots,u_k\in
  V(A)$, it holds that
  \[
   G\models\hat\psi(u_1,\ldots,u_k)\iff A\models
   \psi(u_1,\ldots,u_k).
 \]
\end{lemma}

\begin{proof}
  We construct $\hat\psi$ by induction on $\psi$. If $\psi$ is atomic,
  then we simply let $\hat\psi\coloneqq\psi$. If $\psi=\neg\phi$ we
  let $\hat\psi\coloneqq\neg\hat\phi$, and if $\psi=\phi_1\vee\phi_2$
  we let $\hat\psi\coloneqq\hat\phi_1\vee\hat\phi_2$. The only
  interesting case is that
  $\psi(x_1,\ldots,x_k)=\exists^{\ge p}
  y\phi(x_1,\ldots,x_k,y)$. Note that the variable $y$ may be among
  $x_1,\ldots,x_k$. If this is the case,
  $\phi(x_1,\ldots,x_k,y)$ is the same formula as
  $\phi(x_1,\ldots,x_k)$. Without loss of generality we may assume
  that there is a $j\le k$ such that $y\neq x_j$. This is obvious
  if $k\ge 2$. If $k=1$, we can rename the bound variable $y$ and choose $j=1$. We let
$\hat\psi(x_1,\ldots,x_k)=\exists^{\ge p}
  y\big(\formel{dist}_{\le n-1}(x_j,y)\wedge\hat\psi(x_1,\ldots,x_k,y)\big)$,
  where $\formel{dist}_{\le n-1}$ is the $\LW[3]$-formula defined in Example~\ref{exa:dist}.
\end{proof}

Recall that
the notation $\psi(x_1,\ldots,x_k)$ merely says that the free variables
of the formula $\psi$ are \emph{among} $x_1,\ldots,x_k$; not all of
these variables actually have to appear. Thus we can also apply the
lemma to sentences $\psi$ and obtain the following corollary.

\begin{corollary}\label{cor:rel2comp}
    Let $n\ge 1$ and $\ell\ge 3$. Then for every $\LW[\ell]$-sentence
  $\phi$ there is a $\LW[\ell]$-formula
  $\hat\phi(x)$ such that for every graph $G$ of
  order $|G|\le n$ and every $u\in V(G)$ we have 
  $G\models\hat\phi(u)$ if and only if $A\models\phi$ for the connected component $A$ of
  $u$ in $G$.
\end{corollary}

\begin{corollary}\label{cor:id-comp}
  Let $\ell\ge 3$, and let $G$ be a graph such that every connected
  component of $G$ is identified by a $\LW[\ell]$-sentence. Then $G$ is
  identified by a $\LW[\ell]$-sentence.
\end{corollary}

Observe that the corollary fails for $\ell=2$. An example is the graph
$G$ that is the disjoint union of two triangles.

\begin{lemma}\label{lem:rel2comp2}
    Let $k \geq 0$, $n \geq 1$, $\ell \geq 3$ and let $\psi\in\LW[\ell]$ and $\phi(x_1,\ldots,x_k,y)\in\LW[m]$. Then there is a formula
    $\tilde\psi(x_1,\ldots,x_k,y) \in
    \LW[\max\{k+\ell,m\}]$ such that for all graphs $G$ of order
    $|G|\le n$ and all $u_1,\ldots,u_k,v\in V(G)$ the following
    holds. Let $U \coloneqq \phi[G,u_1,\ldots,u_k,y] $, and let $A_v$ be the
    connected component of $v$ in $G\setminus U$ (assuming
    $v\not\in U$). Then
    \[
     G \models\tilde\psi(u_1,\ldots,u_k,v)\iff v\not\in
      U\text{ and }A_v \models\psi.
    \]
  \end{lemma}

  \begin{proof}
    Again, without loss of generality, we assume that in all formulae of the form $\exists^{\ge p}z\chi$ that we consider, the variable $z$ occurs free in $\chi$. We apply Corollary~\ref{cor:rel2comp} to $\psi$ and obtain a
    $\LW[\ell]$-formula $\hat\psi(y)$ such that for every graph $H$ of order at most $n$
    and every $v \in V(H)$ we have $H\models\hat\psi(v)$ if and only if
    $A_v\models\psi$, where $A_v$ is the connected component of $v$ in
    $H$. In particular, this holds for the graph $H\coloneqq
    G\setminus U$.

    Without loss of generality we may assume that the variables
    $x_1,\ldots,x_k$ do not appear in $\hat\psi(y)$. We let
    $\hat\psi'(x_1,\ldots,x_k,y)$ be the  $\LW[\max\{k+\ell,m\}]$-formula obtained
    from $\hat\psi(y)$ by replacing each subformula $\exists^{\ge
      p}z\chi$ with $\exists ^{\ge
      p}z(\neg\phi(x_1,\ldots,x_k,z)\wedge\chi)$. Then for all $v\in
    V(H)$ we have $G\models\hat\psi'(v)\iff H\models\hat\psi(v)$. We let
    \[
      \tilde\psi(x_1,\ldots,x_k,y)\coloneqq
      \neg\phi(x_1,\ldots,x_k,y)\wedge\hat\psi'(x_1,\ldots,x_k,y).
      \qedhere
    \]
  \end{proof}

We need one more technical lemma which will be applied in one case of
the proof of our main theorem in Section~\ref{subsec:necklace}. The reason
we put it here is that we do not want to interrupt the flow of the
main argument later. The reader may safely skip the lemma on first
reading the paper and get back to it later.

For the purposes of the lemma, we need a way to prevent some free variables from counting towards the width of a formula.
We shall use the symbol $\circ$ as a special placeholder that can be substituted
for the free occurrences of variables with the effect that this placeholder does not count as a variable for the width.
For example, for $\psi(x,z) \coloneqq \exists y \big(E(x,y) \land E(y,z)\big)
\in \LW[3]$, we have $\psi(\circ,z) \in \LW[2]$ and $\psi(\circ,\circ)
\in \LW[1]$. Recall that for a graph $G$ and a subgraph $A\subseteq
G$, by $G/A$ we denote the graph obtained from $G$ by identifying all
vertices of $A$ and that $a$ is the vertex of $G/A$ corresponding to $A$.

\begin{lemma}\label{lem:factor}
Let $0\le\ell<m<k$ and $\xi(x_1,\ldots,x_\ell,y), \psi(x_1,\ldots,x_m,z) \in
\LW[k]$ such that $\psi(\circ,\ldots,\circ,x_{\ell+1},\ldots,x_m,\circ) \in
\LW[k-\ell]$. Then there is a formula $\phi(x_1,\ldots,x_m)
\in \LW[k]$ such that the following holds.

Let $G$ be a graph and let $A\subseteq G$. Suppose $u_1,\ldots,u_m\in
V(G)\setminus V(A)$ such that $V(A)=\xi[G,u_1,\ldots,u_\ell,y]$. Then 
\[
G\models\phi(u_1,\ldots,u_m)\iff G/A\models\psi(u_1,\ldots,u_m,a).
\]
\end{lemma}

\begin{proof}
    Without loss of generality, we assume that every bound variable in $\psi$ does not occur free in $\psi$ or $\xi$ and is not bound by a second quantifier in $\psi$.

    We let $\phi \coloneqq \psi^*$, where we define the transformation $*$ inductively to eliminate the variable $z$ as follows.

    For atoms $\alpha$ that do not mention $z$, we let $\alpha^* \coloneqq \alpha$.
    Atoms with $z$ are treated as follows, where $x$ denotes a variable distinct from $z$.
    For equality atoms, we define $(z = z)^* \coloneqq \textsc{true}$ and $(x = z)^* \coloneqq (z = x)^* \coloneqq \textsc{false}$.
    For atoms with predicate symbol $E$, we let $E(z,z)^* \coloneqq \textsc{false}$, and $E(x,z)^* \coloneqq \exists z \big(\xi(x_1,\ldots,x_\ell,z) \land E(x,z)\big)$, and $E(z,x)^*$ analogous to $E(x,z)$.
    For atoms with predicate symbol $R_C$, where the colour $C$ is a multiset with $r$ distinct elements $c_1,\ldots,c_r$ of multiplicities $p_1,\ldots,p_r$, we define $R_C(z,z)^* \coloneqq \textsc{false}$, and $R_C(x,z)^* \coloneqq \bigwedge_{j=1}^r \exists^{= p_j} z \big(\xi(x_1,\ldots,x_\ell,z) \land R_{c_j}(x,z)\big)$, and $R_C(z,x)^*$ analogous to $R_C(x,z)$.

    Inductively, we define $(\neg \chi)^* \coloneqq \neg \chi^*$ and $(\chi_1 \lor \chi_2)^* \coloneqq (\chi_1^* \lor \chi_2^*)$.
    For the case $\exists^{\geq p} x \chi(x_1,\ldots,x_n,x,z)$ for $p\ge 2$, we define
    \begin{equation}
      \begin{array}{r@{\,}l}
      (\exists^{\geq p} x \chi)^* \coloneqq {} & \Big(\chi(x_1,\ldots,x_n,z,z)^* \land
                        \exists^{\geq p-1} x \big(\neg
                        \xi(x_1,\ldots,x_\ell,x) \land
                        \chi(x_1,\ldots,x_n,x,z)^*\big)\Big)\\
      &{}\lor \exists^{\geq p} x \big(\neg
              \xi(x_1,\ldots,x_\ell,x) \land
              \chi(x_1,\ldots,x_n,x,z)^*\big). \label{eq:4}
      \end{array}
    \end{equation}
    Note that the formula $\chi(x_1,\ldots,x_n,z,z)^*$ is obtained by first
    substituting $z$ for $x$ in $\chi$ and then applying $*$
    to the resulting formula to eliminate $z$.
    The case $p = 1$ is dealt with analogously.

    To prove the correctness of the construction, we need to show that the free
    variables of $\psi^*$ are among $\{x_1,\ldots,x_m\}$ and $\psi^* \in \LW[k]$,
    and that $\psi^*$ has the correct meaning. 

    First, observe that a straightforward induction obtains that for every formula $\chi$,
    \begin{equation}
      \label{eq:2}
      \free(\chi^*)\subseteq \big(\free(\chi)\setminus\{z\}\big)\cup\{x_1,\ldots,x_\ell\},
    \end{equation}
    where $\free(\chi)$ denotes the free variables of $\chi$.
    Thus, $\free(\psi^*) \subseteq \{x_1,\ldots,x_m\}$.

    Second, observe that the condition
    $\psi(\circ,\ldots,\circ,x_{\ell+1},\ldots,x_m,\circ) \in
    \LW[k-\ell]$ expresses that no subformula of $\psi$ (including
    $\psi$ itself) has more than $k-\ell$ free variables that are not
    contained in the set $\{x_1,\ldots,x_\ell,z\}$. So we can assume
    that all subformulae of $\psi$ satisfy this
    condition.

    Now we are ready to prove $\psi^* \in \LW[k]$ by
    induction on $\psi$.  For the base steps, note that
    $E(x,z)^*, E(z,x)^*, R_C(x,z)^*, R_C(z,x) \in \LW[\ell+2]$ and $\ell+2\le k$;
    the other base cases are trivial.

    For the inductive step, the case $\neg\chi$ is trivial. 
    For the case $\chi_1\vee\chi_2$ we exploit observations
    \eqref{eq:2} and that $\psi$ has at most $k-\ell$ free
    variables not in $\{x_1,\ldots,x_\ell,z\}$. The case
    $\exists^{\geq p} x \chi(x_1,\ldots,x_n,x,z)$ follows immediately by induction, since we have $\chi(x_1,\ldots,x_n,z,z)^*\in\LW[k]$
    and $\chi(x_1,\ldots,x_n,x,z)^*\in\LW[k]$. 

    Finally, we show the following statement for every formula $\chi(x_1,\dots,x_n,z)$, where $n \geq \ell$ and every bound variable in $\chi$ does not occur bound in $\chi$ or $\xi$ and is not bound by a second quantifier in $\chi$: for every graph $G$, every subgraph
    $A\subseteq G$, and all $u_{1},\ldots,u_n\in V(G)\setminus V(A)$
    such that $V(A)=\xi[G,u_1,\ldots,u_\ell,y]$ we have
    \begin{equation}
      \label{eq:3}
      G\models\chi^*(u_1,\ldots,u_n)\iff G/A\models\chi(u_1,\ldots,u_n,a).
    \end{equation}
    The proof is by induction on $\chi$.
    This statement in particular applies to $\psi(x_1,\ldots,x_m,z)$ and thus completes the proof of the lemma.

    The base step for atomic formulae follows from the fact that $u_i \neq a$ for every $1 \leq i \leq n$ and the definition of $G/A$ and its colouring.

    In the inductive step, the negation and disjunction cases are
    trivial. Now consider the case $\exists^{\geq p} x \chi(x_1,\ldots,x_{n},x,z)$. Recall
    the definition in \eqref{eq:4}. To understand the
    following argument, it is important to know exactly which
    variables occur free in $\big(\exists^{\geq p} x
    \chi(x_1,\ldots,x_n,x,z)\big)^*$ and its constituent
    formulae. The formula $\chi^1 \coloneqq \chi(x_1,\ldots,x_n,x,z)^*$ has free
    variables among $x_1,\ldots,x_n,x$; we write
    $\chi^1(x_1,\ldots,x_n,x)$ to make this explicit.  The formula $\chi^2 \coloneqq \chi(x_1,\ldots,x_n,z,z)^*$ has free
    variables among $x_1,\ldots,x_n$; we write
    $\chi^2(x_1,\ldots,x_n)$. The formula $(\exists^{\geq p} x \chi)^*$ has free
    variables among $x_1,\ldots,x_n$; we write $(\exists^{\geq p} x \chi)^*(x_1,\ldots,x_n)$.

Let $G$
    be a graph, $A\subseteq G$, and all $u_1,\ldots,u_n\in V(G)\setminus V(A)$
    such that $V(A)=\xi[G,u_1,\ldots,u_\ell,y]$. By the induction
    hypothesis, for all $u \in V(G)\setminus V(A)$ we have 
    \begin{align}
      \label{eq:5}
      G\models \chi^1(u_1,\ldots,u_n,u)&\iff
                                           G/A\models\chi(u_1,\ldots,u_n,u,a)\\
      \intertext{and}
      \label{eq:6}
      G\models \chi^2(u_1,\ldots,u_n)&\iff
                                           G/A\models\chi(u_1,\ldots,u_n,a,a).
    \end{align}
    To prove the forward
    direction of \eqref{eq:3}, suppose that
    $G\models(\exists^{\geq p} x \chi)^*(u_1,\ldots,u_n)$. 
    \begin{cs}
      \case1 $G\models\chi^2(u_1,\ldots,u_n)$ and there are pairwise
      distinct $u^1 ,\ldots,u^{p-1}\in V(G)$ such that for
      all $j$, 
      $G\not\models\xi(u_1,\ldots,u_\ell,u^j)$ and
      $G\models\chi^1(u_1,\ldots,u_n,u^j)$.

      Then it holds that 
      $u^j\neq a$ by the assumption that
      $V(A)=\xi[G,u_1,\ldots,u_\ell,y]$.
      Furthermore, $G/A\models\chi(u_1,\ldots,u_n,a,a)$ by \eqref{eq:6} and
      $G/A\models\chi(u_1,\ldots,u_n,u^j,a)$ by \eqref{eq:5}.
      Thus $a,u^1,\ldots,u^{p-1}$ witness that
      $G/A\models\exists^{\geq p} x \chi(u_1,\ldots,u_n,x,a)$.

     \case2 There are pairwise
      distinct $u^1 ,\ldots,u^p\in V(G)$ such that for
      all $j$, 
      $G\not\models\xi(u_1,\ldots,u_\ell,u^j)$ and
      $G\models\chi^2(u_1,\ldots,u_n,u^j)$.

      Then it holds that 
      $u^j\neq a$ by the assumption that
      $V(A)=\xi[G,u_1,\ldots,u_\ell,y]$.
      Furthermore, $G/A\models\chi(u_1,\ldots,u_n,u^j,a)$ by \eqref{eq:5}.
      Thus $u^1,\ldots,u^{p}$ witness that
      $G/A\models\exists^{\geq p} x \chi(u_1,\ldots,u_n,x,a)$.
    \end{cs}
    The backward direction of \eqref{eq:3} is proved by reverting the
    same argument.
\end{proof}

\section{The WL Dimension}\label{subsec:wl-logic}
We start by reviewing the \emph{$k$-dimensional WL
  algorithm} (for short: \emph{$k$-WL}) for $k\ge 1$. 

The
\emph{atomic type} $\atp(G,\bar u)$ of a $k$-tuple
$\bar u=(u_1,\ldots,u_k)$ of vertices of a (possibly coloured) graph $G$ is the set of all atomic facts satisfied by these
vertices. The exact encoding is not important for us, the relevant
property is that tuples $\bar
u=(u_1,\ldots,u_k)$ and $\bar v=(v_1,\ldots,v_k)$ of vertices of
graphs $G,H$, respectively, have the same atomic type if and only if the mapping $u_i\mapsto v_i$ is an isomorphism from
the induced subgraph $G[\{u_1,\ldots,u_k\}]$ to the induced subgraph
$H[\{v_1,\ldots,v_k\}]$. 

Now $k$-WL is the algorithm that, given a graph
$G$, computes the following
sequence of ``colourings'' $C_i^k$ of $V(G)^k$ for $i\ge0$ until it
returns $C_\infty^k \coloneqq C_i^k$ for the smallest $i$ such
that for all $\bar u,\bar v$ it holds that $C_i^k(\bar u)=C_i^k(\bar v)\iff
C_{i+1}^k(\bar u)=C_{i+1}^k(\bar v)$. The initial colouring $C_0^k$ assigns to
each tuple its atomic type: $C_0^k(\bar u) \coloneqq \atp(G,\bar u)$. 
In the $(i+1)$-st \emph{refinement round}, the colouring $C_{i+1}^k$ is defined
by
$
C_{i+1}^k(\bar u) \coloneqq \big(C_i^k(\bar u),M_{i}(\bar u)\big),
$
where, for $\bar u=(u_1,\ldots,u_k)$, $M_i(\bar u)$ is the multiset
\begin{align*}
&\llcurly\big(\atp(G,(u_1,\ldots,u_k,v)),C_i^k(u_1,\ldots,u_{k-1},v),\\
  &\hspace{3.8cm}C_i^k(u_1,\ldots,u_{k-2},v,u_k),\ldots,C_i^k(v,u_2,\ldots,u_k)\big)\mid
v\in V\rrcurly
\end{align*}
We say that $k$-WL \emph{distinguishes} two graphs $G$, $H$ if there is
some colour $c$ in the range of $C_\infty^k$ such that the number of
tuples $\bar u\in V(G)^k$ with $C_\infty^k(\bar u)=c$ is different
from the number of
tuples $\bar v\in V(H)^k$ with $C_\infty^k(\bar v)=c$. We say that
$k$-WL \emph{identifies} $G$ if it distinguishes $G$ from all graphs
$H$ not isomorphic to $G$. The \emph{WL dimension} of $G$ is the least
$k$ such that $k$-WL identifies $G$.

\begin{definition}[see Definition 12, \cite{kieponschwe17}]
 Let~$\mathcal{H}$ be a set of graphs. We say that the~$k$-dimensional WL algorithm \emph{determines orbits} in~$\mathcal{H}$ if for all coloured graphs~$(G, \lambda)$ and all coloured graphs $(G',\lambda')$ (with colourings $\lambda$ and $\lambda'$) and all vertices~$s \in V(G)$ and~$s' \in V(G')$ the following holds: there exists an isomorphism from~$(G,\lambda)$ to~$(G',\lambda')$ mapping~$s$ to~$s'$ if and only if~$C^k_{\infty}(s) = C^k_{\infty}(s')$.
\end{definition}

The following proposition is a useful correspondence between identification and determination of orbits in a graph.

\begin{proposition}\label{prop:corrdetorbits}
  Let $k \geq 1$ be a natural number and let $G$ be a coloured graph. Suppose $k$-WL identifies all vertex-coloured versions of $G$. Then $(k+1)$-WL determines orbits on $G$.
\end{proposition}

\begin{proof}
  Let $\chi^k$ denote the stable colouring computed by $k$-WL. Let $G$ be a graph. Suppose there are a graph $H$ and vertices $v \in V(G)$, $v' \in V(H)$ such that $\chi^{k+1}(v) = \chi^{k+1}(v')$ holds.
  Then we can individualise $v$ in $G$ and $v'$ in $H$ and apply $k$-WL to these coloured graphs $G_v$ and $H_{v'}$. Since $\chi^{k+1}(v) = \chi^{k+1}(v')$, we have that $\{\!\!\{\chi^{k+1}(v,w_1, \dots, w_k) \mid (w_1, \dots, w_k) \in V^k(G) \}\!\!\} = \{\!\!\{\chi^{k+1}(v',w'_1, \dots, w'_k) \mid (w'_1, \dots, w'_k) \in V^k(H)\} \!\!\}$. Thus, the graphs $G_v$ and $H_{v'}$ obtain isomorphic colourings under $k$-WL. By assumption, this implies $G_v \cong H_{v'}$, which is equivalent to the existence of an isomorphism from $G$ to $H$ mapping $v$ to $v'$.
\end{proof}

For the following lemma, we assume that the reader is familiar with
graph minors. For those who are not, we remark that for every $g\ge 0$
the class $\CE_g$ of all graphs of Euler genus at most $g$ is closed
under taking minors. We will only apply the lemma to these classes.

For a class $\CC$ of (uncoloured) graphs, we let $\CC^*$ be the class of all coloured graphs with underlying graph in $\CC$.

\begin{lemma}[\cite{kieponschwe17}]\label{lem:3-connected}
Let $\CC$ be a graph class that is closed under taking
minors. Suppose $k$-WL identifies all 3-connected graphs
in $\CC^*$. Then  $(k+1)$-WL identifies all graphs in $\CC^*$.
\end{lemma}

\begin{proof}
  By \cite[Theorem~13]{kieponschwe17}, $(k+1)$-WL identifies all graphs in $\CC^*$ if $(k+1)$-WL determines orbits on all 3-connected graphs in
  $\CC^*$. Thus, the statement follows from Proposition \ref{prop:corrdetorbits}.
\end{proof}

In this paper, we reason about the WL dimension in terms of logic, using the following correspondence.

\begin{theorem}[\cite{caifurimm92,immlan90}]\label{thm:immlan}
  Let $k \geq 1$. Let $G$ and $H$ be graphs, possibly coloured,
  and $\bar{u} \coloneqq (u_1,\ldots,u_k)\in V(G)^k$ and
  $\bar{v} \coloneqq (v_1,\ldots,v_k)\in V(H)^k$. Then the following are equivalent:
    \begin{enumerate}
      \item $C_\infty^k(\bar u)=C_\infty^k(\bar v)$;
      \item
        $G\models\phi(u_1,\ldots,u_k)\iff
        H\models\phi(v_1,\ldots,v_k)$ for all $\LC[k+1]$-formulae
        $\phi(x_1,\ldots,x_k)$.
    \end{enumerate}
\end{theorem} 

Recall that we say a graph $G$ is \emph{identified} by the logic $\LC^{k}$
if there is a sentence $\iso G \in \LC^{k}$ such that for all graphs
$H$ we have
$H\models \iso G$ if and
only if $H$ is isomorphic to $G$.

\begin{corollary}\label{cor:correspondence}
  A graph has WL dimension $k$ if and only if it is identified by $\LC[k+1]$.
\end{corollary}

The WL dimension of a planar graph is at most $3$
\cite{kieponschwe17}. Using the previous corollary, we can re-phrase
this as follows.

\begin{theorem}[see \cite{kieponschwe17}]\label{thm:planar-dimension}
  For every colored planar graph $G$ there is a
  $\LC[4]$-sentence $\iso G$ that identifies $G$.
\end{theorem}



In the following sections, we use these formulae characterising certain parts of a decomposition of $G$ in order to obtain a bound on the number of variables we need to identify the entire graph.

\section{Shortest Path Systems, Patches and
  Necklaces}\label{sec:sps-patch-necklace}

Here we introduce the graph-theoretic machinery necessary
to prove our main theorem. Essentially, the definitions and results of
this section are from \cite[Chapter~15]{gro17}. In fact, things are
simpler here because \cite[Chapter~15]{gro17} deals with graphs
\emph{almost} embedded in a surface, whereas we only need to
consider surface graphs. Sometimes, we need to change the definitions in
order to improve the resulting bounds on the WL dimension
later. Notably, our definition of \emph{necklaces} is different from the
one in \cite{gro17}. This also requires an adaptation of the proof
that reducing necklaces exist.

\begin{definition}
  Let $G$ be a graph and $u,u'\in V(G)$. A \emph{shortest path system
    (sps) from $u$ to $u'$} is a family $\CQ$ of shortest paths in $G$ from
  $u$ to $u'$ such that every shortest path from $u$ to $u'$ in the
  subgraph $\bigcup_{Q\in\CQ}Q$ is contained in ${\CQ}$.

  We let
    $V(\mathcal Q)\coloneqq\bigcup_{Q\in{\CQ}}V(Q)$ and $E(\mathcal Q)\coloneqq\bigcup_{Q\in{\CQ}}E(Q)$ and
    $G(\CQ)\coloneqq\big(V({\CQ}),E({\CQ})\big)=\bigcup_{Q\in\CQ}Q$. We call
    ${\CQ}$ \emph{trivial} if
    $|V({\CQ})|\le 2$, that is, if $G(Q)$ consists of a single vertex or a single edge.

The \emph{height} $\height^\CQ(v)$ of $v\in V(\CQ)$ is the distance from $u$ to $v$. 
  The vertices in $\bigcap_{Q
  \in \CQ} V(Q)$ are the \emph{articulation vertices} of $\CQ$. An
articulation vertex $v$ is \emph{proper} if $v \neq u$ and $v \neq
u'$. We denote the set of all articulation vertices of $\CQ$ by
$\art(\CQ)$.  

For all $u$, $u' \in V(G)$ such that there is a path from $u$ to $u'$
in $G$, the \emph{canonical sps from $u$ to $u'$ in $G$} is the set
$\CQ^G(u,u')$ of all shortest paths from $u$ to $u'$ in $G$.
\end{definition}

For a path $Q$ and vertices $u, v \in V(Q)$, we denote by $uQv$ the segment of $Q$ from $u$ to $v$.
With every sps $\CQ$ from $u$ to $u'$ we can associate a partial order
$\dagle^{\CQ}$ on $V(\CQ)$ by letting $v\dagle^{\CQ}w$ if $v$ appears
  before $w$ on some path $Q\in \CQ$. For $v\dagle^{\CQ} w$, we
  define the \emph{segment} $\CQ[v,w]$ to be the set of segments
  $vQw$ from $v$ to $w$ of all paths $Q\in\CQ$ that contain both $v$
  and $w$. Observe that $\CQ[v,w]$ is an sps from $v$ to $w$.

  \begin{lemma}[\cite{gro17}, Lemma~15.2.3]\label{lem:sps2}
  Let $\CQ$ be an sps. Then $\CQ$ is non-trivial and has no
  proper articulation vertices if and only if the graph $G(\CQ)$ is 2-connected.
\end{lemma}

\begin{lemma}[\cite{gro17}, Lemma~15.2.4]\label{lem:sps2a}
  Let $\CQ$ be a non-trivial sps that has no proper
  articulation vertices. Then there are internally disjoint paths $Q,Q'\in\CQ$.
\end{lemma}

While shortest paths systems are defined with respect to abstract
graphs, the following notions are defined with respect to embedded
graphs. For the rest of the section, we make the following assumption.

\begin{assumption}\label{ass:polyhedral}
  $G$ is a graph polyhedrally embedded in a surface $\FS$ of Euler
  genus $g\ge 1$.
\end{assumption}

\begin{definition}\label{def:patch}
  A \emph{patch} in $G$ is an sps $\CQ$ in $G$ such that:
  \begin{eroman}
  \item $\CQ$ has no proper articulation vertices.
  \item There is a closed disk
    $\FD\subseteq\FS$ such that $\FCQ\subseteq\FD$.\uend
  \end{eroman}
\end{definition}

 Fact~\ref{fact:disk-or-nc} and Lemma~\ref{lem:sps2} imply that if $\CQ$ is a non-trivial patch then there is a unique disk
$\FD(\CQ)$ such that $\FCQ\subseteq\FD$ and $\Fbd\big(\FD(\CQ)\big)=\FC(\CQ)$
for a cycle $C(\CQ)\subseteq G(\CQ)$. Furthermore, $C(\CQ)=Q\cup Q'$ for two paths $Q,Q'\in\CQ$.

\begin{definition}\label{def:simplifying}
  A subgraph $H\subseteq G$
  is \emph{simplifying} if every
    connected component of $G\setminus H$ belongs to
    $\CE_{g-1}$.

  A patch $\CQ$ is \emph{simplifying} if the graph $G(\CQ)$ is simplifying.\uend
\end{definition}

\begin{lemma}[\cite{gro17}, Corollary 15.3.5]\label{fact:astar}
  For a non-simplifying subgraph $H\subseteq G$, there is at most one
  connected component $A^*$ of $G\setminus H$ with
  $A^* \notin \CE_{g-1}$, and all other connected components are planar.
\end{lemma}

It turns out that non-simplifying patches form the basic building blocks of our theory.
Let
$\CQ$ be a non-trivial non-simplifying path. Let $A^*$ be the unique
connected of $G\setminus V(\CQ)$ that is not planar (the existence and
uniqueness of $A^*$ follow from Lemma~\ref{fact:astar}). Let $G/A^*$
be the graph obtained from $G$ by contracting the subgraph $A^*$ to a
single vertex $a^*$. By
\cite[Corollary~15.4.5]{gro17}, $G/A^*$ is a 3-connected planar graph.
Figure \ref{fig:astar} displays a schematic view of a patch $\CQ$ with
some attached (planar) connected components as well as the non-planar
component $A^*$, the disk $\FD(\CQ)$, and the
boundary cycle $\FC(\CQ)$. 

\begin{figure}
  \centering
  \includegraphics[width=0.45\textwidth]{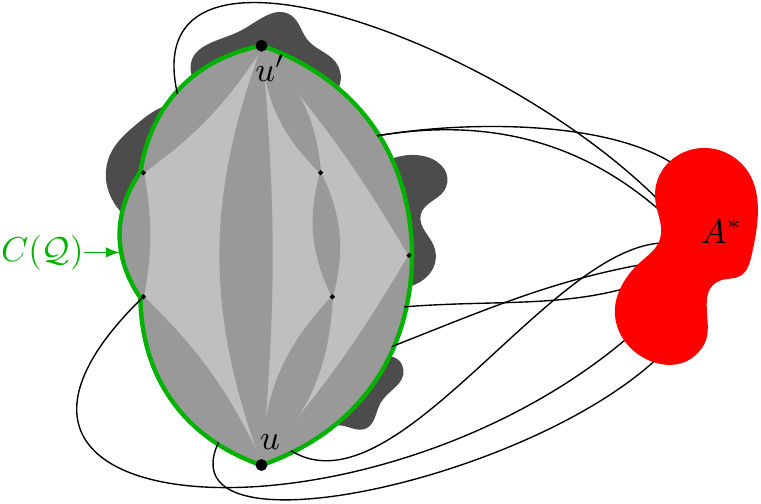}
  \hfill
  \includegraphics[width=0.45\textwidth]{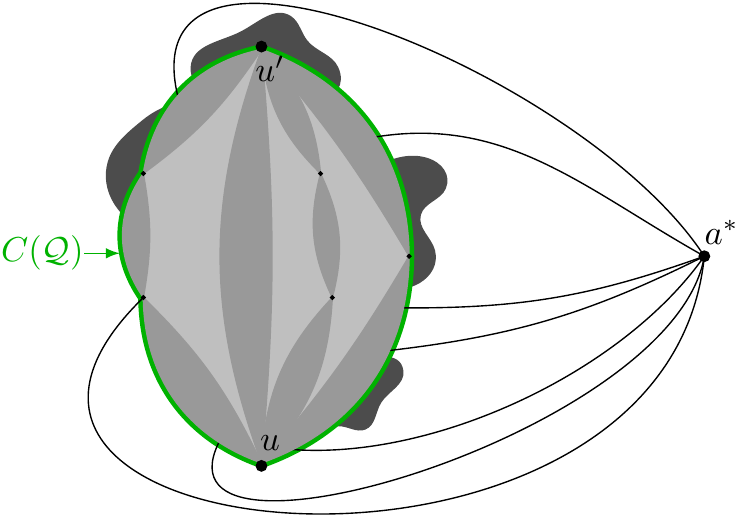}
  \caption{Left: A patch $\CQ$ with non-planar component $A^*$ and
    cycle $C(\CQ)$. The curve $\FC$ is the boundary of $\FD$. Right: the (planar)
    factor graph $\CQ/A^*$.}\label{fig:astar}
\end{figure}

We
define the \emph{internal graph} of a non-trivial patch $\CQ$ to be the graph
$I\coloneqq I(\CQ)$ with vertex set $V(I)\coloneqq V(G)\cap\FD(\CQ)$ and edge set
$E(I)\coloneqq\{\Fe\in E(G)\mid \Fe\subseteq\FD(\CQ)\}$. Note that
$C(\CQ)\subseteq I$. Formally, the definitions of the graphs $C(\CQ)$ and $I(\CQ)$ do not only depend on the abstract graph $G$ and the sps $\CQ$, but on the
embedding of $G$ in $\FS$. However, it can be proved that actually the
graphs are invariant under embeddings.

\begin{lemma}[\cite{gro17}]\label{lem:ns-abstract}
  Let $\CQ$ be a non-simplifying patch in $G$. Let $G'$ be a graph
  embedded in a surface $\FS'$ of Euler genus $g$ such that $G$ and
  $G'$ are isomorphic (as abstract graphs), and let $f$ be an
  isomorphism from $G$ to $G'$. Then $\CQ'\coloneqq f(\CQ')$ is a
  non-simplifying patch in $G'$, and it holds that $f\big(C(\CQ)\big)=C(\CQ')$
  and $f\big(I(\CQ)\big)=I(\CQ')$.
\end{lemma}

This follows from \cite[Lemma~15.4.10]{gro17}. Intuitively, the reason
this holds is that the 3-connected planar graph $G/A^*$ has a unique
embedding (see Section~\ref{subsec:topo}).

\begin{corollary}\label{cor:Q-invariance}
  Let $u,u'\in V(G)$ and $\CQ\coloneqq\CQ^G(u,u')$ such that $\CQ$ is
  a non-trivial non-simplifying patch. Let $f$ be an automorphism of $G$
  such that $f(u)=u$ and $f(u')=u'$. Then $f\big(C(\CQ)\big)=C(\CQ)$
  and $f\big(I(\CQ)\big)=I(\CQ)$.
\end{corollary}

We remark that the analogue of Corollary~\ref{cor:Q-invariance} for
simplifying patches does not hold. (Figure~\ref{fig:simp-patch} in Section \ref{subsec:simplifying} shows an example.) The analysis of simplifying patches is
much more involved, and we defer it to Section~\ref{subsec:simplifying}.

The final objects we define in this section are \emph{necklaces}.

\begin{definition}\label{def:necklace}
 A \emph{necklace} in $G$ is a tuple $\CB \coloneqq (u^0,\CQ^0,u^1,\mathcal{Q}^1,u^2,\CQ^2)$,
  where $u^0,u^1,u^2\in V(G)$ and $\CQ^i \coloneqq \CQ^G(u^{i},u^{i+1})$
  (indices taken modulo 3) is the canonical sps
  from $u^{i}$ to $u^{i+1}$, such that the following
  conditions are satisfied:
  \begin{enumerate}
  \item $u^0,u^1,u^2$ are pairwise distinct.\label{li:b1}
  \item
    $V(\CQ^{i})\cap V(\CQ^{i+1})=\{u^{i+1}\}$ (indices
    modulo 3).\label{li:b2}
  \item There is a disk $\FD_i\subseteq \FS$ such that $\mathbf{G(\CQ^i)}\subseteq \FD_i$.\label{li:b3} \uend
  \end{enumerate}
\end{definition}

For a necklace $\CB \coloneqq (u^0, \CQ^0, u^1, \CQ^1, u^2, \CQ^2)$ we
write $V(\CB)$ for the set $\bigcup_{i=0}^2 \bigcup_{Q
  \in \CQ^i} V(Q)$ and $E(\CB)$ for $\bigcup_{i=0}^2
\bigcup_{Q \in \CQ^i} E(Q)$, and we let
$G(\CB)\coloneqq \big(V(\CB),E(\CB)\big)$. Moreover, we define the set of
\emph{articulation vertices} of $\CB$ to be $\art(\CB) \coloneqq \{u_0,u_1,u_2\} \cup \bigcup_{i=0} ^2 \art(\CQ^i)$.

\begin{definition}\label{def:necklace-nc}
  A necklace $\CB \coloneqq (u^0,\CQ^1,u^1,\CQ^2,u^2,\CQ^3)$ is \emph{reducing} if
  there are paths $Q^i\in\CQ^i$ such that $B\coloneqq Q^1\cup Q^2\cup Q^3$ is a non-contractible cycle. \uend
\end{definition}

Figure \ref{fig:necklace} shows a reducing necklace on a torus with articulation
vertices $u^0$, $u^0_1$, $u^1$, $u^2$. 

\begin{figure}
  \centering
  \includegraphics[width=0.5\textwidth]{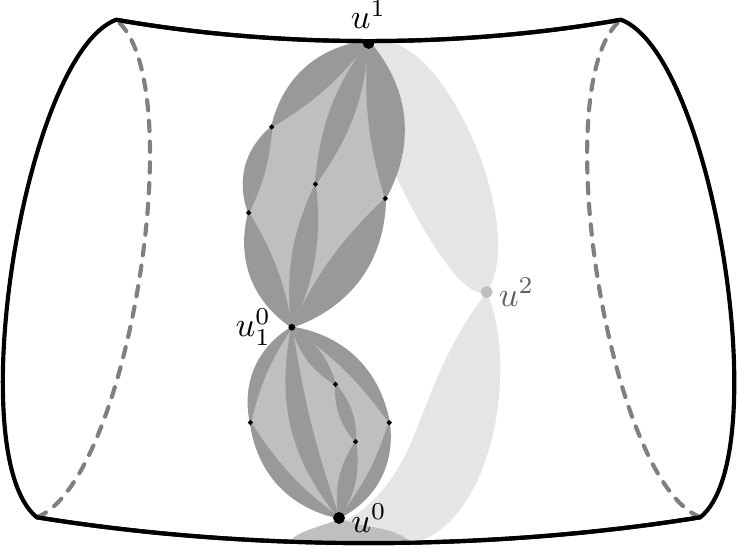}
  \caption{A reducing necklace on a torus section.}\label{fig:necklace}
\end{figure}

\begin{lemma}[Necklace Lemma]\label{lem:exnecklace}
  $G$ has a reducing necklace.
\end{lemma}

Essentially, this is \cite[Lemma~15.5.8]{gro17}, with the necklaces corresponding to the \emph{belts} there. But since apart from a renaming, we have also
slightly changed the content of the definition of a necklace/belt, the proof also needs to be
adapted. For the proof of the Necklace Lemma, we need one well-known fact
and more complicated lemma from \cite{gro17}.

\begin{lemma}\label{lem:3curves}
  Let $\FS$ be a surface, and let $\Fg_1,\Fg_2,\Fg_3\subseteq\FS$ be simple
  curves with the same endpoints and mutually disjoint interiors. Then
   $\Fg_1\cup\Fg_2$, $\Fg_2\cup\Fg_3$, and  $\Fg_1\cup\Fg_3$ are
   simple closed curves, and if $\Fg_1\cup\Fg_2$
   and $\Fg_2\cup\Fg_3$ are contractible, then $\Fg_1\cup\Fg_3$ is
   contractible as well.
\end{lemma}

For a proof, see \cite[Proposition~4.3.1]{mohtho01}.

\begin{lemma}[\cite{gro17}, Lemma~15.5.9]\label{lem:exnecklace1}
  Let $\CQ$ be an sps in $G$ such that there is no disk
  $\FD\subseteq\FS$ with ${\FCQ}\subseteq\FD$, but for every
  proper segment $\CQ'$ of $\CQ$ there is a disk $\FD'\subseteq\FS$
  with $\mathbf{G(\CQ')}\subseteq\FD'$. Then there are internally
  disjoint paths $Q,Q'\in\CQ$ such that $\FQ\cup \FQ'$ is a
  non-contractible simple closed curve in $\FS$.
\end{lemma}

With these tools at hand, we can now prove the existence of a reducing necklace in $G$.

\begin{proof}[Proof of Lemma~\ref{lem:exnecklace}]
  By Fact~\ref{fact:disk-or-nc}, there is a cycle $C\subseteq G$ such
  that $\FC$ is a non-contractible simple closed curve in $\FS$. We
  choose such a cycle $C$ of minimum length.  We let $u^0,u^1,u^2\in V(C)$ such that
  \begin{equation}
    \label{eq:1}
    \floor{\frac{\|C\|}{3}}\le\dist^C(u^i,u^j)\le\ceil{\frac{\|C\|}{3}}.
  \end{equation}
  We let $\CQ^i \coloneqq \CQ^G(u^{i},u^{i+1})$ and
  $\CB \coloneqq (u^0,\CQ^0,u^1,\CQ^1,u^2,\CQ^2)$.  Here and
  throughout the proof, indices $i,j$ are taken from $\ZZ_3$ with
  addition modulo 3.

  It follows from \eqref{eq:1} that the $u^i$ are mutually
  distinct. Thus $\CB$ satisfies Condition \ref{li:b1} of Definition~\ref{def:necklace}.
  
  Let $Q^i$ be the segment of $C$ from $u^{i}$ to $u^{i+1}$  that does not contain $u_{i+2}$. Then
  $C=Q^0\cup Q^1\cup Q^2$.

  \begin{claim}\label{cl:rednecklace0}
    Let $P\subseteq G$ be a shortest path with distinct endvertices
    $u,u'\in V(C)$ and no internal vertices in $C$. Let $Q,Q'$ be the
    two segments of $C$ from $u$ to $u'$. Then $P\cup Q$ or $P\cup Q'$ is a
    non-contractible cycle. Furthermore, if $P\cup Q$ is a
    non-contractible cycle, then $\|Q'\|=\|P\|$, and if $P\cup Q'$ is a
    non-contractible cycle, then $\|Q\|=\|P\|$.

    \proof
    Clearly, since $P$ has no internal vertices in $C$, both $P\cup Q$ and $P\cup Q'$ are cycles. By Lemma~\ref{lem:3curves}, we know that $P\cup Q$ or $P\cup Q'$ is
    non-contractible. Say, $P\cup Q$ is. Since $C$ is a shortest
    non-contractible cycle, we have $\|P\cup Q\|\ge\|C\|=\|Q\cup
    Q'\|$. Thus $\|Q'\|\le\|P\|$, and since $P$ is a shortest path,
    equality holds.
    \uend
  \end{claim}

    \begin{claim}[resume]\label{cl:rednecklace5}
      Let $Q\in\CQ^i$. Then $V(Q)\cap V(C)\subseteq V(Q^i)$.

      \proof By symmetry, it suffices to prove the claim for $i=0$.
      Suppose for contradiction that there is a path $Q\in\CQ^0$ with
      $V(Q)\cap V(C)\not\subseteq V(Q^0)$. Fix $Q$ to be such a path with
      the maximum number of edges in $E(C)$. 

      Note that $Q \neq Q^0$. Thus, $Q\not\subseteq C$, because the
      only paths in $C$ from $u^0$ to $u^1$ are $Q^0$ and
      $Q^1 \cup Q^2$. However, $Q\neq Q^0$ and
      $\|Q\|\le\|Q^0\|<\|Q^1\|+\|Q^2\| = \|Q^1 \cup Q^2\|$, which
      implies $Q\neq Q^1\cup Q^2$.

      Recall that for a path $Q$ and vertices $v,w \in V(Q)$, we
      denote by $vQw$ the segment of $Q$ from $v$ to $w$. Throughout
      this proof, for a second path $Q'$ with $v, w \in V(Q')$, we
      denote by $uQvQ'w$ the walk from $u$ to $w$ obtained by
      following $Q$ from $u$ to $v$ and then following $Q'$ from $v$
      to $w$. (We also use this notation style to compose multiple segments of paths.)

      Let $u \neq u'$ be vertices in $V(C)$ and let $P=uQu'$ be a segment of $Q$ with endvertices $u,u'\in V(C)$
      and all internal vertices and edges of $P$ not in $C$. Then $P$
      is a shortest path from $u$ to $u'$. 
      Let $R$, $R'$ be the two segments of
      $C$ with endpoints $u,u'$. Then by Lemma \ref{lem:3curves}, one of $R \cup P$ and $R' \cup P$ must be a non-contractible cycle, say $R'\cup P$. Then $\|P\|=\|R\|$. 

      \begin{cs}
        \case1
        $Q$ has an empty intersection with the interior of $R$.\\
        Then $u^0QuRu'Qu^1$ is a path from $u^0$ to $u^1$ that has
        the same length as $Q$, but more edges in $E(C)$. This
        contradicts the maximality of $Q$.  
        \case2
        The segment $u^0Qu$ contains an internal vertex that lies in $R$.\\
        Let $v$ be the first vertex of $Q$ in $R$. Then $v$ appears
        on $Q$ before $u$. Let $w$ be the last vertex of $Q$ in $R$
      (possibly, $w=u'$). Then $u^0QvRwQu^1$ is a path from $u^0$ to
      $u^1$ that is shorter than $Q$, which contradicts $Q$ being a
      shortest path.
        \case3
        The segment $u'Qu^1$ contains an internal vertex that lies in $R$.\\
        Let $w$ be the last vertex of $Q$ in $R$. Then $w$ appears
        on $Q$ after $u'$. Let $v$ be the first vertex of $Q$ in $R$
      (possibly, $v=u$). Then $u^0QvRwQu^1$ is a path from $u^0$ to
      $u^1$ that is shorter than $Q$, which again contradicts $Q$ being a
      shortest path.
    \end{cs}
Thus, the segment $P$ does not exist, which implies the claim.\uend
    \end{claim}

    \begin{claim}[resume]\label{cl:rednecklace6}
      $Q^i\in\CQ^i$.

      \proof
      Again, by symmetry it suffices to prove the claim for $i=0$.
      Let $Q\in\CQ^0$ with
      a maximum number of edges in $E(C)$. Arguing with similar techniques as in
      the proof of Claim~\ref{cl:rednecklace5}, we can show that
      $Q=Q^0$.
      \uend
    \end{claim}

    \begin{claim}[resume]\label{cl:rednecklace7}
      Let $Q\in\CQ^i$ and $Q'\in\CQ^{i+1}$. Then $V(Q)\cap V(Q')=\{u^{i+1}\}$.

      \proof
      As usual, we assume $i=0$.

      Suppose that $v_0=u^0,v_1,v_2,\ldots,v_m=u^1$ are the vertices
      in $Q\cap Q^0$ in the order in which they appear on $Q^0$. Then
      the vertices appear on $Q$ in the same order, because by Claim \ref{cl:rednecklace6} both $Q^0$
      and $Q$ are shortest paths. For $i=0,\ldots,m-1$, let $P_i$ be
      the segment of $Q$ from $v_i$ to $v_{i+1}$. Note that no
      internal vertex of $P_i$ is in $V(C)$. Thus either
      $P_i=v_iQ^0v_{i+1}$ is a single edge or
      $P_i\cup v_iQ^0v_{i+1}$ is a cycle. Since this cycle is shorter than
      $C$, it must be contractible. Let $C_0\coloneqq C$, and for
      $1\le i\le m-1$, let $C_{i}$ be the cycle obtained from
      $C_{i-1}$ by replacing the segment $v_{i}Q^0v_{i+1}$ with
    $P_i$. It follows from Claim~\ref{cl:rednecklace0} applied to the
    cycle $C_{i-1}$ and the path $P_i$ that each
    $C_i$ is non-contractible. In particular, 
      \[
      C' \coloneqq C_{m-1}=u^0Qu^1Q^1u^{2}Q^2u^0
      \]
      is a non-contractible cycle of the same length as $C$.

      Thus, $C'$ is also a shortest non-contractible cycle through
      $u^0,u^1,u^2$ and $\dist^{C'}(u^i,u^j)=\dist^{C}(u^i,u^j)$. This
      means that we can apply all previous claims to $C'$ instead of
      $C$. In particular, it follows from Claim~\ref{cl:rednecklace5}
      applied to $C'$ and $Q'$ that $V(Q')\cap V(Q)=\{u^1\}$.
      \uend
    \end{claim}

  \begin{claim}[resume]\label{cl:rednecklace4}
    Let $Q,Q'\subseteq G$ be paths from $u^{i}$ to
    $u^{i+1}$ such that $\|Q\|,\|Q'\|\le\|Q^i\|$. Then there is a disk
    $\FD\subseteq\FS$ such that $\FQ\cup\FQ'\subseteq\Fint(\FD)$.
    
    \proof We have $\|Q\cup Q'\|\le\|Q\|+\|Q'\|\le
    2\|Q^i\|<\|C\|$. Thus the graph $Q\cup Q'$ does not contain a
    non-contractible cycle, and by Fact~\ref{fact:disk-or-nc} there
    is a closed disk $\FD'\subseteq\FS$ such that
    $\FQ,\FQ'\subseteq\FD'$. We can slightly increase $\FD'$ to get a disk $\FD$ such that
    $\FQ,\FQ'\subseteq\Fint(\FD)$.  \uend
  \end{claim}

    \begin{claim}[resume]\label{cl:rednecklace8}
      There is a disk $\FD\subseteq\FS$ such that
      $\mathbf{G(\CQ^i)}\subseteq \FD$. 

      \proof Suppose for contradiction that there is no such disk. Let
      $\CQ$ be a segment of $\CQ^i$ such that there is no disk
      $\FD\subseteq\FS$ with ${\FCQ}\subseteq\FD$, but for every
      proper segment $\CQ'$ of $\CQ$ there is a disk
      $\FD'\subseteq\FS$ with $\mathbf{G(\CQ')}\subseteq\FD'$. Then by
      Lemma~\ref{lem:exnecklace1}, there are paths $Q,Q'\in\CQ$ such that
      $\FQ\cup \FQ'$ is a non-contractible simple closed curve in
      $\FS$. This contradicts Claim~\ref{cl:rednecklace4}.  \uend
    \end{claim}

    We have already noted that $\CB$ satisfies Condition \ref{li:b1}
    of Definition \ref{def:necklace}. It follows from
    Claim~\ref{cl:rednecklace7} that it satisfies Condition
    \ref{li:b2} and Claim~\ref{cl:rednecklace8} implies that it
    satisfies Condition \ref{li:b3} as well. Thus $\CB$ is a
    necklace. Claim~\ref{cl:rednecklace6} implies that this necklace
    is reducing.
\end{proof}

\section{Upper Bound on the WL Dimension}\label{sec:mainthm}
Finally, in this section we give the proof of our main theorem
(Theorem \ref{thm:main}). By the correspondence between $k$-WL and the logic $\LC^{k+1}$ as stated in Corollary~\ref{cor:correspondence}, we need to prove that every graph
of Euler genus at most $g$ can be identified by a
$\LC[4g+4]$-sentence. The proof is by induction on $g$. The base step
$g=0$ is Theorem~\ref{thm:planar-dimension}.

For the inductive step, we make the following assumption.

\begin{assumption}\label{ass:induction}
  Assume $g\ge 1$ and there is a natural number $s\ge 4$ such that every
  graph in $\CE_{g-1}$ is identified by a
$\LW[s]$-sentence. 
\end{assumption}

Our goal is to prove the following lemma (under Assumption
\ref{ass:induction}). The lemma implies Theorem~\ref{thm:main} by induction.

\begin{lemma}[Inductive Step]\label{lem:inductive-step}
  For every coloured graph $G$ in $\CE_g$ there is a sentence $\iso G \in \LW[s+4]$ that identifies $G$. 
\end{lemma}

The proof will proceed in a sequence of lemmas. Eventually, it will
diverge into two main cases, to be dealt with in
Subsections~\ref{subsec:necklace} and \ref{subsec:simplifying}. We first show that we can assume without loss of generality that $\rho(G) \geq 3$.  

\begin{lemma}\label{lem:representativity}
  Let $G$ be a coloured graph that has an embedding of
  representativity at most $2$ into a surface of Euler genus at most
  $g$. Then there is a sentence
  $\iso G\in\LW[s+2]$ that identifies $G$.
\end{lemma}

\begin{proof}
  Suppose that $G$ is embedded in a surface $\FS$ of Euler genus $g$
  with representativity $\rho(G)\le2$. Let $\Fg$ be a $G$-normal non-contractible
  simple closed curve in $\FS$ such that $U \coloneqq \Fg\cap V(G)$ contains at
  most two vertices. We only consider the case that $U=\{u_1,u_2\}$ for some $u_1, u_2 \in V(G)$ (possibly equal), the case $U = \emptyset$ follows similarly. Let $H_1,\ldots,H_m$ be the connected
  components of $G\setminus U$. 
 Every $H_i$
  can be embedded into a simpler surface obtained from $\FS$ by
  cutting along $\Fg$ and gluing (a) disk(s) on the hole(s). This
  means that $\eg(H_i)\le g-1$. We colour the vertices of $H_i$ so as to encode the adjacencies to $u_1$ and $u_2$. By
  Assumption~\ref{ass:induction}, there is a $\LW[s]$-sentence
  $\psi_i$ that identifies the coloured version of $H_i$.  Thus by
  Corollary~\ref{cor:id-comp}, there is a $\LW[s]$-sentence $\psi$ that
  identifies the disjoint union of the coloured $H_i$, that is, the coloured version of
  $G\setminus\{u_1,u_2\}$. Now we can identify $G$ by a sentence saying that there
  exist vertices $x_1,x_2$ such that deleting these vertices leaves a
  graph satisfying $\psi$ and having the correct
  adjacencies to $x_1,x_2$. 
This requires $s+2$ variables.
\end{proof}

So we can restrict our attention to graphs that only have embeddings of
representativity at least $3$. Furthermore, by
Lemma~\ref{lem:3-connected} we can restrict our
attention to 3-connected graphs (at the cost of $1$ more variable). Recall that a \emph{polyhedral
  embedding} is an embedding of representativity at least $3$ of a
3-connected graph. 
Thus to prove
Lemma~\ref{lem:inductive-step} and thereby complete the proof of
Theorem~\ref{thm:main}, it remains to
prove the following lemma.

\begin{lemma}\label{lem:main}
   Let $G$ be a coloured graph polyhedrally embedded in a surface
  $\FS$ of Euler genus $g$. Then there is a sentence
  $\iso G\in\LW[s+3]$ that identifies $G$.
\end{lemma}

For the rest of the section, we fix a positive integer $n$. The
intended meaning of $n$ is that it is the order of the target graph
$G$. At this point we have fixed three numerical parameters:
the Euler genus $g$, the number $s$ of variables required to identify
graphs of smaller Euler genus, and the order $n$.

To prepare for the proof of Lemma~\ref{lem:main}, we define
a number of useful concepts in $\LW[k]$ for sufficiently small $k$.

We start the proof with a simple lemma that follows immediately from
Assumption~\ref{ass:induction}.

\begin{lemma}\label{lem:define-genus}
  Let $h<g$. Then there is a sentence
  $\formel{genus}_{h} \in \LW[s]$ such that for every graph
  $G$ of order $|G| \leq n$, the following holds:
    \[
      G \models \formel{genus}_{h} \iff G \in \CE_{h}.
    \]
\end{lemma}

\begin{proof}
  Since there are only finitely many graphs of order at most $n$,
        we can let $\formel{genus}_{h}$ be a disjunction over the
        $\iso G \in \LW[s]$ for all $H \in \CE_{h}$ with $|H|\le n$.
\end{proof}

In the following lemmas, we study the definability of shortest path
systems, patches, and necklaces. Our strategy will then be to remove
either a (definable) reducing necklace or a (definable) simplifying patch
from the graph, then apply the induction hypothesis
(Assumption~\ref{ass:induction}) to the resulting simpler graph, and
finally lift the identifying sentence to the original graph.

\begin{lemma}\label{lem:csps-def}\sloppy
  There are formulae $\formel{csps-vert}(x,x',y) \in \LW[3]$,
  $\formel{csps-edge}(x,x',y_1,y_2) \in \LW[4]$,
  $\formel{csps-art}(x,x',y) \in \LW[4]$, and, for $i\ge 0$, formulae
  $\formel{csps-height}_k(x,x',y) \in \LW[3]$ and 
  $\formel{csps-art}_i(x,x',y) \in \LW[4]$
   such that for all connected graphs $G$ of order $|G|\le n$ and all vertices
   $u,u'\in V(G)$,
  \begin{align*}
    \formel{csps-vert}[G,u,u',y]&=V\big(\CQ^G(u,u')\big),\\
    \formel{csps-edge}[G,u,u',y_1,y_2]&=E\big(\CQ^G(u,u')\big),\\
    \formel{csps-art}[G,u,u',y]&=\art\big(\CQ^G(u,u')\big),\\
    \formel{csps-height}_i[G,u,u',y]&=\Big\{v \in V\big(\CQ^G(u,u')\big)\Bigmid
    \height^{\CQ^G(u,u')}(v) = i \Big\},\\
    \formel{csps-art}_i[G,u,u',y]&=\{v\},\; \parbox[t]{6cm}{where $v$ is the $i$-th vertex when sorting $\art\big(\CQ^G(u,u')\big)$ by height.}
  \end{align*}
\end{lemma}

Recall that $\CQ^G(u,u')$ is the canonical sps from
$u$ to $u'$, that is, the set of all shortest paths from $u$ to $u'$.
\begin{proof}
  We let $\formel{csps-vert}(x,x',y) \coloneqq
  \bigvee\limits_{k=0}^{n} \left(\formel{dist}_{=k}(x,x') \wedge
    \bigvee\limits_{i=0}^k \big(\formel{dist}_{=i}(x,y) \wedge
      \formel{dist}_{=k-i}(y,x')\big)\right)$, where $\formel{dist}_{=k}(x,x')$ is
  the $\LW[3]$-formula defined in Example~\ref{exa:dist}. Note that $\formel{csps-vert}(x,x',y)\in\LW[3]$.

  Since a vertex $v$ lies on a shortest path from $u$ to $u'$ if and only if taking the shortest path from $u$ to $v$ and then to $u'$ yields no detour, the formula $\formel{csps-vert}$ defines the desired set of vertices.

  An edge is contained in $E\big(\CQ^G(u,u')\big)$ if and only if it connects an sps-vertex of a certain height $h$ with an sps-vertex of height $h+1$. Thus, it is easy to see that the formula for $\formel{csps-edge}$ can be constructed to have width $4$.

  A vertex $v$ is an articulation vertex of $\CQ^G(u,u')$ if every shortest path from $u$ to $u'$ contains $v$:
  \begin{align*}
    \formel{csps-art}(x,x',y) \coloneqq{}
    & \formel{csps-vert}(x,x',y) \wedge{}\\
    & \forall z \Big(\formel{csps-vert}(x,x',z) \rightarrow \big(\formel{csps-vert}(x,y,z) \vee \formel{csps-vert}(y,x',z)\big)\Big).
  \end{align*}
  This formula has width $4$.

  Similarly, the height of $v$ in $\CQ^G(u,u')$ is $i$ if and only if $v$ is contained in the sps and $G \models \formel{dist}_{=i}(u,v)$. Thus, we can construct $\formel{csps-height}_i(x,x',y)$ with width $3$.

  By employing $\formel{csps-art}$ and $\formel{csps-height}_{i}$, we can also construct $\formel{csps-art}_i$ with width $4$.
\end{proof}

The lemma shows how to define canonical shortest paths systems. We
would also like to define patches and necklaces, but they depend on the
embedding and since the embedding may not be unique, in general
the property of an sps being a patch is not definable in a logic which only has access to the abstract graph and not the embedding. We therefore define
``pseudo-patches'' and ``pseudo-necklaces'' purely in terms of the abstract graph; in
some situations they may serve as substitutes for the real object.

\begin{definition}
  Let $G$ be a graph. 
  \begin{enumerate}
  \item A \emph{pseudo-patch} in $G$ is an sps that has no
    articulation vertices.
  \item A \emph{pseudo-necklace} in $G$ is a tuple $\CB \coloneqq (u^0,\CQ^0,u^1,\mathcal{Q}^1,u^2,\CQ^2)$,
  where $u^0,u^1,u^2\in V(G)$ and $\CQ^i = \CQ^G(u^{i},u^{i+1})$
  (indices taken modulo 3) is the canonical sps
  from $u^{i}$ to $u^{i+1}$, such that $u^0,u^1,u^2$ are pairwise
  distinct and 
    $V(\CQ^{i})\cap V(\CQ^{i+1})=\{u^{i+1}\}$ (indices
    modulo 3).
 \uend 
  \end{enumerate}
\end{definition}

All the definitions for general sps apply to pseudo-patches, and we
can generalise all definitions that do not refer to the embedding (for
example, $V(\CB)$, $E(\CB)$, articulation vertices, et cetera) from
necklaces to pseudo-necklaces. Observe that every patch is a pseudo-patch and
every necklace is a pseudo-necklace. 


\begin{corollary}\label{cor:necklace-def1}
  There are $\LW[4]$-formulae 
  \[
  \begin{array}{l@{\hspace{1.5cm}}l}
  \formel{nl-vert}(x^0,x^1,x^2,y),&
  \formel{nl-edge}(x^0,x^1,x^2,y),\\
  \formel{nl-art}(x^0,x^1,x^2,y),&
  \formel{nl-art}_i(x^0,x^1,x^2,y),\\
  \end{array}
  \]
  such that for all connected graphs $G$ of order $|G|\le n$ and all $u^0,u^1,u^2\in
  V(G)$ the following holds. If
  $\CB \coloneqq \big(u^0,\CQ^0,u^1,\CQ^1,u^2,\CQ^2\big)$
  is a pseudo-necklace in $G$,
  then
  \begin{align*}
  \formel{nl-vert}[G,u^0,u^1,u^2,y] &=V(\CB); \\
  \formel{nl-edge}[G,u^0,u^1,u^2,y_1,y_2] &=E(\CB); \\
  \formel{nl-art}[G,u^0,u^1,u^2,y] &=\art(\CB); \\
  \formel{nl-art}_i[G,u^0,u^1,u^2,y] &= \{v\},\; \parbox[t]{7.5cm}{where $v$ is the $i$-th vertex in the linear
    order of $\art(\CB)$ that orders the articulation vertices of
    $\CQ^0$, $\CQ^1$, and $\CQ^2$ by increasing height and puts the articulation vertices of
    $\CQ^0$ before those of $\CQ^1$ and the latter ones before those of $\CQ^2$.}
  \end{align*}
\end{corollary}

\begin{proof}
  For a vertex $v \in V(G)$, we have that $v \in V(\CB)$ if and only if $v \in \CQ^G(u^i, u^{i+1})$ for some $i \in \{0,1,2\}$ (indices taken modulo 3). Similarly, an edge is a necklace edge if and only if for some $i$, it connects a vertex $v \in V\big(\CQ^G(u^i, u^{i+1})\big)$ of a certain height $h$ with a vertex of height $h+1$ in $V\big(\CQ^G(u^i, u^{i+1})\big)$. Thus, containment in $V(\CB)$ and in $E(\CB)$ is definable in $\LW[4]$.

  A vertex $v$ is an articulation vertex of $\CB$ if $v$ equals $u_i$ or $v$ is an articulation vertex of $\CQ^G(u^i, u^{i+1})$ 
  for some $i \in \{0,1,2\}$. Thus, we can construct $\formel{nl-art}$ to have width $4$.

  We can construct $\formel{nl-art}_i$ in a straightforward manner by employing the subformulae $\formel{nl-art}$ and $\formel{csps-height}_i$.
\end{proof}

From Lemmas~\ref{lem:avoiding-path} and~\ref{lem:csps-def}, we obtain the
following corollary.

  \begin{corollary}
    There is a formula
    $\formel{csps-comp}(x,x',y,y') \in \LW[5]$ such
    that for all connected graphs $G$ of order $|G|\le n$ and all
    $u,u',v,v'\in V(G)$,
    \[
     G \models
      \formel{csps-comp}(u,u',v,v')\iff\parbox[t]{7cm}{$v$
        and $v'$ belong to the same connected component of
        $G\setminus V\big(\CQ^G(u,u')\big)$.}
    \]
  \end{corollary}

From
Lemma~\ref{lem:rel2comp2} applied to the $\LW[s]$-sentence
$\psi\coloneqq\formel{genus}_h$ 
of Lemma~\ref{lem:define-genus} and the $\LW[3]$-fomula
$\phi(x,x',y)\coloneqq \formel{csps-vert}(x,x',y)$  of
Lemma~\ref{lem:csps-def}, we obtain the following corollary.

\begin{corollary}\label{cor:comp-genus2}
    Let $h<g$. Then there is a formula
    $\formel{csps-comp-genus}_{h}(x,x',y) \in
    \LW[s+2]$ such that for all connected graphs $G$ of order
    $|G|\le n$ and all $u,u'\in V(G)$ the following
    holds. Let $\CQ \coloneqq \CQ^G(u,u')$, and let $A$ be the
    connected component of $v$ in $G\setminus V(\CQ)$ (assuming
    $v\not\in V(\CQ)$). Then
    \[
     G \models
      \formel{csps-comp-genus}_{h}(u,u',v)\iff v\not\in
      V(\CQ)\text{ and }\eg(A)\le h.
    \]
  \end{corollary}

  \begin{corollary}\label{cor:comp-genus3}
    There is a formula
    $\formel{csps-simplifying}(x,x') \in \LW[s+2]$ such
    that for all connected graphs $G\in\CE_g$ of order $|G|\le n$ and all
    $u,u'\in V(G)$,
    \[
     G \models
      \formel{csps-simplifying}(u,u')\iff\parbox[t]{7cm}{$\CQ^G(u,u')$
        is simplifying.}
\]
  \end{corollary}

The formulae we have defined so far make no reference to an embedding
of the input graph. However, if we want to talk about patches and
necklaces, we need to take the embedding into account. For the rest of the
section, we fix a specific embedded graph $G$.

\begin{assumption}\label{ass:polyhedral2}
  $G$ is a coloured graph of order $|G|=n$ that is polyhedrally
  embedded in a surface $\FS$ of Euler genus $g$.
\end{assumption}

It is our goal to construct a $\LW[s+3]$-sentence that identifies $G$. 

Intuitively, the followinglemma says that even though the logical formulae only have access to the abstract graph and the disk of a patch and the internal graph depend on the embedding, we can still define the
internal graph. This is non-trivial and somewhat surprising.

\begin{lemma}\label{lem:disk-non-def}
  There are formulae 
  $\formel{int-vert}(x,x',y)$, $\formel{int-edge}(x,x',y_1,y_2)$, $\formel{bd-vert}(x,x',y)$,
  $\formel{bd-edge}(x,x',y_1,y_2)$ in $\LW[7]$ such that for all vertices
  $u,u'\in V(G)$ for which $\CQ \coloneqq \CQ^G(u,u')$ is a non-trivial
  non-simplifying patch, the following holds:
  \begin{align*}
    \formel{int-vert}[G,u,u',y]&=V\big(I(\CQ)\big),\\
    \formel{int-edge}[G,u,u',y_1,y_2]&=E\big(I(\CQ)\big),\\
    \formel{bd-vert}[G,u,u',y]&=V\big(C(\CQ)\big),\\
    \formel{bd-edge}[G,u,u',y_1,y_2]&=E\big(C(\CQ)\big).
  \end{align*}
\end{lemma}

\begin{proof}
  Let $u,u'\in V(G)$ such that $\CQ \coloneqq \CQ^G(u,u')$ is a
  non-trivial non-simplifying patch. Let $\FD \coloneqq \FD(\CQ)$, $C \coloneqq C(\CQ)$,
  and $I\coloneqq I(\CQ)$ (see Section~\ref{sec:sps-patch-necklace}).

By Lemma~\ref{fact:astar}, the graph $G \setminus V(\CQ)$ has a unique
non-planar connected component $A^*$. We let
\begin{align*}
\formel{planar-comp}(x,x',y) &\coloneqq {}
  \formel{csps-comp-genus}_0(x,x',y)\qquad\text{(see
  Corollary~\ref{cor:comp-genus2})}
\intertext{and}
\formel{astar}(x,x',y)&\coloneqq {}\neg \formel{csps-vert}(x,x',y) \wedge
\neg\formel{planar-comp}(x,x',y). 
\end{align*}
Note that both $\formel{planar-comp}(x,x',y)$ and $\formel{astar}(x,x',y)$ are
$\LW[s+2]$-formulae. In fact, since we can
identify planar graphs in the logic $\LW[4]$, we can construct these formulae as
$\LW[6]$-formulae.
For $v\in V(G)\setminus
V(\CQ)$, we have $G\models \formel{planar-comp}(u,u',v)$ if and
only if the connected component of $v$ in $G\setminus V(\CQ)$ is
planar, and $G\models \formel{astar}(u,u',v)$ if and
only if the connected component of $v$ in $G\setminus V(\CQ)$ is
$A^*$.
  
Let $v_1$ be a vertex in $V(\CQ)$ that is adjacent to $A^*$ and among
all such vertices has minimal height in the sps, and let $h \coloneqq \height^{\CQ}(v_1)$. Since $A^*$ is embedded outside of the disk $\FD$, the vertex
$v_1$ must be on the boundary cycle $C$ of $\FD$. There is at
most one other vertex of height $h$ on this cycle. Thus, even though
$v_1$ is not unique, there are at most two choices. If there is
a second vertex of height $h$ adjacent to $A^*$, let us call it
$v_1'$. Let
  \begin{align*}
    \varphi_1(x,x',y_1) \coloneqq{} &\formel{csps-height}_h(x,x',y_1) \wedge \exists z^* \big(E(z^*,y_1) \wedge \formel{astar}(x,x',z^*)\big) \wedge {}
    \\ & \neg \exists y_1' \left(\bigvee\limits_{i=0}^{h-1} \formel{csps-height}_i(x,x',y_1') \wedge \exists z^* \big(E(z^*,y_1') \wedge \formel{astar}(x,x',z^*)\big)\right).
  \end{align*}

Then $v_1$ and possibly
$v_1'$ are the only vertices in $\phi_1[G,u,u',y_1]$. Note that $\varphi_1 \in \LW[6]$.

Recall that $G/A^*$ denotes the graph obtained from $G$ by contracting
the connected subgraph $A^*$ to a single vertex, which we call $a^*$,
and that the graph $G/A^*$ is a 3-connected planar graph.
By Whitney's Theorem, the facial subgraph of a
3-connected plane graph are precisely the chordless non-separating
cycles. In particular, they are independent of the
embedding. Furthermore, every edge is contained in exactly two of
these facial cyles. Let us consider the edge $v_1a^*$ in the graph
$G/A^*$. Let $F$ and $F'$ be the two facial cycles that contain this
edge. Both $F$ and $F'$ contain exactly one neighbour of $a^*$
distinct from $v_1$. Let $v_2$ and $v_2'$ be these neighbours. 

By \cite[Lemma~22]{kieponschwe17}, if we have a 3-connected planar
graph $H$ and three vertices $w_1,w_2,w_3$ on a common facial cycle,
then after individualising these three vertices, the 1-dimensional WL
algorithm computes a discrete colouring. By Theorem~\ref{thm:immlan}, this implies that for every vertex $w \in V(H)$ there is a formula
$\psi_{H,w}(z_1,z_2,z_3,z)\in\LW[5]$ such that
$\psi_{H,w}[H,w_1,w_2,w_3,z]=\{w\}$. We apply \cite[Lemma~22]{kieponschwe17} to the graph
$G/A^*$ and the three vertices $a^*,v_1,v_2$ and obtain, for every vertex $w\in V(G/A^*)=(V(G)\setminus A^*)\cup\{a^*\}$, a formula
$\psi_w(z^*,y_1,y_2,z)\in\LW[5]$ such that
$\psi_w[G/A^*,a^*,v_1,v_2,z]=\{w\}$.

Let $w\in V(G)\setminus A^*$. Recall that $\formel{astar}(x,x',y) \in
\LW[6]$ and $\psi_w(z^*,y_1,y_2,y) \in \LW[5]$. By
Lemma~\ref{lem:factor} (applied to $k=6$, $\ell=2$, $m=5$ and the formulae
$\xi(x,x',z^*)\coloneqq\formel{astar}(x,x',z^*)$ and
$\psi(y_1,y_2,y,z^*)\coloneqq \psi_w(z^*,y_1,y_2,y)$), there is a formula $\tilde\psi_w(x,x',y_1,y_2,y)\in\LW[6]$ such that
 $\tilde\psi_w[G,u,u',v_1,v_2,y]=\{w\}$. 

Since $A^*\cap \FD=\emptyset$,
we have $V(I)=V(G)\cap \FD\subseteq V(G\setminus A^*)$. We let
  \[
    \delta(x,x',y_1,y_2,z) \coloneqq {}\bigvee_{w\in V(I)}\tilde\psi_w(x,x',y_1,y_2,z).
  \] 
  Then $\delta[G,u,u',v_1,v_2,z]=V(I)$. Thus $\delta(x,x',y_1,y_2,z)$ is almost
  the formula $\formel{int-vert}$ we want, except that it
  has two additional parameters $v_1,v_2$ which we have to get rid of. 

  We will apply \cite[Corollary~26]{kieponschwe17}, which says that
  the $3$-dimensional WL algorithm determines orbits in
  coloured 3-connected graphs. This implies that within a given graph, the $3$-dimensional WL algorithm distinguishes
  two vertices if and only if they belong to different orbits of the
  automorphism group of the graph. It follows that for every 3-connected planar
  graph $H$ and for every orbit $O$ of the automorphism group of $H$
  there is a formula $\xi_{H,O}(y_2)\in\LW[4]$ such that $\xi_{H,O}[H,y_2]=O$. 

  To eliminate the parameter $v_2$, we apply the corollary to the
  graph $G/A^*$, but only after individualising the vertices
  $a^*$ and $v_1$. (That is, we modify the colouring such that each of the two vertices has its own colour and is thus fixed by all
  automorphisms.) Let $O_2$ be the orbit of $v_2$ in this coloured
  graph. By the definition of $v_2$, either $O_2=\{v_2,v_2'\}$ or
  $O_2=\{v_2\}$. Since the graph $G/A^*$ is 3-connected, by eliminating the colour relations for $a^*$ and $v_1$ at the cost of new free variables $z^*$ and $y_1$, we obtain a new formula
 $\psi_2(z^*,y_1,y_2)\in\LW[6]$ such that
  $\psi_2[G/A^*,a^*,v_1,y_2]=O_2$. Since $\formel{astar}(x,x',y) \in
  \LW[6]$ and $\xi_{H,O}(y)\in\LW[4]$, by Lemma~\ref{lem:factor}
  (with $k=6$, $\ell=2$, $m=4$ and the formulae
  $\xi(x,x',z^*)\coloneqq\formel{astar}(x,x',z^*)$ and
  $\psi(y_1,y_2,z^*)\coloneqq\psi_2(z^*,y_1,y_2)$), there is a formula
  $\tilde\psi_2(x,x',y_1,y_2)\in\LW[6]$ such that $\tilde\psi_2[G,u,u',v_1,y_2]=O_2$. We let
  \[
    \delta'(x,x',y_1,z) \coloneqq {}\exists
    y_2\big(\tilde\psi_2(x,x',y_1,y_2)\wedge\delta(x,x',y_1,y_2,z)\big).
  \]
  If $O_2=\{v_2\}$ then clearly
  $\delta'[G,u,u',v_1,z]=\delta[G,u,u',v_1,v_2,z]=V(I)$. So suppose
  that $O_2=\{v_2,v_2'\}$, and let $f$ be an automorphism of $G$ with
  $f(u)=u$, $f(u')=u'$, $f(v_1)=v_1$, and $f(v_2)=v_2'$. By
  Corollary~\ref{cor:Q-invariance}, we have $f\big(V(I)\big)=V(I)$ and thus
  \begin{align*}
    \delta[G,u,u',v_1,v_2',z]&=\delta[f(G),f(u),f(u'),f(v_1),f(v_2),z]\\
    &=f\big(
      \delta[G,u,u',v_1,v_2,z]\big)\\
    &=f\big(V(I)\big)=V(I).
  \end{align*}
  It follows that
  \[
    \delta'[G,u,u',v_1,z]=\delta[G,u,u',v_1,v_2,z]\cup
    \delta[G,u,u',v_1,v_2',z]=V(I).
  \]
  So we have eliminated the parameter $v_2$. To eliminate $v_1$, we
  use essentially the same argument. Let $O_1$ be the orbit of $v_1$
  in the graph $G/A^*$ with the vertices $a^*$, $u$, $u'$
  individualised. Then either $O_1=\{v_1,v_1'\}$ for some $v_1' \neq v_1$ or
  $O_1=\{v_1\}$. 

  By \cite[Corollary~26]{kieponschwe17}, there is a formula
  $\xi_{G/A^*,O_1}(y_1) \in \LW[4]$ such that 
  \[
    \xi_{G/A^*,O_1}[(G/A^*)_{a^*,u,u'},y_1] = O_1.
  \]
  Then by eliminating the colour relations for $a^*$, $u$, $u'$ at the cost of new free variables $z^*$, $x$, $x'$, we obtain a formula $\psi_1(z^*,x,x',y_1)\in\LW[7]$ such
  that $\psi_1[G/A^*,a^*,u,u',y_1]=O_1$. Since $\formel{astar}(x,x',y)
  \in \LW[6]$ and $\xi_{G/A^*,O}(y_1) \in \LW[4]$, by
  Lemma~\ref{lem:factor} (with $k=7$, $\ell=2$, $m=3$ and the formulae $\xi(x,x',z^*)\coloneqq\formel{astar}(x,x',z^*)$ and
  $\psi(x,x',y_1,z^*)\coloneqq\psi_1(z^*,x,x',y_1)$), there is a formula
  $\tilde\psi_1(x,x',y_1)\in\LW[7]$ such that $\tilde\psi_1[G,u,u',y_1]=O_1$.

  We let
  \[
    \formel{int-vert}(x,x',z) \coloneqq {}\exists
    y_1\big(\tilde\psi_1(x,x',y_1)\wedge\delta'(x,x',y_1,z)\big).
  \]
  Now a similar argument as above shows that
  $\formel{int-vert}[G,u,u',z]=V(I)$. Moreover, since $\delta', \tilde\psi_1 \in \LW[7]$, we have $\formel{int-vert} \in \LW[7]$.

  The formulae $\formel{int-edge}(x,x',y_1,y_2)$, $\formel{bd-vert}(x,x',y)$,
  $\formel{bd-edge}(x,x',y_1,y_2)$ can be defined similarly.
\end{proof}

Now we branch into two cases, depending on whether $G$ contains a simplifying
patch or not.

\subsection{Case 1: Absence of simplifying patches}\label{subsec:necklace} 

Throughout this subsection, in addition to Assumption \ref{ass:polyhedral2}, we assume the following.

\begin{assumption}\label{ass:no-simplifying}
  $G$ does not contain any simplifying patches.
\end{assumption}

By Lemma~\ref{lem:exnecklace}, $G$ contains a reducing necklace $\CB$. We are going
to define a subgraph $\Cut(\CB)$ of $G$ that is obtained from $G$ by ``cutting through the beads''. Since the
necklace is reducing, the Euler genus of every connected component of $\Cut(\CB)$ is at most $g-1$ and
we can identify it with a $\LC[s]$-sentence. We colour $\Cut(\CB)$ in
such a way that we can reconstruct $G$ and identify it using only $3$
more variables.

For a necklace $\CB \coloneqq (u^0, \CQ^1, u^1, \CQ^2, u^2, \CQ^3)$ in $G$, let
$u^i = u_0^i, u_1^i, \dotsc, u_{n_i}^i = u^{i+1}$ be the articulation
vertices of $\CQ^i$, ordered by height, and for
$j \in \{0, \dotsc, n_i-1\}$ let
$\CQ_j^i \coloneqq \CQ^i[u_j^i, u_{j+1}^i]$ be the segment of $\CQ^i$ between $u_j^i$ and $u_{j+1}^i$. If the patch $\CQ^i_j$ is
trivial, we denote its unique edge by $e^i_j$. If $Q^i_j$ is
non-trivial, we let $\FD^i_j \coloneqq \FD(\CQ^i_j)$.

The \emph{region} of $\CB$ is the point set
\[
\FR(\CB) \coloneqq \bigcup_{i=0}^2 
\Big(\bigcup_{\substack{0\le j\le     n_i-1\\\CQ\text{ non-trivial}}}\FD^i_j\cup \bigcup_{\substack{0\le j\le
     n_i-1\\\CQ\text{ trivial}}}\Fe^i_j\Big).
\]

Recall that the internal graph of a non-simplifying patch $\CQ$ is the graph $I\coloneqq I(\CQ)$ with vertex set $V(I)\coloneqq V(G)\cap\FD(\CQ)$ and edge set
$E(I)\coloneqq\{\Fe\in E(G)\mid \Fe\subseteq\FD(\CQ)\}$. We associate three subgraphs of $G$ with $\CB$:

\begin{definition}\label{def:cut}
The \emph{inside} of $\CB$ is 
$I(\CB)\coloneqq\bigcup_{i=0}^2\bigcup_{j=0}^{n_i-1}I(\CQ^i_j)$.

The \emph{outside} of $\CB$ is the graph $O(\CB)$ defined by 
\begin{equation*}
  \begin{split}
  V\big(O(\CB)\big)&\coloneqq {}V(G)\setminus\Fint(\FR),\\
  E\big(O(\CB)\big)&\coloneqq {}E(G)\setminus\big\{\Fe\in E(G)\bigmid \Fe\cap\Fint(\FR)\neq\emptyset\big\}.
  \end{split}
\end{equation*}
The \emph{cut graph} of $\CB$ is $\Cut(\CB)\coloneqq
O(\CB) \backslash \art(\CB)$. \uend
\end{definition}

\begin{lemma}\label{lem:necklace-simp}
  Suppose $\CB$ is a reducing necklace in $G$. Then every connected
  component of $\Cut(\CB)$ is in $\CE_{g-1}$.
\end{lemma}

\begin{proof}
    This proof is a slight adaptation of the proof of \cite[Lemma~15.5.6]{gro17}. 

    Let $\FR \coloneqq \FR(\CB)$. For all $i,j$ such that $\CQ^i_j$ is
    non-trivial, we let $\FD^i_j \coloneqq \FD(\CQ^i_j)$.

  Let $B$ be a non-contractible cyle in $\CB$, whose existence is
guaranteed by Definition \ref{def:necklace-nc}. Then $\FB\subseteq\FR$ is
a simple closed curve, and for all $i,j$ such that $\CQ^i_j$ is non-trivial, the intersection
  $\FB^i_j \coloneqq \FD^i_j\cap\FB$ is a simple curve in the disk $\FD^i_j$
  with endpoints $u^i_{j}$ and $u^i_{j+1}$. By slightly perturbing
  $\FB$, we obtain a homotopic simple closed curve $\Fb\subseteq\FR$ such that for
  all $i,j$ with non-trivial $\CQ^i_j$ we have
  $\Fb\cap\Fbd(\FD^i_j)= \{u^i_{j},u^i_{j+1}\}$. This new curve $\Fb$ is
  still non-contractible, and it intersects
  $\Fbd(\FR)$ only in the articulation vertices $u^i_j$ of $\CB$ and in the
  edges $\Fe^i_j$ of the trivial patches $\CQ^i_j$. 

  This implies that for $H \coloneqq \Cut(\CB)$ we have
  $\Fb\cap\FH=\emptyset$. 
  Thus $\FH\subseteq G\setminus\Fb$, and since $\Fb$ is
  non-contractible, this implies that every connected component of $H$
  is embeddable in a surface of Euler genus at most $g-1$ obtained
  from $\FS$ by cutting along $\Fg$ and gluing a disk on each hole.
\end{proof}

Our next goal is to show that the cut graph is definable in
$\LW[s+3]$. We start with the definability of patches.

From Lemma~\ref{lem:disk-non-def}, we obtain that $\LW[7]$ distinguishes the internal graph of a reducing necklace from the remainder of the graph.

\begin{corollary}\label{lem:necklace-def1}
  There are $\LW[7]$-formulae 
  \[
  \begin{array}{l@{\hspace{1.5cm}}l}
  \formel{nl-int-vert}(x^0,x^1,x^2,y),&
  \formel{nl-int-edge}(x^0,x^1,x^2,y_1,y_2).
  \end{array}
  \]
  such that for $u^0,u^1,u^2\in
  V(G)$ the following holds. 

  If
  $\CB \coloneqq (u^0,\CQ^0,u^1,\CQ^1,u^2,\CQ^2)$ is a necklace in $G$,
  then
  \begin{align*}
    \formel{nl-int-vert}[G,u^0,u^1,u^2,y]&=V\big(I(\CB)\big); \\
  \formel{nl-int-edge}[G,u^0,u^1,u^2,y_1,y_2]&=E\big(I(\CB)\big).
  \end{align*}
\end{corollary}

\begin{proof}
  Remember that we suppose Assumption \ref{ass:no-simplifying}. Thus, we can simply define 
    \begin{align*}
      \formel{nl-int-vert}(x^0,x^1,x^2,y) \coloneqq {}\bigvee\limits_{i=0}^2 \formel{int-vert}(x^i,x^{i+1},y).
    \end{align*}
  Similarly, we obtain the formula $\formel{nl-int-edge}$ with the desired width.
\end{proof}

In the following we show that $\LC[7]$ distinguishes vertices in the outside and the cut graph of $\CB$ from the rest of the graph.

\begin{lemma}\label{lem:necklace-def2}\sloppy
  There are $\LW[7]$-formulae
  \[
  \begin{array}{l@{\hspace{1.5cm}}l}
    \formel{nl-out-vert}(x^0,x^1,x^2,y),&
    \formel{nl-out-edge}(x^0,x^1,x^2,y),\\
    \formel{nl-cut-vert}(x^0,x^1,x^2,y),&
    \formel{nl-cut-edge}(x^0,x^1,x^2,y_1,y_2)
  \end{array}
  \]
  such that for all
  $u^0,u^1,u^2\in V(G)$ the following holds: if
  $\CB \coloneqq (u^0,\CQ^0,u^1,\CQ^1,u^2,\CQ^2)$ is
  a necklace in $G$, then
  \begin{align*}
  \formel{nl-out-vert}[G,u^0,u^1,u^2,y]=V\big(O(\CB)\big); \\
  \formel{nl-out-edge}[G,u^0,u^1,u^2,y_1,y_2]=E\big(O(\CB)\big); \\
  \formel{nl-cut-vert}[G,u^0,u^1,u^2,y]=V\big(\Cut(\CB)\big); \\
  \formel{nl-cut-edge}[G,u^0,u^1,u^2,y_1,y_2]=E\big(\Cut(\CB)\big).
  \end{align*}
\end{lemma}

\begin{proof}
  Let $\FR \coloneqq \FR(\CB)$ and recall that $u^i = u_0^i, u_1^i, \dotsc, u_{n_i}^i = u^{i+1}$ denote the articulation
vertices of $\CB$, ordered by height, and that for
$j \in \{0, \dotsc, n_i-1\}$, we denote the segment of $\CQ^i$ between $u_j^i$ and $u_{j+1}^i$ by 
$\CQ_j^i \coloneqq \CQ^i[u_j^i, u_{j+1}^i]$.
  Since by
  Assumption~\ref{ass:no-simplifying}, all
 subpatches are non-simplifying, it holds that 
  \[V(G)\cap\Fint(\FR)=\bigcup_{i=0}^2\bigcup_{\substack{0\le j\le
  n_i-1\\\CQ^i_j\text{ non-trivial}}} V\left(I(\CQ^i_j) \setminus C(\CQ^i_j)\right).\] 
  Therefore,
  \begin{align*}
    V\big(O(\CB)\big)&=V(G)\setminus \bigcup_{i=0}^2\bigcup_{\substack{0\le j\le
  n_i-1\\\CQ^i_j\text{ non-trivial}}} V\big(I(\CQ^i_j) \setminus C(\CQ^i_j)\big),\\
    E\big(O(\CB)\big)&=E(G)\setminus\bigcup_{i=0}^2\bigcup_{\substack{0\le j\le
    n_i -1\\\CQ^i_j\text{ non-trivial}}}E\big(I(\CQ^i_j) \setminus C(\CQ^i_j)\big).
  \end{align*}

  Thus, we can just let 
  \begin{align*}
    \formel{nl-out-vert}(x^0,x^1,x^2,y) \coloneqq {}&\bigvee\limits_{i=1}^n \exists z \exists z' \big(\formel{nl-art}_i(x^0,x^1,x^2,z) \wedge \formel{nl-art}_{i+1}(x^0,x^1,x^2,z') 
    \\&\wedge \neg \formel{nl-edge}(x^0,x^1,x^2,z,z') \wedge \formel{bd-vert}(z,z',y)\big) 
    \\&\vee \formel{nl-art}(x^0,x^1,x^2,y) \vee \neg \formel{nl-int-vert}(x^0,x^1,x^2,y),
  \end{align*}
  where the big disjunction expresses that the given vertex lies on the boundary of some disk of a non-trivial patch.

Similarly, we obtain the formula $\formel{nl-out-edge}$ of width 7.

To define that a vertex is contained in the cut graph, we just need to guarantee that it is contained in $O(\CB)$ and that is not an articulation vertex of the necklace. Similarly, for an edge contained in $O(\CB)$, to appear in $\Cut(\CB)$, its incident vertices must not be articulation vertices of $\CB$. We obtain the desired $\LW[7]$-formulae $\formel{nl-cut-vert}$ and $\formel{nl-cut-edge}$.
\end{proof}

We have collected all ingredients to show the statement from Lemma \ref{lem:main} in case $G$ contains no simplifying patches.

\begin{proof}[Proof of Lemma \ref{lem:main}, Case 1]
  We show that the statement holds if $g \geq 1$ and $G$ does not
  contain any simplifying patches.

  Recall that by Assumption \ref{ass:induction}, for every
  coloured graph $H \in \CE_{g-1}$, we assume the existence of a
  formula $\iso H \in \LW[s]$ such that for all graphs $G'$ it holds
  that \[G' \models \iso H \iff G' \cong H.\]

  Let $G$ be a coloured graph that does not contain any
  simplifying patches and is polyhedrally embedded in a surface $\FS$
  of genus $g \geq 1$. Let $\hat G$ be a second coloured graph such
  that there is no formula in $\LC[s+3]$ which distinguishes
  $G$ and $\hat G $. We show that $G \cong \hat G$.

  We may assume $|\hat G | = |G|$, otherwise we can distinguish $G$ and $\hat G $
  via the formula
  $\exists^{= |G|} v (v=v)$. 

  \cite[Theorem~5]{kieponschwe17} implies that if for some $k \geq 3$,
  the logic $\LC[k]$ distinguishes all non-isomorphic pairs of
  coloured $2$-connected graphs, then it distinguishes all pairs
  of non-isomorphic graphs in $\mathcal{C}$. Thus, if
  $\LC[s+3]$ distinguishes (the 2-connected graph) $G$ from
  every non-isomorphic 2-connected coloured graph, then the
  same logic distinguishes $G$ from every arbitrary non-isomorphic coloured graph and
  thus, it identifies $G$. Hence, we can assume $\hat G $ to be
  $2$-connected.

  Moreover, if $\hat G $ is not $3$-connected, then it has a separator of
  size $2$ whereas $G$ does not. Since for $k \geq 3$, the
  $k$-dimensional WL algorithm distinguishes $2$-separators from other
  pairs of vertices (see \cite[Corollary~14]{kieponschwe17}), by
  Corollary \ref{cor:correspondence}, there is a formula in $\LW[4]$
  which distinguishes $G$ and $\hat G $.

  Hence, without loss of generality we may assume that $\hat G $ is
  $3$-connected.

  By Lemma \ref{lem:exnecklace}, there is a reducing necklace
  $\CB \coloneqq (u^0, \CQ^0, u^1, \CQ^1, u^2, \CQ^2)$
  in $G$, which we fix for the rest of the proof. For a pseudo-necklace
  $\hat \CB \coloneqq (\hat u^0, \hat \CQ^0, \hat u^1, \hat \CQ^1, \hat u^2, \hat \CQ^2)$
  in ${\hat G }$, we say $\CB$ and $\hat \CB$ are \emph{isomorphic}, and write
  $\CB \cong \hat \CB$, if there is an isomorphism from $G(\CB)$ to
  $\hat G (\hat \CB)$ mapping $u^i$ to $\hat u^i$ for $i \in
  \{0,1,2\}$.

  \begin{claim}\label{ns:claim1}
    There is a formula $\formel{nl-iso}(x^0,x^1,x^2) \in \LW[6]$
    (not depending on $\hat G $) such that
    ${\hat G } \models \formel{nl-iso}(\hat u^0,\hat u^1,\hat u^2)$ if and only if
    $\hat \CB \coloneqq (\hat u^0, \hat \CQ^0, \hat u^1, \hat \CQ^1, \hat u^2, \hat \CQ^2)$
    is a pseudo-necklace with $\hat \CB \cong \CB$.

    \proof 
    $\hat \CB$ is a pseudo-necklace isomorphic to $\CB$ if and only
    if for all $i \in \{0,1,2\}$, the following two conditions hold for $\hat \CQ^i \coloneqq \CQ(\hat u^i,\hat u^{i+1})$.
    \begin{enumerate}
    \item $\hat G ({\hat \CQ}^i) \cong G(\CQ^i)$ via an isomorphism mapping $\hat u^i$ to
      $u^i$ and $\hat u^{i+1}$ to $u^{i+1}$.\label{it1:necklace-iso}
    \item $V(\hat \CQ^{i-1}) \cap V(\hat \CQ^{i+1}) = \{{\hat u}^i\}$. \label{it2:necklace-iso}
    \end{enumerate}

    Condition \ref{it2:necklace-iso} is easy to express in $\LW[6]$. To treat Condition \ref{it1:necklace-iso}, let $\formel{sps-iso}'_i \in \LW[4]$ be the sentence from
    Theorem \ref{thm:planar-dimension} which identifies the planar
    coloured graph
    $\CQ^i_{u^i,u^{i+1}} \coloneqq
    G(\CQ^i)_{u^i,u^{i+1}}$.
    Let $R^i$ and $R^{i+1}$ be the
    relations representing the unique colours of $u^i$ and
    $u^{i+1}$ in $\CQ^i_{u^i,u^{i+1}}$. We transform
    $\formel{sps-iso}'_i$ into a formula
    $\formel{sps-iso}_i(x^i, x^{i+1})$ such that
    $\hat G  \models \formel{sps-iso}_i(\hat u^i,\hat u^{i+1})$ if and only if
 $\hat G (\hat \CQ^i) \cong G(\CQ^i)$ via an isomorphism that maps $\hat u^i$ to
    $u^i$ and $\hat u^{i+1}$ to $u^{i+1}$. To this end, we first
    replace in $\formel{sps-iso}'_i$ every $R^i(z,z)$ with the formula
    $z=x^i$ and every $R^{i+1}(z,z)$ with $z=x^{i+1}$. To relativise
    $\formel{sps-iso}'_i$ to the $i$-th shortest path system, we also replace
    subformulae of the form $\exists y \varphi$ with
    $\exists y (\formel{csps-vert}(x^i, x^{i+1},y) \wedge \varphi)$
    and $E(y_1,y_2)$ with
    $\formel{csps-edge}(x^i,x^{i+1},y_1,y_2)$. By Lemma
    \ref{lem:csps-def}, the resulting formula
    $\formel{sps-iso}_i(x^i,x^{i+1})$ is in $\LW[6]$.

    Now we can define the desired $\LW[6]$-formula
    \begin{align*}
      \formel{nl-iso}(x^0,x^1,x^2) \coloneqq {}&\bigwedge\limits_{i=0}^2 \formel{sps-iso}_i(x^i, x^{i+1}) \wedge \forall y \bigwedge\limits_{i=0}^2 \Big(\big(\formel{csps-vert}(x^{i-1},x^i,y) \wedge {}
      \\&\formel{csps-vert}(x^i,x^{i+1},y)\big) \rightarrow y = x^i\Big),
    \end{align*}
    where we take indices modulo 3.
    \uend
  \end{claim}

  For the remainder of this proof, let
  $\hat \CB \coloneqq (\hat u^0, \hat \CQ^0, \hat u^1, \hat \CQ^1, \hat u^2, \hat \CQ^2)$
  be a pseudo-necklace in $\hat G $ such that $\CB\cong\hat \CB$. If no such
  pseudo-necklace exists, we can distinguish $G$ and $\hat G $ in $\LW[6]$ using
  Claim~\ref{ns:claim1}.
  Let $u^i = u_1^i, u_2^i, \dots, u_{n_i}^i$ be the articulation
  vertices of $\CQ^i$ ordered by increasing height in $\CB$. Since
  $\hat \CB \cong \CB$, there is a bijection from $\art(\CB)$ to
  $\art(\hat \CB)$ mapping each articulation vertex to one of equal
  height in $\hat \CB$. Thus, for simplicity, we use the same name for
  the two articulation vertices in $\CB$ and $\hat \CB$ of equal height. In the
  following, let $\CQ_{i,j} \coloneqq \CQ^G(u^i_j,u^i_{j+1})$ and
  $\hat \CQ_{i,j} \coloneqq \CQ^{\hat G }(u^i_j,u^i_{j+1})$. Note that the
  $\hat \CQ_{i,j}$ are pseudo-patches.
  By our assumption that $\hat \CB \cong \CB$, the pseudo-patch
  $\hat \CQ_{i,j}$ is trivial if and only if the patch $\CQ_{i,j}$ is.

  Let $I \coloneqq I(\CB)$. Let $\hat I$ be the graph
  with vertex set $V(\hat I)\coloneqq\formel{nl-int-vert}[\hat G ,\hat u^0,\hat u^1,\hat u^2,y]$ and
  edge set $E(\hat I)\coloneqq\formel{nl-int-edge}[\hat G ,\hat u^0,\hat u^1,\hat u^2,y_1,y_2]$. Since $\hat \CB$ need not be a proper necklace (it
  might not comply with the third item in Definition \ref{def:necklace}),
  the graph $\hat I$ is not necessarily the inside of a
  necklace. However, for simplicity, we also use the letter $I$ to refer
  to such a ``pseudo-inside'' just as we also use $\CB$ for all
  pseudo-necklaces. For simplicity, we call $I$ and $\hat I$
  \emph{isomorphic}, and we write $I \cong \hat I$, if
  $I_{u^0,u^1,u^2} \cong \hat I_{\hat u^0,\hat u^1,\hat u^2}$,
  i.e., if there is an isomorphism from $I$ to $\hat I$
  mapping $u^i$ to $\hat u^i$ for every $i \in \{0,1,2\}$. Note that
  every such isomorphism induces an isomorphism from $\CB$ to
  $\hat \CB$. 
  We also define a ``pseudo-inside'' for all the pseudo-patches
  $\hat \CQ_{i,j}$: we let $I(\hat \CQ_{i,j})$ be the graph with vertex set
    $\formel{int-vert}[\hat G ,u_j^i,u_{j+1}^i,y]$ and edge set
    $\formel{int-edge}[\hat G ,u_j^i,u_{j+1}^i,y]$.

  \begin{claim}[resume]\label{ns:claim2}
    There is a formula $\formel{inside-iso}(x^0,x^1,x^2) \in \LW[7]$
    (not depending on $\hat G $) such that
    $\hat G  \models \formel{inside-iso}(\hat u^0,\hat u^1,\hat u^2)$ if and only if
    $\hat I \cong I$.
 
    \proof 
    We have that $\hat I \cong I$ if and only if $\hat G $
    satisfies the following conditions for all $i \in \{0,1,2\}$.
    \begin{enumerate}[({1}a)]
    \item $I(\hat \CQ_{i,j})_{u^i_j, u^i_{j+1}} \cong I(\CQ_{i,j})_{u^i_j, u^i_{j+1}}$ for all $j \in [n_i-1]$.\label{item:1a}
      
    \item If $i' = i + 1 \mod 3 \text{ and } j = n_i-1 \text{ and } j' = 1$, it holds that $I(\hat \CQ_{i,j}) \cap I(\hat \CQ_{i',j'}) = \{\hat u^{i+1}\}$.\label{item:1b}
      
    \item If $j \in [n_i-1]$ and $j' = j+1$, it holds that $I(\hat \CQ_{i,j}) \cap I(\hat \CQ_{i,j'}) = \{u_{j+1}^i\}$.\label{item:1c}
      
    \item If $i' = i + 1 \mod 3 \text{ and } (j \neq n_i-1 \text{ or } j' \neq 1)$, it holds that $I(\hat \CQ_{i,j}) \cap I(\hat \CQ_{i',j'}) = \emptyset$. If $|j'-j| \geq 2$, it holds that $I(\hat \CQ_{i,j}) \cap I(\hat \CQ_{i,j'}) = \emptyset$.\label{item:1d}
    \end{enumerate}
    
    The ``only if'' follows from the definition of $I$ and the
    $\LC$-definability of $I(\CQ_{i,j})$. We now show the ``if''-part. Consider isomorphisms $\pi_{i,j}$ witnessing Item
    (1\ref{item:1a}). We define an isomorphism $\pi$ from $\hat I$ to
    $I$ by letting $\pi(v)$ be $\pi_{i,j}(v)$ where $i$ and
    $j$ are such that $I(\hat \CQ_{i,j})$ contains $v$. Note that by
    Items (1\ref{item:1b}), (1\ref{item:1c}) and (1\ref{item:1d}), if
    $I(\hat \CQ_{i,j}) \cap I(\hat \CQ_{i',j'}) \neq \emptyset$ for $\hat \CQ_{i,j} \neq \hat \CQ_{i',j'}$, then
    there is a unique vertex $v \in I(\hat \CQ_{i,j}) \cap I(\hat \CQ_{i',j'})$. In that
    case, Item (1\ref{item:1a}) guarantees that the two possible images
    $\pi(v)$ coincide. Thus, $\pi$ is well-defined and it certainly is
    an isomorphism.
    
    We still need to translate Items (1\ref{item:1a})--(1\ref{item:1d}) into
    $\LC$-formulae. Since the subgraph $I(\CQ_{i,j})$ of $G$ is embedded in the disk
    $\FD(\CQ_{i,j})$, it is planar. Hence, by Theorem
    \ref{thm:planar-dimension}, there is a sentence
    $\formel{disk-iso}'_{i,j} \in \LC[4]$ which identifies
    $I(\CQ_{i,j})_{u^i_j,
      u^i_{j+1}}$. Let $R^i_j$ and $R^i_{j+1}$ be the relations that occur in $\formel{disk-iso}'_{i,j}$ for the colours of
    $u^i_j$ and $u^i_{j+1}$, respectively.
    
    To relativise $\formel{disk-iso}'_{i,j}$ to $I(\hat \CQ_{i,j})$, we
    replace every $\exists^{\geq p} x \varphi$ with
    $\exists^{\geq p} x (\formel{int-vert}(x_j^i,x_{j+1}^i,x)
    \wedge \varphi)$ and every $E(x,y)$ with
    $\formel{int-edge}(x_j^i,x_{j+1}^i,x)$. Furthermore, we replace
    every $R^i_j(z,z)$ with the formula $z = x^i_j$ and every
    $R^i_{j+1}(z,z)$ with $z = x^i_{j+1}$. By Lemma \ref{lem:disk-non-def}, the resulting formula
    $\formel{disk-iso}_{i,j}(x^i_j,x^i_{j+1})$ is in $\LW[7]$.
    
    Again using Lemma \ref{lem:disk-non-def}, it is tedious but
    straightforward to construct a $\LW[7]$-formula
    $\formel{disk-chain}(x^0,x^1,x^2)$ which checks if Items
    (1\ref{item:1b}), (1\ref{item:1c}) and (1\ref{item:1d}) hold.

    Now we can just let
    \begin{align*}
      \formel{inside-iso}(x^0,x^1,x^2) \coloneqq {}&\formel{nl-iso}(x^0,x^1,x^2) \wedge \formel{disk-chain}(x^0,x^1,x^2) \wedge {}
      \\&\bigwedge\limits_{i=0}^2 \bigwedge\limits_{j=1}^{n_i-1} \exists z \exists z'\big(\formel{csps-art}_j(x^i,x^{i+1},z) \wedge {}
       \\&\formel{csps-art}_{j+1}(x^i,x^{i+1},z') \wedge \formel{disk-iso}_{i,j}(z,z')\big).
    \uenda
    \end{align*}
  \end{claim}

  In the following, we assume without loss of generality that
  $\hat I \cong I$.
  
  Since $C(\CQ_{i,j})$ is a cycle, the two sets of vertices of the
  segments on $C(\CQ_{i,j})$ between $u^i_j$ and $u^i_{j+1}$ form
  blocks of the automorphism group of
  $C(\CQ_{i,j})_{u^i_j,u^i_{j+1}}$ and thus by Corollary
  \ref{cor:Q-invariance}, also of the automorphism group of
  $G_{u_j^i, u_{j+1}^i}$.
  To be more precise, every automorphism of $G$ that fixes $u_j^i$ and
  $u_{j+1}^i$ either leaves each of the two segments invariant or
  ``swaps sides'', i.e., maps the two segments onto each other while
  preserving heights. Moreover, there is such an automorphism swapping
  sides in $I(\CQ_{i,j})$ if and only if there is a vertex
  $v \in C(\CQ_{i,j})$ with $v \notin \{u^i_j,u^i_{j+1}\}$ whose
  orbit of the automorphism group of $G_{u_j^i, u_{j+1}^i}$ has size
  greater than 1. (In that case, it has size 2.)

  Recall the definition of the cut graph $\Cut(\CB)$
  from Definition~\ref{def:cut}. Also recall Lemma~\ref{lem:necklace-def2},
  where we introduced $\LW[7]$-formulae
  $\formel{nl-cut-vert}(x^0,x^1,x^2,y)$,
  $\formel{nl-cut-edge}(x^0,x^1,x^2,y_1,y_2)$ defining the vertex
  set and edge set of the cut graph. We define a
  \emph{pseudo-cut graph} $\Cut(\hat \CB)$ of $\hat \CB$ by letting
  \begin{align*}
    V\big(\Cut(\hat \CB)\big) &\coloneqq {}
  \formel{nl-cut-vert}[\hat G ,\hat u^0,\hat u^1,\hat u^2,y],\\
    E\big(\Cut(\hat \CB)\big) &\coloneqq {}
  \formel{nl-cut-edge}[\hat G ,\hat u^0,\hat u^1,\hat u^2,y_1,y_2].
  \end{align*}
    Let $C(\hat \CQ_{i,j})$ be the graph with $V\big(C(\hat \CQ_{i,j})\big) = \formel{bd-vert}[\hat G ,u_j^i,u_{j+1}^i,y]$ and $E\big(C(\hat \CQ_{i,j})\big) = \formel{bd-edge}[\hat G ,u_j^i,u_{j+1}^i,y_1,y_2]$,
  where $\formel{bd-vert}(x,x',y)$,
  $\formel{bd-edge}(x,x',y_1,y_2)$ are the
  $\LW[7]$-formulae from
  Lemma~\ref{lem:disk-non-def}. 
  Furthermore, let $I^*(\hat \CQ_{i,j})$ be the graph resulting from
  $I(\hat \CQ_{i,j})_{u^i_j, u^i_{j+1}}$ by assigning all vertices in
  $V\big(C(\hat \CQ_{i,j})\big)$ a common distinct colour and proceeding similarly
  with $E\big(C(\hat \CQ_{i,j})\big)$. Define the graph $I^*(\CQ_{i,j})$
  similarly.

  Let $\Cut^*(\CB)$ be the (coloured) graph resulting from
  $\Cut(\CB)$ by adding colours corresponding to the following unary
  and binary relations:
  \begin{enumerate}[({2}a)]
  \item for each set $J \subseteq [|\art(\CB)|]$ a relation $R_J$
    with $v \in R_J$ if and only if
    $J = \{i \mid vw \in E(G) \text{ for a } w \in
    \formel{nl-art}_i[G,u^0,u^1,u^2,y]\}$, where
    $\formel{nl-art}_i(x^0,x^1,x^2,y)$ is the $\LW[4]$-formula
    introduced in Corollary~\ref{cor:necklace-def1},\label{item:2a}
  \item for each $i \in \{0,1,2\}$ and each $j \in [n_i-1]$ a relation
    $R^1_{i,j}$ with $v \in R^1_{i,j}$ if and only if
    $v \in \formel{bd-vert}[G, u_j^i, u_{j+1}^i, y]$,\label{item:2b}
  \item for each $i \in \{0,1,2\}$ and each $j \in [n_i-1]$ a relation
    $R^2_{i,j}$ with $e \in R^2_{i,j}$ if and only if
    $e \in \formel{bd-edge}[G, u_j^i, u_{j+1}^i,
    y_1,y_2]$,\label{item:2c}
  \item for every orbit $O$ of the automorphism group of
    $I^*(\CQ_{i,j})$ a relation $R_O$ with $v \in R_O$ if and only
    if $v \in O$.\label{item:2d}
  \end{enumerate}
  We show that all of the relations introduced in Items
  (2\ref{item:2a})--(2\ref{item:2d}) can be defined in $\LW[s+3]$
  by providing formulae that express containment in the relations. We
  omit the correctness proofs since they are straightforward.
  \begin{enumerate}[({3}a)]
  \item For $R \coloneqq R_J$, let
    \begin{align*}
      \varphi_R(x^0,x^1,x^2, y) \coloneqq {}&\formel{nl-cut-vert}(x^0,x^1,x^2,y) \wedge \bigwedge\limits_{i
      \in J} \exists z \big(\formel{nl-art}_i(x^0,x^1,x^2,z) \wedge
    E(y,z)\big)\\
      &\wedge \bigwedge\limits_{i \in [|\art(\CB)|] \setminus
      J} \neg\exists z \big(\formel{nl-art}_i(x^0,x^1,x^2,z) \wedge
    E(y,z)\big).
    \end{align*}
       
  \item For $R \coloneqq R^1_{i,j}$, let
    \begin{align*}
      \varphi_R(x^0,x^1,x^2,y) \coloneqq {}
    &\formel{nl-cut-vert}(x^0,x^1,x^2,y)\wedge \exists x_j^i \exists
    x_{j+1}^i \big(\formel{csps-art}_j(x^i,x^{i+1},x_j^i)\\&\wedge
    \formel{csps-art}_{j+1}(x^i,x^{i+1},x_{j+1}^i) \wedge
    \formel{bd-vert}(x_j^i,x_{j+1}^i,y)\big).
\end{align*}
  \item For $R \coloneqq R^2_{i,j}$, let
    \begin{align*}
      \varphi_R(x^0,x^1,x^2,y_1,y_2) \coloneqq {}
    &\formel{nl-cut-edge}(x^0,x^1,x^2,y_1,y_2) \wedge {}
    \\&\exists x_j^i
    \exists x_{j+1}^i \big(\formel{csps-art}_j(x^i,x^{i+1},x_j^i) \wedge
    \formel{csps-art}_{j+1}(x^i,x^{i+1},x_{j+1}^i) 
    \\ &\wedge
    \formel{bd-edge}(x_j^i,x_{j+1}^i,y_1,y_2)\big).
    \end{align*}

  \item Let $R \coloneqq R_O$. By Proposition \ref{prop:corrdetorbits} and the correspondence from Corollary \ref{cor:correspondence},
    since $\LC[4]$ identifies $I^*(\CQ_{i,j})$, the logic $\LC[5]$ determines orbits on $I^*(\CQ_{i,j})$. Thus, there is
    a $\LW[5]$-formula $\varphi'_R(x)$ such that for any graph $H$ it
    holds that $H \models \varphi'_R(v)$ if and only if there is an
    isomorphism $\pi^*$ from $I^*(\CQ_{i,j})$ to $H$ such that
    $\pi^*(w) = v$ for some $w \in O$. 

    We relativise $\varphi'_R(x)$ to $I(\CQ_{i,j})$ by replacing
    every occurrence of the form $\exists^{\geq p}x \psi$ with
    \[
      \exists^{\geq p}x \exists z \exists z'
    (\formel{csps-art}_{j}(x^i,x^{i+1},z) \wedge
    \formel{csps-art}_{j+1}(x^i,x^{i+1},z') \wedge
    \formel{int-vert}(z,z',x) \wedge \psi)
  \] 
  and proceeding similarly
    for the edges.

    Let $R^1_C$ and $R^2_C$ be the colour relations corresponding to
    $V\big(C(\CQ_{i,j})\big)$ and $E\big(C(\CQ_{i,j})\big)$ in $I^*(\CQ_{i,j})$,
    respectively. We replace $R^1_C(z,z)$ with
    \[
      \exists y \exists y' (\formel{csps-art}_{j}(x^i,x^{i+1},y) \wedge
    \formel{csps-art}_{j+1}(x^i,x^{i+1},y') \wedge
    \formel{bd-vert}(y,y',z)
  \] 
  and proceed similarly with
    $R^2_C(z,z')$. Recall that $R^i_j$ and $R^i_{j+1}$ are the colour relations for $u^i_j$ and $u^i_{j+1}$, respectively. We replace $R^i_j(z,z)$ with the formula
    $\formel{csps-art}_j(x^i,x^{i+1},z)$ and
    do the analogous for $R^i_{j+1}(z,z)$. By Corollary \ref{cor:necklace-def1} and Lemma
    \ref{lem:disk-non-def}, the
    resulting formula $\varphi_R(x^0,x^1,x^2,x)$ is in
    $\LW[7]$. \label{item:3d}
  \end{enumerate}

  Our assumption $\hat \CB \cong \CB$ implies that $|\art(\hat \CB)| =
  |\art(\CB)|$. Define $\Cut^*(\hat \CB)$ as the graph resulting from
  $\Cut(\hat \CB)$ by interpreting each relation $R$ from Items 2\ref{item:2a},
  (2\ref{item:2b}), (2\ref{item:2d}) as $\varphi_R[\hat G ,\hat u^0,\hat u^1,\hat u^2,y]$
  and each $R$ from Item (2\ref{item:2c}) as
  $\varphi_R[\hat G ,\hat u^0,\hat u^1,\hat u^2,y_1,y_2]$.

  \begin{claim}[resume]\label{ns:claim3}
    $\Cut^*(\CB) \cong \Cut^*(\hat \CB)$ if
    and only if
    $G_{u^0,u^1,u^2} \cong \hat G _{\hat u^0,\hat u^1,\hat u^2}$.
    
    \proof We prove the backward direction first. Assume that
    $G_{u^0,u^1,u^2} \cong \hat G _{\hat u^0,\hat u^1,\hat u^2}$, and let $\pi$
    be an isomorphism from $G$ to $\hat G $ mapping $u^i$ to $\hat u^i$.
    Since $\CB \cong \hat \CB$, for each $i,j$ the isomorphism
    $\pi$ maps $\CQ_{i,j}$ to
    $\hat \CQ_{i,j}$. By Lemma \ref{lem:necklace-def2}, it holds that
    $V\big(\Cut(\CB)\big) = \formel{nl-cut-vert}[G, u^0,u^1,u^2,y]$
    and
    $E\big(\Cut(\CB)\big) =
    \formel{nl-cut-edge}[G,u^0,u^1,u^2,y_1,y_2]$. Thus, $\pi$
    must map $\Cut(\CB)$ to $\Cut(\hat \CB)$. To see that $\pi$ also
    induces an isomorphism between $\Cut^*(\CB)$ and
    $\Cut^*(\hat \CB)$, consider a vertex $v \in V\big(\Cut^*(\CB)\big)$ and
    suppose $\pi(v)$ has a different colour (i.e.\ satisfies different
    colour relations) in $\Cut^*(\hat \CB)$ than $v$ in $\Cut^*(\CB)$.

    Let $R$ be one of the unary relations from Items
    (2\ref{item:2a})--(2\ref{item:2d}). Since $\pi$ is an isomorphism which
    maps $u^i$ to $\hat u^i$ for $i \in \{0,1,2\}$, we have that
    \begin{align*}
      v \in R(G) &\iff v \in \varphi_R[G,u^0,u^1,u^2,y] 
      \\&\iff \pi(v) \in \varphi_R[\hat G ,\hat u^0,\hat u^1,\hat u^2,y] 
      \\&\iff \pi(v) \in R(\hat G).
    \end{align*}
    Similarly, we can show that every edge $e \in E\big(\Cut(\CB)\big)$ is
    mapped to an edge $\pi(e) \in E\big(\Cut(\hat \CB)\big)$ contained in the
    same colour relations. Thus,
    $\pi$ induces an isomorphism between $\Cut^*(\CB)$ and
    $\Cut^*(\hat \CB)$, which concludes the backward direction of the proof.

    For the forward direction, assume that
    $\Cut^*(\CB) \cong \Cut^*(\hat \CB)$ via an isomorphism $\pi$. Then
    since $|\art(\CB)| = |\art(\hat \CB)|$, by Items (2\ref{item:2a}) and
    (2\ref{item:2b}), the isomorphism $\pi$ can be extended to an
    isomorphism $\pi'$ from the graph with vertex set
    $V\big(\Cut(\CB)\big) \cup \art(\CB)$ and whose edge set is the extension of
    $E\big(\Cut(\CB)\big)$ by all edges between
    $V\big(\Cut(\CB)\big) \cup \art(\CB)$ and $\art(\CB)$ to the
    corresponding subgraph of $\hat G $, where $\pi'$ maps every
    articulation vertex of $\CB$ to one of equal height in $\hat \CB$.

    Furthermore, by Items (2\ref{item:2b}) and (2\ref{item:2c}), the mapping $\pi$ induces an isomorphism from $C(\CQ_{i,j})$ to $C(\hat \CQ_{i,j})$ which fixes $u_j^i$ and $u_{j+1}^i$ (and thus preserves heights). 
    Let $L_{i,j}$ and $R_{i,j}$ as well as $\widehat{L}_{i,j}$ and $\widehat{R}_{i,j}$ be the two segments of $C(\CQ_{i,j})$ and $C(\hat \CQ_{i,j})$ between $u_j^i$ and $u_{j+1}^i$, respectively. Then for every pair $(i,j)$, the isomorphism $\pi$ either maps $L_{i,j}$ to $\widehat{L}_{i,j}$ and $R_{i,j}$ to $\widehat{R}_{i,j}$, or $L_{i,j}$ to $\widehat{R}_{i,j}$ and $R_{i,j}$ to $\widehat{L}_{i,j}$. Without loss of generality, assume the first. 

    If for every pair $(i,j)$, the coloured graphs $I^*(\CQ_{i,j})$ and $I^*(\hat \CQ_{i,j})$ are isomorphic via an isomorphism $\pi_{i,j}'$ mapping $L_{i,j}$ to $\widehat{L}_{i,j}$, then the collection of the $\pi_{i,j}$ clearly extends $\pi$ to an isomorphism between $G$ and $\hat G $. 

    Thus, suppose there is a pair $(i,j)$ such that every isomorphism from $I^*(\CQ_{i,j})$ to $I^*(\hat \CQ_{i,j})$ swaps sides. Let $v \in C(\CQ_{i,j})$ and let $O$ be the orbit of $v$ with respect to the automorphism group of $I^*(\CQ_{i,j})$. Let $R \coloneqq R_O$. Then, by Item (3\ref{item:3d}), it holds that $v \in \varphi_R[G,u^0,u^1,u^2,x]$ but $\pi(v) \notin \varphi_R[\hat G ,\hat u^0,\hat u^1,\hat u^2,x]$. However, this is a contradiction since $\pi$ must respect all relations.
      \uend
    \end{claim}

    Thus, to check whether $G_{u^0, u^1, u^2} \cong \hat G _{\hat u^0, \hat u^1, \hat u^2}$ it suffices to consider the (pseudo-)cut graphs of $G$ and $\hat G $.

    \begin{claim}[resume]\label{ns:claim4}
      There is a formula $\formel{cut-iso}_G(x^0,x^1,x^2) \in
        \LW[s+3]$ (not depending on $\hat G $) such that $\hat G  \models \formel{cut-iso}(\hat u^0,\hat u^1,\hat u^2)$ if and only if $\Cut^*(\hat \CB) \cong \Cut^*(\CB)$.

    \proof
  By Lemma \ref{lem:necklace-simp}, every connected component of $\Cut(\CB)$ is in $\CE_{g-1}$. Therefore, by Corollary \ref{cor:id-comp} and by the induction assumption there is a sentence $\formel{cut-iso}'_G \in \LC[s]$ which identifies $\Cut^*(\CB)$. By replacing every subformula $\exists^{\geq p} x \varphi$ with $\exists^{\geq p} x (\formel{nl-cut-vert}(x^0,x^1,x^2,x) \wedge \varphi)$ and $E(x,y)$ with $\formel{nl-cut-edge}(x^0,x^1,x^2,x,y)$, we relativise $\formel{cut-iso}'_G$ to the (pseudo-)cut graph. To transform it into a formula working also on the uncoloured cut graph, for every relation $R$ from Items (2\ref{item:2a}), (2\ref{item:2b}), (2\ref{item:2d}), we replace each occurrence $R(z,z)$ with $\varphi_R(x^0,x^1,x^2,z)$ and proceed analogously for every $R$ from Item (2\ref{item:2c}). 

  The resulting formula $\formel{cut-iso}_G(x^0,x^1,x^2)$ is in $\LW[s+3]$ and it holds that
  \begin{align*}
    \hat G  \models \formel{cut-iso}_G(\hat u^0,\hat u^1,\hat u^2) \iff \Cut^*(\hat \CB) \cong \Cut^*(\CB).\uenda
  \end{align*}
    \end{claim}

  Now the formula 
  \[
    \iso G \coloneqq {} \exists x^0 \exists x^1 \exists x^2
    \big(\formel{inside-iso}(x^0,x^1,x^2) \wedge
    \formel{cut-iso}_G(x^0,x^1,x^2)\big) \in \LW[s+3]
  \] 
  identifies $G$. An application of Lemma \ref{lem:width} concludes the proof.
\end{proof}

\subsection{Case 2: Presence of simplifying patches}
\label{subsec:simplifying}

In this section, we still assume that $G$ is polyhedrally embedded in
the surface $\FS$ of Euler genus $g$ and that $n=|G|$ (Assumption
\ref{ass:polyhedral2}), but we replace
Assumption~\ref{ass:no-simplifying} with the following assumption.

\begin{assumption}\label{ass:simplifying}
  $G$ contains a simplifying patch.
\end{assumption}

This case sounds simpler than the first one: we only need to remove a simplifying patch from our
graph. The remaining pieces have smaller Euler genus and thus can be
identified in the logic $\LW[s]$. Thus, all we need to do is colour the
pieces in such a way that we can reconstruct the original graph. The
problem with this line of reasoning is that simplifying patches have a
much more complicated structure than non-simplifying patches. For
example, we cannot define the internal graph of a simplifying patch in the
same way as we did for non-simplifying patches in
Lemma~\ref{lem:disk-non-def}. A consequence is that there is
no easy way to reconstruct the original graph from the graph obtained
by removing a simplifying patch. 

The first lemma handles trivial simplifying patches, so that afterwards we can focus on non-trivial ones.

\begin{lemma}
  If $G$ has a trivial simplifying patch, then there is a sentence $\iso G\in\LW[s+2]$ that identifies $G$.
\end{lemma}

\begin{proof}
  If $G$ has a trivial simplifying patch consisting of an edge $uu'$,
  then each connected component of $G\setminus\{u,u'\}$ is in
  $\CE_{g-1}$ and can be identified by a sentence in $\LW[s]$. From
  these sentences, we can construct a sentence in $\LW[s+2]$ identifying $G$
  (arguing as in the proof of Lemma~\ref{lem:representativity}).
\end{proof}

From now on, we make the following assumption.

\begin{assumption}\label{ass:non-trivial}
  $G$ contains no trivial simplifying patch.
\end{assumption}

Recall the definition of a
\emph{segment} $\CQ[v,v']$ of an sps $\CQ$ from Section \ref{sec:sps-patch-necklace}. A \emph{subpatch} of a
patch $\CQ$ is a segment of $\CQ$ which is a patch, i.\,e., which has no proper articulation
vertices. A patch ${\CQ}$ is a \emph{minimal simplifying patch} if
${\CQ}$ is simplifying and all proper subpatches of ${\CQ}$ are
non-simplifying. We are mostly interested in minimal simplifying patches.

The \emph{internal region} of an sps $\CQ$ is the set
\[
  \FR(\CQ) \coloneqq {}\bigcup_{\Fe\in E(\CQ)}\Fe\quad\cup\bigcup_{\substack{\CQ'\text{
        non-trivial non-simplifying}\\\text{subpatch of
      }\CQ}}\FD(\CQ').
\]
Note that if $\CQ$ is a non-trivial patch, then $\FR(\CQ)\subseteq\FD(\CQ)$, because $\FD(\CQ')\subseteq\FD(\CQ)$ for
every subpatch $\CQ'$ of $\CQ$. The \emph{regional graph} of a patch $\CQ$ is
the graph $J\coloneqq J(\CQ)$ with vertex set $V(J)\coloneqq V(G)\cap\FR(\CQ)$ and
$E(J)\coloneqq\{\Fe\in E(G)\mid \Fe\subseteq\FR(\CQ)\}$. It follows
from Lemma~\ref{lem:ns-abstract} that the graph $J$ only depends on
the abstract graph $G$ and not on the embedding of $G$.
Observe that if $\CQ$ is a non-trivial and non-simplifying patch, then
$\FR(\CQ)=\FD(\CQ)$ and $J(\CQ)=I(\CQ)$.

Our first lemma shows that the regional graph of a patch is definable in $\LW[\max\{7,s+2\}]$.

\begin{lemma}\label{lem:regional-def}
    There are $\LW[\max\{7,s+2\}]$-formulae $\formel{J-vert}(x,x',y)$ and
    $\formel{J-edge}(x,x',y_1,y_2)$ such that for all $u,u'\in V(G)$ the following holds.
  If the canonical sps $\CQ \coloneqq \CQ^G(u,u')$ from $u$ to $u'$ is a 
  patch, 
then for all
  $u,u' \in V(G)$,
  \begin{align*}
    \formel{J-vert}[G,u,u',y]&=V\big(J(\CQ)\big),\\
   \formel{J-edge}[G,u,u',y_1,y_2]&=E\big(J(\CQ)\big).
  \end{align*}
\end{lemma}

\begin{proof}
  We first define a $\LW[4]$-formula $\formel{csps-path}$ such that
  $G \models \formel{csps-path}(u,u',v,w)$ if and only if there is a
  path $Q \in \CQ$ such that $v,w\in V(Q)$ and $\height^{\CQ}(v) \le \height^{\CQ}(w)$.
  We set
  \begin{align*}
    \formel{csps-path}(x,x',z,z') \coloneqq {}&\formel{csps-vert}(x,x',z) \wedge \formel{csps-vert}(x,x',z') \wedge 
                                              \bigvee\limits_{i=0}^{n-1} \bigg(\formel{dist}_{=i}(x,x')\wedge {}
    \\ &\bigvee\limits_{j \leq i} \Big(\formel{dist}_{=j}(x,z) \wedge \bigvee\limits_{k \leq i-j} \big(\formel{dist}_{=k}(z,z') \wedge \formel{dist}_{=i-j-k}(z',x')\big)\Big)\bigg).
  \end{align*} 
  
  Now we can let 
   \begin{align*}
   \formel{J-vert}(x,x',y) \coloneqq {}&\formel{csps-vert}(x,x',y) \vee \exists z \exists z' \big(\formel{csps-path}(x,x',z,z') \wedge {}
   \\&\neg\formel{csps-simplifying}(z,z') \wedge \neg E(z,z') \wedge \formel{int-vert}(z,z',y)\big).
  \end{align*}

   By Corollary \ref{cor:comp-genus3} and Lemma \ref{lem:disk-non-def}, we have $\formel{J-vert} \in \LW[\max\{7,s+2\}]$.

   Similarly, we can define the $\LW[\max\{7,s+2\}]$-formula

  \begin{align*}
    \formel{J-edge}(x,x',y_1,y_2) \coloneqq {}&\formel{csps-edge}(x,x',y_1,y_2) \vee \exists z \exists z' \big(\formel{csps-path}(x,x',z,z') \wedge {}
    \\&\neg\formel{csps-simplifying}(z,z') \wedge \formel{J-vert}(z,z',y_1) \wedge \formel{J-vert}(z,z',y_2) \wedge {}
    \\& E(y_1,y_2)\big).\qedhere
  \end{align*}
\end{proof}

For the remainder of this section, we fix a minimal simplifying patch of $G$. 

\begin{assumption}\label{ass:fixpatch}
  $\CQ \coloneqq \CQ^G(u,u')$ is a minimal (non-trivial) simplifying
  patch of $G$. Furthermore, $\FD\coloneqq\FD(\CQ)$,
  $\FR\coloneqq\FR(\CQ)$, and $J\coloneqq J(\CQ)$.
\end{assumption}

By Lemma \ref{lem:regional-def}, the logic $\LC[\max\{7, s+2\}]$
distinguishes the regional graph $J$ from the remainder of the
graph. Furthermore, since $\CQ$ is simplifying, every connected component of $G \setminus J$ is
contained in $\CE_{g-1}$. We need to branch into two cases once more.

\subsubsection*{Case 2.1:~$G\setminus J$ is connected}
The proof in this case is similar to Case 1, but simpler.  The key
fact is that in this case we have $\FR=\FD$, by
\cite[Lemma~15.4.22]{gro17}(1). As remarked in
Section~\ref{sec:sps-patch-necklace} (after Definition~\ref{def:patch}),
there are paths $Q,Q'\in \CQ$ such that $C\coloneqq Q\cup Q'$ is
a cycle and $\Fbd(\FD)=\FC$. It turns out that the cycle $C$ only
depends on $u$, $u'$ and the abstract graph $G$ (that is, we have
analogues of Lemma~\ref{lem:ns-abstract} and
Corollary~\ref{cor:Q-invariance}). Indeed, our first step in this case
will be to define the cycle $C$ in the logic $\LW[7]$.

Let $A_1,\ldots,A_m$ be the connected components
of the graph $A^* \coloneqq G\setminus J$. Each $A_i$ is embedded in $\FS\setminus\FD$, because each
component of $G\setminus V(\CQ)$ embedded in $\FD$ belongs to
$J$. Hence $N(A^*)\subseteq V(C)$. This implies that the graph
$G/A^*$ obtained from $G$ by identifying all vertices
in $A^*$ is planar. Furthermore, by \cite[Corollary~15.2.7]{gro17}, the graph
$G/A^*$ is 3-connected.

\begin{lemma}\label{lem:1-fibre-boundary}
  There are $\LW[7]$-formulae 
  $\formel{bd-vert}(x,x',y)$,
  $\formel{bd-edge}(x,x',y_1,y_2)$ such that
  \begin{align*}
    \formel{bd-vert}[G,u,u',y]&=V(C),\\
    \formel{bd-edge}[G,u,u',y_1,y_2]&=E(C).
  \end{align*}
\end{lemma}

\begin{proof}
  The proof is completely analogous to the proof of
  Lemma~\ref{lem:disk-non-def}, just redefining the formula
  $\formel{astar}$:
  \[
    \formel{astar}(x,x',y)\coloneqq\neg \formel{J-vert}(x,x',y).
  \]
  Note that in the proof of Lemma~\ref{lem:disk-non-def} we never use that $A^*$ is connected.
\end{proof}

Now we fix one addtional vertex: let $u''$ be the neighbour of $u$ on
the path $Q$. Let $k\coloneqq|C|/2=|Q|-1=|Q'|-1$. 
Using $u''$, we can enumerate the vertices of the cycle $C$
in the cyclic order given by letting $u_0\coloneqq u$, $u_1\coloneqq u''$ , then moving
along $Q$ to $u_{k}\coloneqq u'$, and from there moving backwards
along $Q'$ to the neighbour $u_{2k-1}$ of $u$ on the path $Q'$.

\begin{lemma}
  For $0\le i\le 2k-1$ there is a $\LW[7]$-formula 
  $\formel{bd-vert}_i(x,x',x'',y)$ such that
  $\formel{bd-vert}_i[G,u,u',u'',y]=\{u_i\}$. 
\end{lemma}

\begin{proof}
  We let $\formel{bd-vert}_0(x,x',x'',y)\coloneqq {} (y=x)$,
  $\formel{bd-vert}_1(x,x',x'',y)\coloneqq (y=x'')$, and 
  \[
    \formel{bd-vert}_i\coloneqq {}\neg\formel{bd-vert}_{i-2}(x,x',x'',y)\wedge\exists
    y'\big(\formel{bd-vert}_{i-1}(x,x',x'',y)\wedge
    \formel{bd-edge}(x,x',y',y)\big)
  \]
  for $2\le i\le 2k-1$.
\end{proof}

We are ready to finish this subcase.

\begin{proof}[Proof of Lemma~\ref{lem:main}, Case 2.1]
  Let $\hat G$ be an arbitrary graph.
  We shall prove that if there is no $\LW[s+3]$-formula
  distinguishing $G$ and $\hat G$ then the two graphs are isomorphic.
  
  So assume that there is no $\LW[s+3]$-formula
  distinguishing $G$ and $\hat G$. Then there are vertices $\hat u,\hat u',\hat u''\in V(\hat G)$
  such that for all $\LW[s+3]$-formulae $\phi(x,x',x'')$ we
  have $G\models\phi(u,u',u'')\iff\hat G\models\phi(\hat u,\hat
  u',\hat u'')$. We fix such vertices $\hat u,\hat u',\hat u''$. We
  shall prove that there is an isomorphism from $G$ to $\hat G$
  mapping $u$ to $\hat u$, $u'$ to $\hat u'$, and $u''$ to $\hat u''$.

  Let $\hat J$ be the subgraph of $\hat G$ with
  vertex set $V(\hat J)\coloneqq \formel{J-vert}[\hat G,\hat u,\hat
  u',y]$ and  $E(\hat J)\coloneqq \formel{J-edge}[\hat G,\hat u,\hat
  u',y_1,y_2]$. Similarly, let $\hat C$ be the subgraph of $\hat G$ with
  vertex set $V(\hat C)\coloneqq \formel{bd-vert}[\hat G,\hat u,\hat
  u',y]$ and  $E(\hat C)\coloneqq \formel{bd-edge}[\hat G,\hat u,\hat
  u',y_1,y_2]$. Then $\hat C\subseteq\hat J$ and $\hat C$ is a
  cycle. For every $i\in \{0,\ldots,2k-1\}$, let $\hat
  u_i$ be the unique vertex such that $\hat
  G\models\formel{bd-vert}_i(\hat u,\hat u',\hat u'',\hat u_i)$. Then
  $V(\hat C)=\{u_0,\ldots,u_{2k-1}\}$, and the vertices $\hat u_i$
  appear on $\hat C$ in cyclic order starting from $\hat u_0=\hat u$ and $\hat u_1=\hat u''$. Moreover, $\hat u_k=\hat u'$.

  Now we individualise the vertices $u_i$ in $J$ and $\hat u_i$ in
  $\hat J$ using the same colour. More formally, for every $i$ we introduce
  a new colour relation $R_i$ and define $R_i(J)\coloneqq \{u_i\}$ and
  $R_i(\hat J)\coloneqq \{\hat u_i\}$. We observe that the obtained coloured
  versions of $J$ and $\hat J$ satisfy the same $\LW[4]$-sentences, because
  the vertices $u_i$ and $\hat u_i$ are defined in terms of $u,u',u''$ and
  $\hat u,\hat u',\hat u''$, respectively, and for all
  $\LW[7]$-formulae $\phi(x,x',x'')$ we have $G\models\phi(u,u',u'')\iff\hat G\models\phi(\hat u,\hat
  u',\hat u'')$. Since $J$ is planar and every planar graph is 
  identified by a $\LC[4]$-sentence, it follows that the coloured
  graphs are isomorphic. Hence there is an
  isomorphism $f \colon V(J)\to V(\hat J)$ such that $f(u_i)=\hat u_i$ for
  $0\le i\le 2k-1$.

  We shall extend $f$ to an isomorphism from $G$ to
  $\hat G$. Let $\hat A^*\coloneqq \hat G\setminus\hat J$. Note that
  $N^{\hat G}(\hat A^*)\subseteq V(\hat C)$, because $G$ and thus also
  $\hat G$ satisfies the
  $\LW[7]$-formula 
  \[
  \forall y\forall
  z\big(\neg\formel{J-vert}(x,x',y)\wedge
  \formel{J-vert}(x,x',z)\wedge
  E(y,z)\to\formel{bd-vert}(x,x',z)\big).
  \]
  We colour the graphs $A^*=G\setminus J$ and $\hat A^*$ using new
  colour relations $R_I$ for $I\subseteq \{0,\ldots,2k-1\}$. We let
  $R_I(G)$ be the set of all $v\in V(A^*)$ such that
  $N^G(v)=\{u_i\mid i\in I\}$, and similarly we let $R_I(G)$ be the
  set of all $\hat v\in V(\hat A^*)$ such that
  $N^{\hat G}(\hat v)=\{\hat u_i\mid i\in I\}$. Observe that there
  is a $\LW[7]$-formula $\chi_I(x,x',x'',y)$ such that
  $\chi_I[G,u,u',u'',y]=R_I(G)$ and
  $\chi_I[\hat G,\hat u,\hat u',\hat u'',y]=R_I(\hat G)$. Thus the coloured graphs $A^*$ and $\hat A^*$ satisfy the same
  $\LW[s]$-sentences, because for
  all $\LW[s+3]$-formulae $\phi(x,x',x'')$ we have
  $G\models\phi(u,u',u'')\iff\hat G\models\phi(\hat u,\hat u',\hat
  u'')$,

  As $\CQ$ is simplifying, all connected components of $A^*$ are in
  $\CE_{g-1}$. Hence by Assumption~\ref{ass:induction} and Corollary~\ref{cor:id-comp}, there is a $\LW[s]$-sentence $\iso{A^*}$ that
  identifies $A^*$. As $A^*$ and $\hat A^*$ satisfy the same
  $\LW[s]$-sentences, there is an isomorphism $g$ from $A^*$ to $\hat
  A^*$. The colour relations $R_I$ guarantee that $f\cup g$ is an
  isomorphism from $G$ to $\hat G$.
\end{proof}
 
\subsubsection*{Case~2.2: $G\setminus J$ is disconnected}
In this case, we need to analyse the structure of the graph $J$ in
more detail. Let $\mathring H_1,\ldots,\mathring H_\ell$ be
the connected components of $J\setminus\{u,u'\}$. By the assumption of this case, we have $\ell\ge
2$. For every
$i\in[{\ell}]$, let $H_i\coloneqq J[V(\mathring
H_i)\cup\{u,u'\}]$, and let $\CQ_i$ be the set of all paths
$Q\in\CQ$ such that $Q\subseteq H_i$. Then $\CQ_i$ is a shortest path
system from $u$ to $u'$. We call the $\CQ_i$ the \emph{fibres} of
$\CQ$. Note that $\FH_i\subseteq \FD$ for all $i$. Let
$\FR_i\coloneqq \FR(\CQ_i)$.  Then $\FR=\bigcup_{i=1}^{\ell}\FR_i$ and
$\FR_i\cap\FR_j=\{u,u'\}$ for $i\neq j$. Let
$\Ff_1,\ldots,\Ff_{{\ell}'}$ be the arcwise connected
components
of
$\FD\setminus\FR$. By \cite[Lemma~15.4.22]{gro17}(2), we have
${{\ell}'}={\ell}-1$ and there is a permutation $\pi\in S_{\ell}$ such
that $\Fbd(\Ff_i)\subseteq \FR_{\pi(i)}\cup\FR_{\pi(i+1)}$. Without
loss of generality, we assume that $\pi$ is the identity, that is,
$\Fbd(\Ff_i)\subseteq \FR_{i}\cup\FR_{i+1}$. It is not hard to see
(and shown in the proof of \cite[Lemma~15.4.22]{gro17}) that there are
paths $Q_i'\in\CQ_i$ and $Q_{i+1}\in\CQ_{i+1}$ such that
$C_i\coloneqq Q_i'\cup Q_{i+1}$ is a cycle and
$\Fbd(\Ff_i)=\FC_i$. Let $\Ff_\ell\coloneqq\FS\setminus\FD$. Then there are
paths $Q_1\in\CQ_1$ and $Q'_{\ell}\in\CQ_{\ell}$ such that
$C_{\ell} \coloneqq Q'_{\ell} \cup Q_1$ is a cycle and
$\FC_{\ell}=\Fbd(\FD)=\Fbd(\Ff_{\ell})$. For every $i$ we have
$\Fbd(\FR_i)=\FQ_i\cup\FQ'_i$. But note that $Q_i\cup Q_i'$ is not
necessarily a cycle.

We use the notation introduced in this section so far (that is,
$\FD$, $\FR$, $\Ff_i$, $J$, $H_i$, $\mathring H_i$, $\CQ_i$) throughout
the remainder of this subsection. Moreover, we always use indices 
modulo $\ell$. For example, $H_{\ell+1}$ refers to $H_1$.

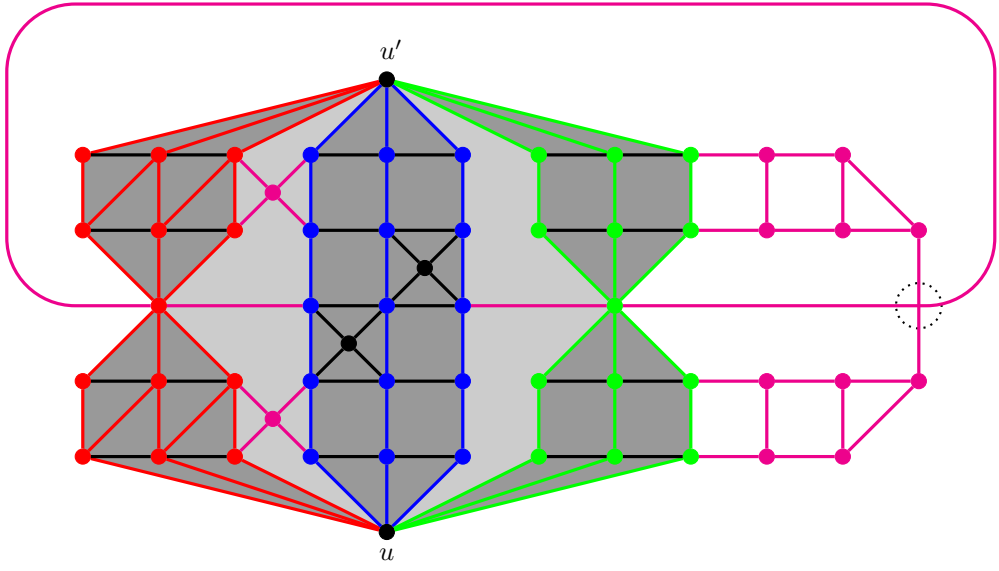
\begin{figure}
  \centering
\begin{tikzpicture}[
  vertex/.style={circle,draw,fill,minimum size=2mm,inner sep=0mm}
  ]

  \fill[black!20] (0,0) -- (-4,1) -- (-4,2) -- (-3,3) -- (-4,4) --
  (-4,5) -- (0,6) -- (4,5) -- (4,4) -- (3,3)  -- (4,2) -- (4,1) -- cycle;
 
  \fill[black!40] (0,0) -- (-4,1) -- (-4,2) -- (-3,3) -- (-4,4) --
  (-4,5) -- (0,6) -- (-2,5) -- (-2,4) -- (-3,3)  -- (-2,2) -- (-2,1) -- cycle;

\fill[black!40] (0,0) -- (-1,1) -- (-1,5) -- (0,6) -- (1,5) --  (1,1) -- cycle;
 
\fill[black!40] (0,0) -- (2,1) -- (2,2) -- (3,3) -- (2,4) --
  (2,5) -- (0,6) -- (4,5) -- (4,4) -- (3,3)  -- (4,2) -- (4,1) -- cycle;

  \draw[dotted,thick] (7,3) circle (3mm);

  \node[vertex] (u) at (0,0) {};
  \node[vertex,red] (v1) at (-4,1) {};
  \node[vertex,red] (v2) at (-3,1) {};
  \node[vertex,red] (v3) at (-2,1) {};
  \node[vertex,red] (v4) at (-4,2) {};
  \node[vertex,red] (v5) at (-3,2) {};
  \node[vertex,red] (v6) at (-2,2) {};
  \node[vertex,red] (v8) at (-3,3) {}; 
  \node[vertex,red] (v10) at (-4,4) {}; 
  \node[vertex,red] (v11) at (-3,4) {}; 
  \node[vertex,red] (v12) at (-2,4) {}; 
  \node[vertex,red] (v13) at (-4,5) {}; 
  \node[vertex,red] (v14) at (-3,5) {}; 
  \node[vertex,red] (v15) at (-2,5) {}; 
  \node[vertex,blue] (w1) at (-1,1) {};
  \node[vertex,blue] (w2) at (0,1) {};
  \node[vertex,blue] (w3) at (1,1) {};
  \node[vertex,blue] (w4) at (-1,2) {};
  \node[vertex,blue] (w5) at (0,2) {};
  \node[vertex,blue] (w6) at (1,2) {};
  \node[vertex,blue] (w7) at (-1,3) {}; 
  \node[vertex,blue] (w8) at (0,3) {}; 
  \node[vertex,blue] (w9) at (1,3) {}; 
  \node[vertex,blue] (w10) at (-1,4) {}; 
  \node[vertex,blue] (w11) at (0,4) {}; 
  \node[vertex,blue] (w12) at (1,4) {}; 
  \node[vertex,blue] (w13) at (-1,5) {}; 
  \node[vertex,blue] (w14) at (0,5) {}; 
  \node[vertex,blue] (w15) at (1,5) {}; 
  \node[vertex,green] (x1) at (2,1) {};
  \node[vertex,green] (x2) at (3,1) {};
  \node[vertex,green] (x3) at (4,1) {};
  \node[vertex,green] (x4) at (2,2) {};
  \node[vertex,green] (x5) at (3,2) {};
  \node[vertex,green] (x6) at (4,2) {};
  \node[vertex,green] (x8) at (3,3) {}; 
  \node[vertex,green] (x10) at (2,4) {}; 
  \node[vertex,green] (x11) at (3,4) {}; 
  \node[vertex,green] (x12) at (4,4) {}; 
  \node[vertex,green] (x13) at (2,5) {}; 
  \node[vertex,green] (x14) at (3,5) {}; 
  \node[vertex,green] (x15) at (4,5) {}; 
  \node[vertex] (uu) at (0,6) {};

  \node[vertex,magenta] (y1) at (5,1) {};
  \node[vertex,magenta] (y2) at (6,1) {};
  \node[vertex,magenta] (y3) at (5,2) {};
  \node[vertex,magenta] (y4) at (6,2) {};
  \node[vertex,magenta] (y5) at (7,2) {};
  \node[vertex,magenta] (y6) at (5,4) {};
  \node[vertex,magenta] (y7) at (6,4) {};
  \node[vertex,magenta] (y8) at (5,5) {};
  \node[vertex,magenta] (y9) at (6,5) {};
  \node[vertex,magenta] (y10) at (7,4) {};

  \node[vertex,magenta] (z1) at (-1.5,1.5) {}; 
  \node[vertex,magenta] (z2) at (-1.5,4.5) {}; 
 
  \node[vertex] (z4) at (0.5,3.5) {}; 
  \node[vertex] (z3) at (-0.5,2.5) {};

  \draw[red,very thick] (u) edge (v1) edge (v2) edge (v3)
  (v1) edge (v4) edge (v5) (v2) edge (v5) edge (v6) (v3) edge (v6)
  (v4) edge (v8) (v5) edge (v8) (v6) edge (v8)
  (v8) edge (v10) edge (v11) edge (v12)
  (v10) edge (v13) edge (v14) (v11) edge (v14) edge (v15) (v12) edge
  (v15)
  (uu) edge (v13) edge (v14) edge (v15)
  ;
  \draw[very thick] (v1) edge (v2) (v2) edge (v3)
  (v4) edge (v5) (v5) edge (v6)
  (v10) edge (v11) (v11) edge (v12)
  (v13) edge (v14) (v14) edge (v15);

    \draw[blue,very thick] (u) edge (w1) edge (w2) edge (w3)
  (w1) edge (w4) (w2) edge (w5) (w3) edge (w6)
  (w4) edge (w7) (w5) edge (w8) (w6) edge (w9)
  (w7) edge (w10) (w8) edge (w11) (w9) edge (w12)
  (w10) edge (w13) (w11) edge (w14) (w12) edge
  (w15)
  (uu) edge (w13) edge (w14) edge (w15)
  ;
  \draw[very thick] (w1) edge (w2) (w2) edge (w3)
  (w4) edge (w5) (w5) edge (w6)
  (w7) edge (w8) (w8) edge (w9)
  (w10) edge (w11) (w11) edge (w12)
  (w13) edge (w14) (w14) edge (w15);

      \draw[green,very thick] (u) edge (x1) edge (x2) edge (x3)
  (x1) edge (x4) (x2) edge (x5) (x3) edge (x6)
  (x4) edge (x8) (x5) edge (x8) (x6) edge (x8)
  (x8) edge (x10) edge (x11) edge (x12)
  (x10) edge (x13) (x11) edge (x14) (x12) edge
  (x15)
  (uu) edge (x13) edge (x14) edge (x15)
  ;
  \draw[very thick] (x1) edge (x2) (x2) edge (x3)
  (x4) edge (x5) (x5) edge (x6)
  (x10) edge (x11) (x11) edge (x12)
  (x13) edge (x14) (x14) edge (x15);

  \draw[magenta,very thick] (x3) edge (y1) (x6) edge (y3) (y1) edge (y2)
  edge (y3) (y2) edge (y4) edge (y5) (y3) edge (y4) (y4) edge (y5)
  (x12) edge (y6) (x15) edge (y8) (y6) edge (y7) edge (y8) (y7) edge
  (y9) edge (y10) (y8) edge (y9) (y9) edge (y10) (y5) edge (y10)
  ;

  \draw[very thick,magenta] (v8) edge (w7) (w9) edge (x8);
  \draw[very thick,magenta,rounded corners=9mm] (x8) -- (8,3) -- (8,7) -- (-5,7)
  -- (-5,3) -- (v8)
  ;
\draw[very thick] (z3) edge (w4) edge (w5) edge (w7) edge (w8);
\draw[very thick] (z4) edge (w8) edge (w9) edge (w11) edge (w12);

  \draw[very thick,magenta] (z1) edge (v3) edge (v6) edge (w1) edge (w4)
  (z2) edge (v12) edge (v15) edge (w10) edge (w13);

  \path (0,-0.3) node {$u$} (0.06,6.4) node {$u'$};

\end{tikzpicture}
  \caption{A simplifying patch in a graph of orientable genus $1$}
  \label{fig:simp-patch}
\end{figure}

\begin{example}\label{exa:ns-patch1}
  Consider the graph shown in Figure~\ref{fig:simp-patch}. The graph
  can be embedded into a torus in such a way that the red, blue, and green paths
    form a simplifying patch $\CQ \coloneqq \CQ^G(u,u')$. The disk $\FD(\CQ)$ is shown in
    grey; the region $\FR(\CQ)$ within $\FD(\CQ)$ in a darker
    grey. The patch has three fibres $\CQ_1,\CQ_2,\CQ_3$ shown in red,
    blue, green, respectively. The boundary cyle of $\FD(\CQ)$ consists
    of the leftmost red path and the rightmost green path from $u$ to
    $u'$. The two areas in light grey are $\Ff_1$ (between red and blue)
    and $\Ff_2$ (between blue and green). The regional graph $J$
    consists of the red, blue, and green paths and all black edges and
    vertices. 

    Observe that the graph has a second, different embedding into the torus in
    which $\CQ$ is still a patch, but the boundary of its disk
    consists of a green and a blue path (and therefore our numbering
    of the fibres would be different; the red fibre would be in the
    middle).
    \uend
\end{example}

 Recall the definition of a bridge from Section \ref{subsec:graphs}. Let $B_1,\ldots,B_m$ be a list of all $J$-bridges in $G$. If $|B_j|\ge 3$, let $A_j$ be the connected component
 of $G\setminus \FR$ associated with $B_j$. If $B_j$ is just a single
 edge, let $A_j$ be the empty graph. In this case, we call $B_j$ \emph{trivial}. Note that for each $j\in[m]$
 there is an $i\coloneqq i(j)\in[\ell]$ such that $B_j$ is embedded in
 $\Ff_{i}$, or more precisely, in $\Fcl(\Ff_{i})$. We say that$B_j$ is \emph{attached to} $H_i$ if it has a vertex of
 attachment in $V(H_i)$. We say that $B_j$ \emph{connects} $H_i$ and
 $H_{i'}$ if it is attached to both $H_i$ and $H_{i'}$. If $B_j$ is
 attached to $H_i$ then it is embedded in $\Ff_i$ or in $\Ff_{i-1}$
 (indices modulo $m$). Thus, if $B_j$ connects $H_{i}$ and $H_{i'}$,
 then either $i'=i+1$ and $B_j$ is embedded in $\Ff_i$, or $i'=i-1$ and
 $B_j$ is embedded in $\Ff_{i-1}$.

\begin{example}\label{exa:ns-patch2}
  Consider again the graph shown in Figure~\ref{fig:simp-patch} with
  the simplifying patch $\CQ$ detailed in Example~\ref{exa:ns-patch1}.
  The graph $J$ (consisting of all red, green, blue, and black
  vertices and edges) has six bridges, all shown in pink. Three of
  these bridges are trivial. Note
  that in this example, all six bridges are planar; the non-planarity of the entire graph is a result of
  combining the bridges.\uend
\end{example}

Observe that there is at most one $i\in[\ell]$ such that there is no
$J$-bridge connecting $H_i$ and $H_{i+1}$. To see this, towards a
contradiction suppose that there are $i$ and $i'$ with $i<i'$ such that there is no
bridge connecting $H_i$ and $H_{i+1}$ and no
bridge connecting $H_{i'}$ and $H_{i'+1}$. Then $\{u,u'\}$
separates $\mathring H_{i}$ from  $\mathring H_{i+1}$, which is
impossible since $G$ is 3-connected. If there is no
bridge connecting $H_i$ and $H_{i+1}$, then we call $\CQ_i$ and
$\CQ_{i+1}$ \emph{dangling fibres}.

We say that two fibres $\CQ_i$ and $\CQ_{i'}$ are \emph{adjacent} if
$|i-i'|=1$ or $\{i,i'\}=\{1,\ell\}$. Note that $\CQ_i,\CQ_{i'}$ are
adjacent if $i\neq i'$ and 
either there is a $J$-bridge that connects $H_i$ and
 $H_{i'}$ or both $\CQ_i$ and $\CQ_{i'}$ are dangling fibres.
This means that we can detect the cyclic adjacency structure on the fibres
$\CQ_i$ just by looking at the bridges connecting them. It follows
that the cyclic order of the fibres only depends on the abstract
graph $G$ and not on its embedding.

\begin{lemma}\label{lem:sameh-def}
  There are $\LW[\max\{7,s+2\}]$-formulae 
  $\formel{same-H}(x,x',y_1,y_2)$, $\formel{bconn-H}(x,x',y_1,y_2)$, $\formel{adj-H}(x,x',y_1,y_2)$
  such that for all $w_1,w_2\in V(G)$ we have:
  \begin{align*}
    G\models\formel{same-H}(u,u',w_1,w_2)&\iff\text{there is an $i\in[\ell]$
      such that $w_1,w_2\in V(H_i)$},\\
    G\models\formel{bconn-H}(u,u',w_1,w_2)&\iff\parbox[t]{8cm}{there
                                            are distinct $i,i'\in[\ell]$ and
                                            $j\in[m]$ such that
                                            $w_1\in V(\mathring H_i)$
                                            and $w_2\in V(\mathring
                                            H_{i'})$ and $B_j$
                                            connects $\CQ_i$ and $\CQ_{i'}$,}\\
     G\models\formel{adj-H}(u,u',w_1,w_2)&\iff\parbox[t]{8cm}{there is an $i\in[\ell]$
      such that $w_1\in V(\mathring{H}_i)$ and $w_2\in V(\mathring{H}_{i-1})\cup V(\mathring{H}_{i+1})$.}
  \end{align*}
\end{lemma} 

\begin{proof}
  By definition of $H_i$, there is an $i\in[\ell]$ such that
  $w_1,w_2\in V(H_i)$ if and only if $w_1,w_2\in V(J)$ and either
  $\{w_1,w_2\}\cap \{u,u'\}\neq\emptyset$ or $w_1$ and $w_2$ belong to
  the same $\mathring{H}_i$.  Thus, we can let
  \begin{align*}
    \formel{same-H}(x,x',y_1,y_2) \coloneqq {}&\formel{J-vert}(x,x',y_1)
  \wedge \formel{J-vert}(x,x',y_2) \wedge {}\\
    & \big(y_1 = x \vee y_1 = x'
  \vee y_2 = x \vee y_2 = x' \vee \phi(x,x',y_1,y_2)\big),
  \end{align*}
  where $\phi(x,x',y_1,y_2)$ is a $\LW[\max\{7,s+2\}]$-formula
  stating that $y_1,y_2$ belong to the same connected component of
  $J\setminus\{u,u'\}$. Using $\formel{J-vert}(x,x',y)$ and
  $\formel{J-edge}(x,x',y_1,y_2)$ as building blocks, it is easy to
  define such a formula.  

  There is a $J$-bridge that connnects fibres $\CQ_i$ and $\CQ_{i'}$
  if and only if there is a path $P\subseteq G\setminus\{u,u'\}$ from
  a vertex $w_i\in V(\mathring H_i)$ to a vertex
  $w_{i'}\in V(\mathring H_{i'})$ with all internal vertices in
  $G\setminus J$.  Let $\psi(x,x',z_1,z_2)$ be a
  $\LW[\max\{7,s+2\}]$-formula such that $G\models\psi(u,u',w_1,w_2)$
  if and only $w_1,w_2\in V(J)\setminus\{x,x'\}$ and there is a path
  from $w_1$ to $w_2$ with all internal vertices in
  $V(G \setminus J)$. We can easily construct such a formula using
  $\formel{J-vert}(x,x',y)$ and $\formel{J-edge}(x,x',y_1,y_2)$ as
  building blocks. Now we let
  \begin{align*}
    \formel{bconn-H}(x,x',y_1,y_2)\coloneqq {} &
    \formel{J-vert}(x,x',y_1) \wedge \formel{J-vert}(x,x',y_2)
                                              \wedge\neg\formel{same-H}(x,x',y_1,y_2) \wedge {}\\
    &\exists z_1\exists
    z_2\left(\bigwedge_{i=1}^2\formel{same-H}(x,x',y_i,z_i)\wedge
      \psi(x,x',z_1,z_2)\right).
  \end{align*}

  Recall that two fibres are adjacent if and only if either there is a
  $J$-bridge that connects them or both fibres are dangling.
  To define dangling fibres, we use the following formula:
  \begin{align*}
    \formel{dangling}(x,x',y)\coloneqq {} &\forall y'\forall
    y''\big(\formel{bconn-H}(x,x',y,y')\wedge
    \formel{bconn-H}(x,x',y,y'')
    \\&\to\formel{same-H}(x,x',y',y'')\big).
  \end{align*}
    Then $G\models \formel{dangling}(u,u',v)$ if and only if
  $v$ belongs to a dangling fibre.

 We let
\begin{align*}
    \formel{adj-H}(x,x',y_1,y_2)\coloneqq {} &
    \formel{bconn-H}(x,x',y_1,y_2)\vee {}
    \\&\left(\neg\formel{same-H}(x,x',y_1,y_2)\wedge\bigwedge_{i=1}^2\formel{dangling}(x,x',y_i)\right).\qedhere
\end{align*}
\end{proof}


\begin{lemma}\label{lem:fibre-k}
  There is a vertex $u''\in V(J)$ and for every $i\in[\ell]$ a $\LW[\max\{7, s+2\}]$-formula $\formel{H-vert}_i(x,x',x'',y)$ such
  that 
   \[
     \formel{H-vert}_i[G,u,u',u'',y] = V(H_i).
    \]
\end{lemma}

Before we prove the lemma, let us remark that $H_i$ is an induced
subgraph of $G$. Therefore there is no need for a formula
$\formel{H-edge}$ defining $E(H_i)$.

\begin{proof}[Proof of Lemma~\ref{lem:fibre-k}]
  It will be easier to define
  formulae $\formel{H-vert}^\circ_i(x,x',x'',y)$ such that
  $\formel{H-vert}^\circ_i[G,u,u',u'',y] = V(\mathring H_i)$. Then we
  let 
  \[
    \formel{H-vert}_i(x,x',x'',y)\coloneqq {}
    \formel{H-vert}^\circ_i(x,x',x'',y)'\vee y=x\vee y=x'.
  \] 
  We let 
  \[
    \formel{same-H}^\circ(x,x',y_1,y_2)\coloneqq {}
    \formel{same-H}(x,x',y_1,y_2)\wedge \bigwedge_{i=1}^2 (y_i \neq x \wedge y_i\neq x').
  \]
  
  If $\ell=2$, we choose an arbitrary $u''\in V(\mathring H_1)$, and
  we let 
  \begin{align*}
    \formel{H-vert}^\circ_1(x,x',x'',y) &\coloneqq {}
                                    \formel{same-H}^\circ(x,x',x'',y),\\
    \formel{H-vert}^\circ_2(x,x',x'',y) &\coloneqq {}
    \formel{J-vert}(x,x',y)\wedge\neg \formel{same-H}(x,x',x'',y).
  \end{align*}
  In the following, we assume $\ell\ge 3$. If there are dangling fibres, we proceed as follows. Suppose $\CQ_{i-1}$
  and $\CQ_{i}$ are dangling. We choose an arbitrary $u''\in
  V(\mathring H_{i})$. 
  We let 
  \begin{align*}
    \formel{H-vert}^\circ_{i}(x,x',x'',y) &\coloneqq {}
                                    \formel{same-H}^\circ(x,x',x'',y),\\
    \formel{H-vert}^\circ_{i+1}(x,x',x'',y) &\coloneqq {}
    \formel{J-vert}(x,x',y)\wedge\formel{bconn-H}(x,x',x'',y),\\
    \intertext{and for $2\le j\le\ell-1$}
    \formel{H-vert}^\circ_{i+j}(x,x',x'',y) \coloneqq {}
                                          &\formel{J-vert}(x,x',y)
                                                \wedge\neg\formel{H-vert}^\circ_{i+j-2}(x,x',x'',y) \wedge {}\\
                                          &\exists y'\big(\formel{H-vert}^\circ_{i+j-1}(x,x',x'',y')\wedge\formel{bconn-H}(x,x',y',y)\big).
  \end{align*}
  In the following, we assume that there are no dangling fibres. For
  $i\in[\ell]$, denote by \emph{$(i,i+1)$-bridge} a $J$-bridge that connects
  $\CQ_i$ and $\CQ_{i+1}$. Since there is no dangling fibre, for all
  $i\in[\ell]$ there is at least one $(i,i+1)$-bridge. 

  Suppose that for some $i$ there is a vertex $v\in V(\mathring H_i)$
  that is a vertex of attachment of an $(i,i+1)$-bridge, but not of an
  $(i-1,i)$-bridge. Then we let $u'' \coloneqq v$. 
  As before,
  $\formel{H-vert}^\circ_{i}(x,x',x'',y) \coloneqq
    \formel{same-H}^\circ(x,x',x'',y)$. To define
    $\formel{H-vert}^\circ_{i+1}(x,x',x'',y)$, we let
    $\psi(x,x',z_1,z_2)$ be a $\LW[\max\{7,s+2\}]$-formula such that $G\models\psi(u,u',w_1,w_2)$
  if and only $w_1,w_2\in V(J)\setminus\{x,x'\}$ and there is a path
  from $w_1$ to $w_2$ with all internal vertices in
  $V(G \setminus J)$. Then we let
    \[
      \formel{H-vert}^\circ_{i+1}(x,x',x'',y)\coloneqq {}\neg\formel{H-vert}^\circ_i(x,x',x'',y) \wedge \exists
      y'\big(\psi(x,x',x'',y') \wedge \formel{same-H}^\circ(x,x',y,y')\big).
    \]
    For $2\le j\le \ell-1$, we define
    $\formel{H-vert}^\circ_{i+j}(x,x',x'',y) $ as above:
    \begin{align*}
\formel{H-vert}^\circ_{i+j}(x,x',x'',y) \coloneqq {}
                                          & \formel{J-vert}(x,x',y)
                                                \wedge\neg\formel{H-vert}^\circ_{i+j-2}(x,x',x'',y) \wedge {}\\
                                          &\exists y'\big(\formel{H-vert}^\circ_{i+j-1}(x,x',x'',y')\wedge\formel{bconn-H}(x,x',y',y)\big).
  \end{align*}
  In the following, we assume that for every $i\in[\ell]$ and every
  $v\in V(\mathring H_i)$, either $v$ is a vertex of attachment of
  both an $(i-1,i)$-bridge and an $(i,i+1)$-bridge (we call $v$
  \emph{doubly-attached}) or it is neither a vertex of attachment of
  an $(i-1,i)$-bridge nor of an
  $(i,i+1)$-bridge (we call $v$ \emph{unattached}). Observe that if
  $v$ is doubly-attached then it is an articulation vertex of the sps
  $\CQ_i$.

  \begin{claim}
    Let $i\in[2,\ell-1]$. Then no vertex $v\in V(\mathring H_i)$ is
    unattached.
    
    \proof
    Suppose towards a contradiction that $v\in V(\mathring H_i)$ is
    unattached. Suppose first that $v\in V(Q)$ for some path $Q\in
    \CQ_i$ from $u$ to $u'$. As we have no dangling fibres, there is at least one
    doubly-attached vertex in $V(\mathring H_i)$. Since all doubly-attached vertices are articulation vertices, every doubly-attached vertex in $V(\mathring H_i)$ appears on the
    path $Q$. Let $w$ be the last doubly-attached vertex on $Q$ before
    $v$, or if no such vertex exists, let $w\coloneqq u$. Let $w'$ be the first doubly-attached vertex on $Q$ after
    $v$, or if no such vertex exists, let $w'\coloneqq u'$. Then
    by our assumption we have $w\neq u$ or $w'\neq u'$. Say, $w'\neq u'$. Now $\{w,w'\}$ separates $v$ from $u'$. This is an easy
    consequence of the fact that the graph that is the union of
    $H_{i-1},H_i,H_{i+1}$, all $(i-1,i)$-bridges, all
    $(i,i+1)$-bridges, and all $J$-bridges that have all their vertices of
    attachment in $H_i$ is embedded in the disk $\FD$. However, this
    contradicts $G$ being 3-connected.

    It remains to consider the case that $v\in V(J)\setminus
    V(\CQ_i)$. Then $v\in V\big(I(\CQ')\big)$ for some non-trivial non-simplifying
    subpatch $\CQ' = \CQ_i[v_1,v_2]$ of $\CQ_i$. Note that $\CQ'$
    contains no articulation vertices of $\CQ_i$ except for (possibly)
    $v_1$ and $v_2$; otherwise it would not be a patch. Thus every
    vertex $v'\in V(\CQ')\setminus\{v_1,v_2\}$ is also unattached. We
    pick such a $v'$. It is contained in a path $Q\in\CQ_i$. Therefore, we can apply the
    same argument as above to $v'$ instead of $v$ and again obtain a
    contradiction.
    \uend
  \end{claim}

  Let $\formel{unattached}(x,x',y)$ be a
  $\LW[\max\{7,s+2\}]$-formula such that $G\models \formel{unattached}(u,u',v)$ if
  and only if $v\in V(J)\setminus\{u,u'\}$ and $v$ is unattached. It
  is straightforward to construct such a formula. Suppose next that there is an unattached vertex $v$. Then $v\in
  V(\mathring H_1)$ or $v\in V(\mathring H_\ell)$. Say, $v\in
  V(\mathring H_\ell)$. Let $u''$ be an arbitrary vertex in
  $V(\mathring H_1)$. 
  We let 
  \begin{align*}
    \formel{H-vert}^\circ_{1}(x,x',x'',y) \coloneqq {}
                                            &\formel{same-H}^\circ(x,x',x'',y),\\
    \formel{H-vert}^\circ_{\ell}(x,x',x'',y) \coloneqq {}
                                               &\formel{J-vert}(x,x',y)\wedge\neg
                                               \formel{H-vert}^\circ_{1}(x,x',x'',y) \wedge {}\\
    &\exists
                                               y'\big(\formel{same-H}^\circ(x,x',y,y')\wedge\formel{unattached}(x,x',y')\big)\\
    \intertext{and for $2\le i\le\ell-1$}
    \formel{H-vert}^\circ_{i}(x,x',x'',y) \coloneqq {}
                                          &\formel{J-vert}(x,x',y)
                                                \wedge\neg\formel{H-vert}^\circ_{i-2}(x,x',x'',y) \wedge  {}\\
                                          &\exists y'\big(\formel{H-vert}^\circ_{i-1}(x,x',x'',y')\wedge\formel{bconn-H}(x,x',y',y)\big).
  \end{align*}
  In the following, we assume that there is no unattached vertex. This
  implies that every $\CQ_i$ consists of just one path $Q_i$. For
  every $i$, let $L_i$ be the graph that is the union of the paths $Q_i$ and
  $Q_{i+1}$ and all $(i,i+1)$-bridges. Observe that $L_1,\ldots
  L_{\ell-1}$ are planar, because they are embedded in the disk
  $\FD$. In fact, all these $L_i$ have a planar embedding where $Q_i\cup
  Q_{i+1}$ is a facial cycle. Note that if $L_\ell$ also has such a
  planar embedding, then the graph $\bigcup_{i=1}^\ell L_\ell$ is planar.

  Suppose $L_\ell$ does not have a planar embedding where $Q_\ell\cup
  Q_1$ is a facial cycle. Using the fact that planarity is
  expressible in $\LC[4]$, we can construct a $\LW[\max\{7,s+2\}]$-formula
  $\formel{planar-L}(x,x',y,y')$ such that $G\models
  \formel{planar-L}(u,u',v,v')$ if and only if for some $i\in[\ell]$ the following
  conditions are satisfied:
  \begin{itemize}
  \item either $v\in V(\mathring H_i)$ and $v'\in V(\mathring H_{i+1})$, or
    $v'\in V(\mathring H_i)$ and $v\in V(\mathring H_{i+1})$;
  \item $L_i$ has a planar embedding where $Q_i\cup Q_{i+1}$ is a facial cycle.
  \end{itemize}
  We choose $u''\in V(\mathring H_1)$, and we let
  \begin{align*}
    \formel{\llap{H-v}ert}^\circ_{1}(x,x',x'',y) \coloneqq {}
                                            &\formel{same-H}^\circ(x,x',x'',y),\\
    \formel{\llap{H-v}ert}^\circ_{\ell}(x,x',x'',y) \coloneqq {}
                                               &\formel{J-vert}(x,x',y) \wedge {}\\
    &\exists
                                               y'\big(\formel{bconn}(x,x',y,y')\wedge\neg\formel{planar-L}(x,x',y,y')\wedge  \formel{same-H}^\circ(x,x',x'',y')\big)\\
    \intertext{and for $2\le i\le\ell-1$, as before,}
    \formel{\llap{H-v}ert}^\circ_{i}(x,x',x'',y) \coloneqq {}
                                          &\formel{J-vert}(x,x',y)
                                                \wedge\neg\formel{H-vert}^\circ_{i-2}(x,x',x'',y) \wedge {} \\
                                          &\exists y'\big(\formel{H-vert}^\circ_{i-1}(x,x',x'',y')\wedge\formel{bconn-H}(x,x',y',y)\big).
  \end{align*}
 In the following, we assume that  $L_\ell$ has a planar embedding where $Q_\ell\cup
  Q_1$ is a facial cycle. Then the graph $L\coloneqq
  \bigcup_{i=1}^\ell L_\ell$ is planar.  
  However, $G$
  is not planar. The graphs $G$ and $L$ differ only in the $J$-bridges that are attached to just a single fibre. Let us call a
  $J$-bridge whose vertices of attachment are in the fibre $\CQ_i$ (and
  thus on the path $Q_i$) an \emph{$i$-bridge}. Due to the 3-connectivity of $G$, every $i$-bridge
  has at least 3 vertices of attachment in $V(Q_i)$. Furthermore, since all vertices of
  $Q_i$ are doubly-attached (i.\,e., they are vertices of attachment of both an $(i-1,i)$-bridge and an $(i,i+1)$-bridge), there is no way to attach an $i$-bridge for some $i \in \{2, \dots, \ell-1\}$ without destroying the embedding in the disk $\FD$. This is easy to see considering the fact that an $i$-bridge embedded in, say, $\Ff_i$ has two vertices $v$, $v'$ of attachment of distance at least two in $\CQ_i$ and it thus ``blocks'' the vertex between $v$ and $v'$ in $V(\CQ_i)$ from being attached to a vertex in $\CQ_{i+1}$ (cf.\ Corollary 9.1.2, \cite{gro17}).
  Hence there can only be $i$-bridges for
  $i=1$ and $i=\ell$, and there must be at least one such bridge,
  because otherwise $G=L$ would be planar. Say, there is an
  $\ell$-bridge. We can easily construct a $\LW[\max\{7,s+2\}]$-formula
  $\formel{self-bridge}(x,x',y)$ such that
  $G\models\formel{self-bridge}(u,u',v)$ if and only if $v\in V(\mathring
  H_i)$ for some $i$ such that there is an $i$-bridge. We choose
  $u''\in V(\mathring H_1)$ and let
  \begin{align*}
    \formel{H-vert}^\circ_{1}(x,x',x'',y) \coloneqq {}
                                            &\formel{same-H}^\circ(x,x',x'',y),\\
    \formel{H-vert}^\circ_{\ell}(x,x',x'',y) \coloneqq {}&\formel{self-bridge}(x,x',y)\wedge\neg\formel{same-H}(x,x',x'',y)\\
    \intertext{and for $2\le i\le\ell-1$, as before,}
    \formel{H-vert}^\circ_{i}(x,x',x'',y) \coloneqq {}
                                          &\formel{J-vert}(x,x',y)
                                                \wedge\neg\formel{H-vert}^\circ_{i-2}(x,x',x'',y) \wedge {}\\
                                          &\exists y'\big(\formel{H-vert}^\circ_{i-1}(x,x',x'',y')\wedge\formel{bconn-H}(x,x',y',y)\big).
  \end{align*}
  This completes the proof.
\end{proof}

In the following, we fix a vertex $u''$ that is chosen according to Lemma \ref{lem:fibre-k}.

\begin{assumption}\label{ass:u1}
$u''\in V(J)$ is a fixed vertex such that
$\formel{H-vert}_i[G,u,u',u'',y] = V(H_i)$ for every $i\in[\ell]$.
\end{assumption}

A $J$-bridge $B$ is an \emph{inner bridge} if it has at least one
vertex of attachment in $\bigcup_{i=2}^{\ell-1}V(\mathring H_i)$. Note
that all inner bridges are embedded in the disk $\FD$. Let $K$ be the union of $J$ with all inner bridges. Then $K$ is a planar
graph embedded in $\FD$. Using Lemma \ref{lem:avoiding-path} for $\varphi \coloneqq \formel{J-vert}(x,x',y)$, we can construct $\LW[\max\{7,s+2\}]$-formulae that define membership in $K$.

\begin{corollary}\label{cor:K}
    There are $\LW[\max\{7,s+2\}]$-formulae $\formel{K-vert}(x,x',x'',y)$, 
    $\formel{K-edge}(x,x',x'',y_1,y_2)$ such that 
  \begin{align*}
    \formel{K-vert}[G,u,u',u'',y]&=V(K),\\
    \formel{K-edge}[G,u,u',u'',y_1,y_2]&=E(K).
  \end{align*}
\end{corollary}

Finally, we are ready to complete the proof of Lemma \ref{lem:main}.

\begin{proof}[Proof of Lemma \ref{lem:main}, Case 2.2] 
  Let us briefly recall our main assumptions for this case: 
  \begin{itemize}
  \item $G$ is a graph of order $n=|G|$ polyhedrally embedded in a
    surface $\FS$ of Euler genus $g\ge 1$.
  \item $\CQ=\CQ^G(u,u')$ is a non-trivial simplifying patch in $G$
    with $\ell\ge 2$ fibres.
  \item $u''$ is a vertex that allows us to identify the fibres of
    $\CQ$ via the formulae of Lemma~\ref{lem:fibre-k}.
  \end{itemize}
  We continue to use the notation introduced in this section, such as
  $\FD$, $J$, $\CQ_i$, $H_i$ and $\mathring H_i$, et cetera.

  Moreover,
  we define 
  \[
    h\coloneqq\dist(u,u').
  \]
  
  Let $\hat G$ be an arbitrary graph.
  We shall prove that if there is no $\LW[s+3]$-formula
  distinguishing $G$ and $\hat G$, then the two graphs are isomorphic.
  
  So assume that there is no $\LW[s+3]$-formula
  distinguishing $G$ and $\hat G$. Then $|\hat G|=n$ and $\hat
  G\not\in\CE_{g-1}$. Furthermore, there are vertices
  $\hat u,\hat u',\hat u''\in V(\hat G)$ such that for all
  $\LW[s+3]$-formulae $\phi(x,x',x'')$ we have
  $G\models\phi(u,u',u'')\iff\hat G\models\phi(\hat u,\hat u',\hat
  u'')$. We fix such vertices $\hat u,\hat u',\hat u''$. We shall
  prove that there is an isomorphism from $G$ to $\hat G$ mapping $u$
  to $\hat u$, $u'$ to $\hat u'$, and $u''$ to $\hat u''$.

  Let $\hat\CQ\coloneqq\CQ^{\hat G}(\hat u,\hat u')$.  We say $\CQ$
  and $\hat \CQ$ are \emph{isomorphic}, and write
  $\CQ \cong \hat \CQ$, if there is an isomorphism from $G(\CQ)$
  to $\hat G (\hat \CQ)$ mapping $u$ to $\hat u$
  and $u'$ to $\hat u'$. Thus, in the following we always regard $u,u',u''$ and the corresponding
  $\hat u,\hat u',\hat u''$ as distinguished vertices that isomorphisms need to respect.

  Just as in the proof of Claim \ref{ns:claim1} of Case 1, we have a formula $\formel{sps-iso}(x,x') \in \LW[6]$ (not
    depending on $\hat G $) such that $\hat G  \models \formel{sps-iso}(\hat u,\hat u')$ if and only
    if $\hat \CQ$ is a pseudo-patch with
    $\hat \CQ \cong \CQ$.

  Hence $\hat\CQ$ is a non-trivial pseudo-patch in $\hat G$. Let $\hat J$ be the graph
  with vertex set
  $V(\hat J)\coloneqq \formel{J-vert}[\hat G ,\hat u,\hat u',y]$
  and edge set
  $E(\hat J)\coloneqq \formel{J-edge}[\hat G ,\hat u,\hat
  u',y_1,y_2]$. Note that $\hat u''\in V(\hat J)$, because
  $u''\in V(J)=\formel{J-vert}[G,u, u',y]$. Therefore, we call $J$ and $\widehat{J}$
  \emph{isomorphic}, and write $J \cong \widehat{J}$, if
  $J_{u,u',u''} \cong \widehat{J}_{{\hat u},{\hat u'},\hat u''}$, that
  is, there is an isomorphism from $J$ to $\hat J$ that maps $u$ to
  $\hat u$, $u'$ to $\hat u'$, and $u''$ to $\hat u''$. Note that
  every such isomorphism induces an isomorphism from $G(\CQ)$ to
  $G(\hat \CQ)$.

  Let $T\subseteq V( J)$ be the set of vertices of attachment of all
  $J$-bridges in $G$, and, similarly, let $\hat T$ be the set of
  vertices of attachment of all $\hat J$-bridges in $\hat
  G$.

  \begin{claim}\label{s:claim2}
    There is a formula $\formel{J-iso}(x,x',x'') \in \LW[\max\{7,s+2\}]$ such that
    $\hat G  \models \formel{J-iso}({\hat u},{\hat u'},{\hat u''})$ if and
    only if $J\cong\widehat{J}$ via an isomorphism that maps $T$ to
    $\hat T$.

    \proof Let $J^*$ be the graph resulting from $J_{u,u',u''}$ by
    assigning all vertices in $T$ a common
    distinct colour (however maintaining the individual colours for $u$, $u'$, $u''$). Since $J$ is planar, the claim follows by relativising a sentence
    $\formel{J-iso}' \in \LW[4]$ which identifies $J^*$ to the
    subgraph whose vertex and edge set is defined by the formula
        $\formel{J-vert}$ and
        $\formel{J-edge}$, respectively, and by
    replacing the colour relations for $u$, $u'$, and $u''$ with
    equations of the form $z = x$, $z = x'$, $z=x''$ and the colour
    relation for $T$ with
    $\exists z' \big(E(z,z') \wedge \neg \formel{J-edge}(x,x',z,z')\big)$. By
    Lemma \ref{lem:regional-def}, the formula has width
    $\max\{7,s+2\}$.  \uend
  \end{claim}

  In the following, we assume without loss of generality that
  $\widehat{J} \cong J$ and we only consider isomorphisms which
  preserve the property of being a vertex of attachment.

  We intend to equip certain supergraphs of $J$ and $\hat J$ with colour relations such that the
  coloured graphs are isomorphic if and only if $G$ and $\hat G $
  are isomorphic via an isomorphism mapping $u$ to ${\hat u}$, $u'$ to
  ${\hat u'}$, and $u''$ to $\hat u''$. Then we show that the coloured
  supergraph of $J$ can be identified in $\LW[s+3]$.

  First recall that for every $J$-bridge its vertices of attachment
  lie on a shortest path from $u$ to $u'$. Thus, by
  Claim~\ref{s:claim2}, we can assume the same for $\hat G $,
  since we can express the sps-containment of a vertex of
  attachment. Hence, each element in $V(J)$ and $V(\widehat{J})$ which
  is a vertex of attachment of a bridge has a well-defined height in
  $\CQ$ and $\hat \CQ$, respectively.

  For $i\in[\ell]$, we let $\widehat{H}_i$ be the induced
  subgraph of $\hat G$ with vertex set
  $\formel{H-vert}_i[\hat G ,{\hat u},{\hat u'},{\hat
    u''},y]$.
  Then $\hat J=\bigcup_{i=1}^\ell\hat H_i$, because there is a $\LW[\max\{7,s+2\}]$-formula which expresses that
  $J=\bigcup_{i=1}^\ell H_i$.

  We now colour vertices in $V(G \setminus J)$ by their ``attachment pattern'' in $J$. For every $v\in V(G \setminus J)$, we let 
  \begin{align}\label{s:conn-iso}
    S(v) \coloneqq {}\llcurly(i,j) \mid w \in N(v)\cap V(H_i), \text{ $w$ has height $j$ in $\CQ$}\rrcurly.
  \end{align}
  That is, for each $w \in N(v) \cap V(H_i)$ of height $j$, the set
  $S(v)$ contains one separate copy of $(i,j)$. Similarly, for $\hat v\in V(\hat G \setminus \hat J)$ we let
  \begin{align*}\label{s:conn-iso-hat}
    \hat S(\hat v) \coloneqq {}\llcurly(i,j) \mid \hat w\in N(\hat
    v)\cap V(\hat H_i), \text{ $\hat w$ has height $j$ in $\hat\CQ$}\rrcurly.
  \end{align*}
  Let $A$ be a connected component of $G \setminus J$. We view $A$ as
  a coloured graph where (in addition to colours that may have already
  been present in $G$) each vertex $v$ is coloured by the multiset
  $S(v)$. Since $\CQ$ is simplifying, $\eg(A)\le g-1$, and thus there is
  a $\LW[s]$-sentence
  $\formel{bridge-iso}'_{A}$ that identifies $A$. We shall transform it
  into a $\LW[s+3]$-formula
  $\formel{bridge-iso}_{A}(x,x',x'',y)$ such that
  $\hat G  \models \formel{bridge-iso}_{A}({\hat u},{\hat
    u'},{\hat u''},\hat v)$ if and only if $A$ is isomorphic to the
  connected component of $\hat v$ in $\hat G  \setminus \widehat{J}$
  via an isomorphism $\pi$ that preseves the attachment patterns, that
  is, $S(v)=\hat S\big(\pi(v)\big)$ for all $v \in V(A)$. 
  
  Note that if a bridge in $G$ is attached to two vertices $w$ and
  $w'$ with the same label pair $(i,j)$, then it must hold that
  $w = w'$. Thus, for any vertex $v$ in $G$, the multiset $S(v)$ is actually a set, i.e., each tuple occurring in the multiset has
  multiplicity 1. However, in $\hat G $ this might not be the
  case.

  To relativise $\formel{bridge-iso}'_A$ to the connected component of
  a vertex $v$ in $G \setminus J$, we use the formula
  $\formel{comp}_\varphi$ from Lemma \ref{lem:avoiding-path} for
  $\varphi \coloneqq \formel{J-vert}(x,x',y)$ and replace every
  $\exists z \psi$ with
  $\exists z(\formel{comp}_\formel{J-vert}(x,x',y,z) \wedge
  \psi)$. Since $\formel{J-vert}(x,x',y)$ has width $\max \{7,s+2\}$,
  Lemma \ref{lem:avoiding-path} yields that
  $\formel{comp}_{\formel{J-vert}}(x,x',y,z) \in \LW[\max\{7,s+2\}]$.

  \newcommand{\mult}{\operatorname{mult}}

  To account for the colours, we define for each multiset $S$ of label
  pairs $(i,j)$ a relation $R_S$ with $v \in R_S$ if and only if
  $S(v)=S$. Note that all label
  pairs that can occur are contained in $[m]\times\{0,\ldots,h\}$. Let us denote the
  multiplicity of a pair $(i,j)$ in a multiset $S$ by $\mult_S(i,j)$.
  We let
  \begin{align*}
    \formel{att-pat}_S(x,x',x'',y) \coloneqq \hspace{-8mm}\bigwedge_{\substack{(i,j) \in [m]\times\{0,\ldots,h\}\\\mult_S(i,j)=p}}\hspace{-8mm} \exists^{=p} z \big(\formel{H-vert}_i(x,x',x'',z) \wedge \formel{csps-height}_j(x,x',z) \wedge E(y,z)\big).
  \end{align*}
  Then $G\models\formel{att-pat}_S(u,u',u'',v)\iff S(v)=S$. Note that
  by Lemma \ref{lem:fibre-k}, we have
  $\formel{att-pat}_S \in \LW[\max\{7,s+2\}]$. We replace
  every $R_S(z,z)$ in $\formel{bridge-iso}'_A$ with
  $\formel{att-pat}_S(x,x',x'',z)$ and obtain the desired formula
  $\formel{bridge-iso}_A(x,x',x'',y)$, which has width
  $s+3$.

   In the following, we only consider $J$-bridges and $\hat J$-bridges. If the reference to $J$ or $\hat J$ is clear from the context, we often do not mention it explicitly and simply use the term ``bridge''.

  Recall that every
  $J$-bridge is either an \emph{$i$-bridge} with all vertices of
  attachment in a single fibre $H_i$ or an \emph{$(i,i+1)$-bridge}
  with vertices of attachment in two adjacent fibres $H_i$ and $H_{i+1}$
  for some $i\in[\ell]$.

  Recall (from the paragraph preceding Corollary~\ref{cor:K}) that an
  \emph{inner bridge} is a $J$-bridge which has at least one vertex of
  attachment in $\bigcup_{i=2}^{\ell-1}V(\mathring H_i)$ and that the
  graph $K$ is the union of $J$ with all inner
  bridges. $K$ is a planar graph embedded in $\FD$. By
  Corollary~\ref{cor:K} we have $\LW[s+3]$-formulae
  $\formel{K-vert}(x,x',x'',y)$ and
  $\formel{K-edge}(x,x',x'',y_1,y_2)$ such that
  $V(K)=\formel{K-vert}[G,u,u',u'',y]$ and
  $E(K)=\formel{K-edge}[G,u,u',u'',y_1,y_2]$.
  We let $\hat K$ be the subgraph of $\hat G$ with vertex set $V(\hat
  K)=\formel{K-vert}[\hat G,\hat u,\hat u',\hat u'',y]$ and edge set
  $E(\hat K)=\formel{K-edge}[\hat G,\hat u,\hat u',\hat u'',y_1,y_2]$.

\newcommand{\Bcrit}{\CB_{\operatorname{crit}}}
\newcommand{\Bsc}{\CB_{\operatorname{sc}}}
\newcommand{\HBcrit}{\hat \CB_{\operatorname{crit}}}
\newcommand{\HBsc}{\hat \CB_{\operatorname{sc}}}

\begin{figure}
  \centering
   \includegraphics[width=0.5\textwidth]{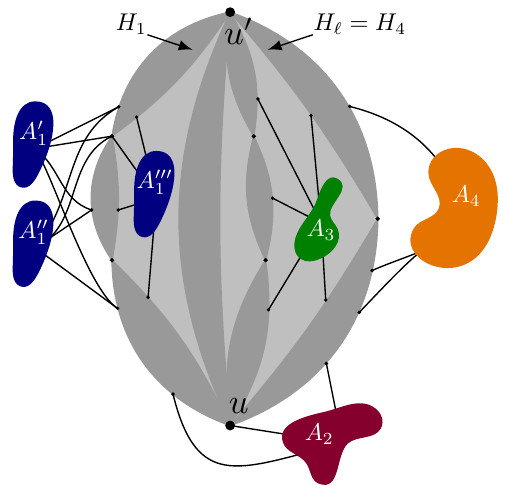}
   \caption{A simplifying patch with an inner bridge (green), multiple critical bridges (blue, orange, red) and a super-critical bridge (blue). Note that without $A''_1$, the graph would have an opposite pair $\{A_1',A_1'''\}$.}\label{fig:critical}
 \end{figure}

  A bridge is \emph{critical} if it is not an inner
  bridge (see Figure \ref{fig:critical}). Observe that a bridge is critical if it is either an
  $\ell$-bridge or a $1$-bridge or an $(\ell,1)$-bridge. Let $\Bcrit$
  denote the set of all critical $J$-bridges. Similarly, let $\HBcrit$
  be the set of all $\hat J$-bridges in $\hat G$ whose vertices of attachment are contained in $V(\hat \CQ_\ell) \cup V(\hat \CQ_1)$ (where $\hat \CQ_i$ is the set of all paths
$\hat Q\in\hat\CQ$ such that $\hat Q\subseteq \hat H_i$).
  With each $J$-bridge
  $B\in\Bcrit$ we associate a \emph{type} $\theta(B)$ as follows.
  \begin{itemize}
  \item If $B$ consists of a single edge $vv'$, then
    $\theta(B)=\{(i,j),(i',j')\}$, where $(i,j)=(0,0)$ if $v=u$,
    $(i,j)=(0,h)$ if $v=u'$, and otherwise $v\in V(\mathring H_i)$ and
    $j=\dist(u,v)$, and similarly $(i',j')=(0,0)$ if $v'=u$,
    $(i',j')=(0,h)$ if $v'=u'$, and otherwise
    $v'\in V(\mathring H_{i'})$ and $j'=\dist(u,v')$.
  \item If $|B|\ge 3$, let $A \coloneqq B\setminus V(J)$ be the connected
    component of $G\setminus J$ associated with $B$. We view $A$ as a
    coloured graph (with colours representing the attachment patterns
    $S(v)$ as above) and choose a label $\theta_A$ for the isomorphism
    type of $A$ (in such a way that $\theta_A=\theta_{A'}\iff A\cong A'$). We let
    $\theta(B)\coloneqq\{\theta_A\}$.
  \end{itemize}
  We can define the type $\hat\theta(B)$ of a $\hat J$-bridge
  $\hat B\in\HBcrit$ similarly. Observe that there is a
  bijection $\beta \colon \Bcrit\to\HBcrit$ such that
  $\theta(B)=\hat\theta\big(\beta(B)\big)$ for all $B\in\Bcrit$. To see this,
  note that we can use the formulae $\formel{bridge-iso}_A$ to construct
  for every type $\theta$ a $\LW[s+3]$-formula that encodes the number
  of bridges of type $\theta$.

  Observe that if two critical bridges have the same type, then either
  both are $(\ell,1)$-bridges or both are $1$-bridges or both are
  $\ell$-bridges.
  
  Recall that $\at(B)$ denotes the set of vertices of attachment of a
  bridge $B$. Let us call bridges $B$, $B'$ \emph{aligned} if
  $\at(B)=\at(B')$. We show that being aligned can be defined in $\LW[\max\{7,s+2\}]$. Let
  \[\formel{att-vert}(x,x',y,z) \coloneqq \formel{J-vert}(x,x',z) \wedge \exists z' \big( E(z,z') \wedge \formel{comp}_{\formel{J-vert}}(x,x',z',y)\big).\]
  Then if $v \notin J$, it holds that $G \models \formel{att-vert}(u,u',v,w)$ if and only if $w$ is a vertex of attachment of some $J$-bridge that contains $v$. By Lemmas \ref{lem:avoiding-path} and \ref{lem:regional-def}, the formula $\formel{att-vert}$ has width $\max\{7,s+2\}$. Now we can define
  \begin{align*}
    \formel{aligned}(x,x',y,y') \coloneqq {} &\neg \formel{comp}_{\formel{J-vert}}(x,x',y,y') \wedge {}
    \\&\forall z \big(\formel{att-vert}(x,x',y,z) \leftrightarrow \formel{att-vert}(x,x',y',z)\big).
  \end{align*}
  Then if $v,v' \notin J$, it holds that $G \models \formel{aligned}(u,u',v,v')$ if and only if $v$ and $v'$ are contained in distinct $J$-bridges $B$ and $B'$, respectively, and $B$ and $B'$ are aligned. Furthermore, $\formel{aligned}$ has width $\max\{7,s+2\}$ since $\formel{att-vert}$ has width $\max\{7,s+2\}$. The two formulae can easily be modified to also capture the case that $v$ or $v'$ itself is a vertex of attachment (and the case of trivial bridges, but we do not need this for our purposes).

  Recall that for all critical bridges $B\in\Bcrit$
  we have $\at(B)\subseteq V(\CQ_1)\cup V(\CQ_\ell)$. Let
  \begin{align*}
    Z&\coloneqq\art(\CQ_1)\cup\art(\CQ_\ell),\\
    \hat Z&\coloneqq\art(\hat\CQ_1)\cup\art(\hat\CQ_\ell).
  \end{align*}

  \begin{claim}[resume]\label{claim:cbridge}
    Let $B,B'\in\Bcrit$ such that $\theta(B)=\theta(B')$.
    \begin{enumerate}
    \item If $B$ and $B'$ are not aligned, then
      \[
        \at(B)\cap\at(B')=\at(B)\cap Z=\at(B')\cap Z,
      \]
      and $B$ or $B'$ is embedded in $\FD$.\label{it:1:cbridge}
    \item At most one of $B,B'$ is embedded in $\FD$. \label{it:2:cbridge}
    \end{enumerate}

    \proof Let $\theta\coloneqq\theta(B)=\theta(B')$. Suppose that
    $\at(B)=\{v_1,\ldots,v_r\}$ and
    $\at(B')=\{v_1',\ldots,v_{r'}'\}$. Since the type of a bridge
    contains information about the number of vertices of attachment,
    their fibres, and their height, we have $s=s'$, and without loss
    of generality we may assume that for every $i$ the vertices $v_i$
    and $v_i'$ belong to the same fibre and have the same height in
    this fibre. Thus, $v_i$ is an articulation vertex of its fibre if and only if $v_i'$ is one (and in this case they are equal). Hence $\at(B)\cap Z =\at(B')\cap Z$ and therefore we have
    \[\at(B)\cap Z = \at(B) \cap \at(B')\cap Z \subseteq \at(B) \cap \at(B').\]
    Furthermore, for every fibre $i$ and every height $j$ there are at
    most two vertices $v,v'\in V(H_i)$ of height $j$ that may be
    vertices of attachment of a bridge, one on the path $Q_i$ and one on
    the path $Q_i'$. It follows that every vertex in $\at(B)\cap \at(B')$ lies in $V(Q_i) \cap V(Q_i')$. Hence, every path from $u$ to $u'$ in $\CQ_i$ passes through $v$ and therefore, $v$ is an articulation vertex of $\CQ_i$. This proves $\at(B) \cap \at(B') \subseteq \at(B) \cap Z$ and hence equality.
    
   This means that if $B$
    and $B'$ are not aligned, one of them has some vertices of attachment in $V(Q_i) \setminus V(Q_i')$
    and the other has some vertices of attachment in $V(Q_i') \setminus V(Q_i)$. Thus, one must be embedded in $\Ff_{i-1}$
    and one in $\Ff_i$. At least one of these sets is a subset of $\FD$.
     
    Since $G$ is 3-connected, we have $r\ge 3$, and this means we
    cannot embed both $B$ and $B'$ into $\FD$, because this would
    violate planarity (reasoning via a $K_{3,3}$-minor).
    \uend
  \end{claim}

  Note that for a connected component of $G \setminus J$ or $\hat G \ setminus \hat J$, there is only a bounded number of possible isomorphism types $\theta_A$. Thus, we can check whether two bridges have the same isomorphism type using the formulae $\formel{bridge-iso}_A$. Hence, employing the formulae $\formel{bridge-iso}_A$, $\formel{aligned}$, $\formel{att-vert}$ and requiring the variable $z$ in the definition of $\formel{csps-art}$ (cf.\ Lemma \ref{lem:csps-def}) additionally to be in $V(H_1)$ and $V(H_\ell)$, respectively, we can show that all restrictions the claim imposes on $G$ are definable in $\LW[s+3]$. Thus, for all $\hat B,\hat B'\in\HBcrit$, if
  $\hat B$
  and $\hat B'$ are not aligned, then
  $\at(\hat B)\cap\at(\hat B')=\at(\hat B)\cap \hat Z=\at(\hat
  B')\cap\hat Z$.

  There is an interesting special case of pairs of critical bridges that
  we need to deal with separately. Consider a type $\theta$ such that
  there are exactly two bridges $B,B'\in\Bcrit$ of type $\theta$, and
  these two bridges are not aligned. Then by
  Claim~\ref{claim:cbridge}, either both $B$ and $B'$ are
  $\ell$-bridges or both are $1$-bridges. Moreover, the claim implies
  that exactly one of them is embedded in
  $\Ff_1\cup\Ff_{\ell-1}\subseteq\FD$. Say, $B$ is embedded in
  $\Ff_1\cup\Ff_{\ell-1}$. We call $\{B,B'\}$ an opposite
    pair. That is, an \emph{opposite pair} in $G$ is an unordered pair
  $\{B,B'\}$ of bridges $B,B'\in\Bcrit$ such that $\theta(B)=\theta(B')$, there is
  no $B''\in\Bcrit\setminus\{B,B'\}$ such that
  $\theta(B'')=\theta(B)$, and $B$,
  $B'$ are not aligned.

  \begin{claim}[resume]\label{claim:opposite}
    Let $\{B_1,B_1'\},\ldots,\{B_p,B_p'\}$ be a list of all opposite
    pairs of $G$. Then the graph
    \[
      K^+\coloneqq {}K\cup\bigcup_{i=1}^p(B_i\cup B_i')
    \]
    is planar.

    \proof Recall that for every $i$ either $B_i$ or $B_i'$ is
    embedded in $\FD$. Without loss of generality we assume that for
    all $i$, the bridge $B_i$ is embedded in $\FD$. Then
    $K\cup\bigcup_{i=1}^p B_i$ is planar, because it is embedded in
    $\FD$. Moreover, for every $i$ the bridges $B_i$ and $B_i'$ are
    isomorphic and have vertices of attachment in the same fibres and
    of the same height. We can just copy the embedding of $B_i$ and
    embed $B_i'$ in the same way outside of $\FD$.  \uend
  \end{claim}

  An \emph{opposite pair} in $\hat G$ is an unordered pair
  $\{\hat B,\hat B'\}$ of bridges $\hat B,\hat B'\in\HBcrit$ such
  that $\theta(\hat B)=\theta(\hat B')$, there is no
  $\hat B''\in\Bcrit\setminus\{\hat B, \hat B'\}$ such that
  $\theta(\hat B'')=\theta(\hat B)$, and $\hat B$, $\hat B'$ are not
  aligned. Let $\{\hat B_1,\hat B_1'\},\ldots,\{\hat B_{p'},\hat B_{p'}'\}$
  be a list of all opposite pairs in $\hat G$. It is easy to see that
  $p=p'$. 
  We let
  \[
    \hat K^+\coloneqq\hat K\cup\bigcup_{i=1}^p(\hat B_i\cup \hat B_i').
  \]
    It is easy to construct $\LW[s+3]$-formulae
    $\formel{K-plus-vert}(x,x',x'',y)$ and
    $\formel{K-plus-edge}(x,x',x'',y_1,y_2)$ such that
    $V(K^+)=\formel{K-plus-vert}[G, u,u',u'',y]$,
    $E(K^+)=\formel{K-plus-edge}[G, u,u',u'',y_1,y_2]$, $V(\hat
    K^+)=\formel{K-plus-vert}[\hat G, \hat u,\hat u',\hat u'',y]$, and
    $E(\hat K^+)=\formel{K-plus-edge}[\hat G, \hat u,\hat u',\hat
    u'',y_1,y_2]$.

  Let us call a $J$-bridge $B$ \emph{super-critical} if it is critical,
  but not contained in an opposite pair. Let $\Bsc$ be the set of all
  super-critical $J$-bridges (see Figure \ref{fig:critical}). Similarly, we call a $\hat J$-bridge
  $\hat B$ \emph{super-critical} if it is critical, but not contained
  in an opposite pair, and we let $\HBsc$ be the set of all
  super-critical $\hat J$-bridges.

  Observe that the bijection $\beta$ between $\Bcrit$ and
  $\HBcrit$ defined above induces a bijection between $\Bsc$ and
  $\HBsc$. Moreover, we have $G=K^+\cup\bigcup_{B\in\Bsc}B$, and this implies
  $\hat G=\hat K^+\cup\bigcup_{\hat B\in\HBsc}\hat B$. 
  
  Next, we expand $K^+$ and $\hat K^+$ by new colours that encode the
  information about which bridges are attached to which vertices. 
  For every $v\in V(K^+)$, we let
  \[
    \Theta(v)\coloneqq\llcurly \theta(B)\mid B\in\Bcrit, v\in\at(B)\rrcurly
  \]
  Moreover,
  we let 
  \[
    \Phi(v)\coloneqq
    \begin{cases}
      x&\text{if }v=u,\\
      x'&\text{if }v=u',\\
      x''&\text{if }v=u'',\\
      i&\text{if }v\in V(\mathring H_i)\setminus\{u''\}\text{ for some }i\in\ell,\\
      \bot&\text{if }v\in V(K)\setminus V(J).
    \end{cases}
  \]
  We view $\Theta(v)$ and $\Phi(v)$ as additional colours of the
  vertices of $K^+$ and in the following view $K^+$ as a coloured
  graph where these new colours are incorporated.  We note that for
  every colour $c\in\operatorname{rg}(\Phi)$ we have a $\LW[s+3]$-formula
  $\formel{Phi}_c(x,x',x'',y)$ such that
  $G\models\formel{Phi}_c(u,u',u'',v)\iff\Phi(v)=c$. Similarly,
  using the formulae $\formel{bridge-iso}_A$ defined above, for
  every colour $c\in\operatorname{rg}(\Theta)$ we can construct a $\LW[s+3]$-formula
  $\formel{Theta}_c(x,x',x'',y)$ such that
  $G\models\formel{Theta}_c(u,u',u'',v)\iff\Theta(v)=c$.

  We can use these formulae to transfer the colouring to the graph
  $\hat K^+$: for $\hat v\in V(\hat K^+)$, we let $\hat\Phi(\hat v)$ be the
  unique $c\in\operatorname{rg}(\Phi)$ such that
  $\hat G\models\formel{Phi}_c(\hat u,\hat u',\hat u'',\hat v)$. If
  there is more than one or no such $c$, then the graphs can be
  distinguished by a $\LW[s+3]$-formula.  Similarly, for
  $\hat v\in V(\hat K^+)$, we let $\hat\Theta(\hat v)$ be the unique
  $c\in\operatorname{rg}(\Theta)$ such that
  $\hat G\models\formel{Theta}_c(\hat u,\hat u',\hat u'',\hat v)$.
  In the following, we regard $K^+$ and $\hat K^+$ as coloured graphs with
  these colours, in addition to the colours inherited from $G$ and $\hat G$.

  \begin{claim}[resume]
    \[
      K^+\cong\hat K^+.
    \]

    \proof
    The graphs $K^+$ and $\hat K^+$ with all colours are
    definable in $G$ by $\LW[s+3]$-formulae using the three parameters
    $u,u',u''$. Moreover, $K^+$ is a
    planar graph, and thus there is a $\LW[4]$-formula that
    identifies it. From this formula and the formulae defining membership in the subgraphs $K^+$ and
    $\hat K^+$ we can construct a $\LW[s+3]$-formula that would
    distinguish $G$ and $\hat G$ if $K^+$ and $\hat K^+$ were
    non-isomorphic.
    \uend
  \end{claim}

  In the following, we let $\pi$ be an isomorphism from $K^+$ to
  $\hat K^+$. It is our goal to extend $\pi$ to an isomorphism from
  $G$ to $\hat G$. For this, we need to extend $\pi$ to all
  super-critical bridges. We process the bridges by type. So let
  $\theta$ be a type. Let $\CB_\theta$ be the set of all $B\in\Bsc$
  with $\theta(B)=\theta$, and similarly, let $\hat\CB_\theta$ be the
  set of all $\hat B\in\HBsc$ with $\hat\theta(B)=\theta$. Then the
  bijection $\beta$ between $\Bcrit$ and $\HBcrit$ defined above
  induces a bijection between $\CB_\theta$ and $\hat\CB_\theta$.

  We shall construct an extension $\pi_\theta$ of $\pi$ that is an
  isomorphism from $K^+\cup\bigcup_{B\in\CB_\theta}B$ to
  $\hat K^+\cup\bigcup_{\hat B\in\hat\CB_\theta}\hat B$. We can easily
  combine all the $\pi_\theta$ to one isomorphism from $G$ to
  $\hat G$, because they all coincide on $K^+$ and
  the intersection between any two bridges is in $K^+$
  and $\hat K^+$, respectively.

  Suppose first that all $B,B' \in\CB_\theta$ are aligned. Then for
  all $\hat B,\hat B' \in \hat \CB_\theta$ we have
  $\at(\hat B)=\at(\hat B')$, since
  $\formel{aligned} \in \LW[\max\{7,s+2\}]$. Note that $\beta$ induces
  an isomorphism from $\bigcup_{B\in\CB_\theta}B$ to
  $\bigcup_{\hat B\in\hat\CB_\theta}\hat B$. We can easily extend this
  isomorphism to an isomorphism from
  $K^+\cup\bigcup_{B\in\CB_\theta}B$ to
  $\hat K^+\cup\bigcup_{\hat B\in\hat\CB_\theta}\hat B$ because the
  attachment pattern is encoded in the colouring of the bridges.

  Suppose next that there are $B_1,B_2\in\CB_{\theta}$ that are not
  aligned. Then $\CB_{\theta}\ge 3$, because otherwise $\{B_1,B_2\}$
  would be an opposite pair. Say,
  $\CB_{\theta}=\{B_1,B_2,\ldots,B_m\}$. Without loss of generality we
  assume that every $B_i$ is a $1$-bridge. The case that every $B_i$ is an $\ell$-bridge or every $B_i$ is an $(\ell,1)$-bridge can be
  dealt with in the same way. By Item \ref{it:1:cbridge} of Claim~\ref{claim:cbridge},
  one of $B_1$ and $B_2$, say $B_1$, is embedded in $\Ff_1$. But
  then by Item \ref{it:2:cbridge} of Claim~\ref{claim:cbridge}, the bridges $B_2,\ldots,B_p$ are not
  embedded in $\FD$. By Item \ref{it:1:cbridge} again,
  $B_2,\ldots,B_p$ are aligned. For every $i\in[p]$, let $X_i\coloneqq
  \at(B_i)\cap\art(\CQ_1)$ and $Y_i\coloneqq \at(B_i)\setminus
  X_i$. Then by Item \ref{it:1:cbridge} we know that $X_1=X_2=\cdots=X_p$
  and $Y_2=\cdots=Y_p$ and $Y_1\cap Y_i=\emptyset$ for $i\ge 2$. Now
  the key observation is that the vertices in $Y_1$ have a different
  colour than the vertices in $Y_2$, because they are attached to a
  different number of bridges of type $\theta$. The isomorphism $\pi$
  maps $Y_1$ to a set $\hat Y_1$ of vertices that are attached to
  exactly one bridge of type $\theta$, and it maps $Y_2$ to a set
  $\hat Y_2$ of vertices that are attached to
  $p-1$ bridges of type $\theta$. Moreover, it maps $X \coloneqq X_1$ to a set
  $\hat X$ of vertices that are attached to $p$ bridges of type
  $\theta$. We can now extend the isomorphism $\pi$ by mapping $B_1$
  to the unique bridge of type $\theta$ that is attached to the
  vertices in $\hat Y_1$ and by mapping $B_2,\ldots,B_p$ to the $p-1$
  bridges of type $\theta$ that are attached to $\hat Y_2$.
\end{proof}

This completes the proof of Lemma \ref{lem:main} and thus also the proof of Theorem \ref{thm:main}. We finally prove the bound $2g+3$ if the surface $\FS$ that $G$ is embedded into is orientable.

\begin{proof}[Proof of Corollary~\ref{cor:orient}]
  The Euler genus of an orientable surface is always even. Suppose $G$ is a graph embeddable in an orientable surface of Euler genus $g$. Since the subgraphs obtained by cutting through the beads are also embeddable in orientable surfaces of smaller Euler genus, their Euler genus is at least $2$ smaller than the Euler genus of $G$. Therefore, inductively proceeding as described in the previous section, redefining $s$ to be the number of variables needed for graphs embeddable in orientable surfaces of Euler genus at most $g-2$, we can improve our bound from Theorem \ref{thm:main} to $2g + 3$.
\end{proof}

\section{Concluding Remarks}

The WL dimension is a measure for the combinatorial and
descriptive complexity of a graph. In view of its numerous, seemingly
unrelated characterisations in terms of logic, algebra, mathematical
programming, and homomorphisms, we can arguably regard the WL
dimension as a 
natural and robust graph invariant.

We have proved an upper bound of $4g+3$ for the WL dimension of graphs
of Euler genus $g$ and showed that if $G$ is known to be embeddable on an orientable surface of Euler genus $g$, the bound improves to $2g+3$. The immediate question that remains is how tight our
bound is. 

We believe that by refining our arguments in some
places it might be possible to reduce the bound from Theorem \ref{thm:main} to $3g+3$ or even $2g+3$; any
further improvement seems to require substantial additional ideas. It is conceivable that the WL dimension of planar graphs is
$2$. If this is the case, the additive term in our bound would
automatically drop to $2$. 

In terms of lower bounds, using the so-called CFI construction
\cite{caifurimm92} it is easy to prove a linear lower bound of $\epsilon\cdot g$ for the WL
dimension of graphs of Euler genus $g$, albeit with a rather small
constant $\epsilon>0$. To close the gap between upper and lower bound, it may
be worthwhile to spend some effort on improving the lower bound.

Beyond graphs of bounded genus, we can try to determine the WL
dimension of other graph classes and tie the WL dimension to other graph
invariants. A natural target would be the class of all graphs that
exclude the complete graph $K_\ell$ as a minor. We know that the WL
dimension of this class is bounded \cite{gro17}. But even an exponential
bound of the WL dimension in terms of $\ell$ would be major progress.

\bibliography{main}
\end{document}